\documentclass[a4paper,11pt]{article} 

\usepackage{amsmath, amssymb, amsthm, mathtools}
\usepackage[mathscr]{euscript}       
\usepackage{dsfont} 
\usepackage{graphicx} 
\usepackage{xspace}
\usepackage{qip} 
\usepackage[section]{thmenvironments} 
\usepackage{layoutcommands} 

\usepackage{underscore} 
\usepackage[shortcuts]{extdash} 
\usepackage{authblk} 
\usepackage{appendix} 
\usepackage{microtype} 
\usepackage{enumitem} 
\usepackage[margin=10pt,font=small,labelfont=bf,labelsep=endash,hypcap=true]{caption}
\usepackage[margin=10pt,font=small,labelfont=bf,labelformat=parens,labelsep=space,hypcap=false,subrefformat=parens]{subcaption}
\usepackage{tikz} 
\usetikzlibrary{arrows,fit,positioning,shapes} 
\usepackage[toc,lof]{multitoc}
\usepackage[linktocpage=true,colorlinks=true]{hyperref}
\usepackage[open,numbered]{bookmark} 
\usepackage{breakurl} 
\usepackage{cite} 
\usepackage{hyperlinks} 
\usepackage{eprint} 

\hypersetup{
  pdftitle={Causal Boxes: Quantum Information-Processing Systems Closed under Composition},
  pdfauthor={Portmann, Matt, Maurer, Renner, Tackmann},
  pdfsubject={Quantum information theory},
  pdfstartview={FitH},
  pdfpagelayout={OneColumn}
}

\title{Causal Boxes: Quantum Information\-/Processing Systems
  Closed under Composition}

\author[1]{Christopher~Portmann\email{chportma@ethz.ch}}
\author[2]{Christian~Matt\email{mattc@inf.ethz.ch}}
\author[2]{Ueli~Maurer\email{maurer@inf.ethz.ch}}
\author[1]{Renato~Renner\email{renner@ethz.ch}}
\author[3]{Bj\"orn~Tackmann\email{bta@zurich.ibm.com}}

\affil[1]{Institute for Theoretical Physics, ETH Zurich, 8093 Z\"urich, Switzerland.}
\affil[2]{Department of Computer Science, ETH Zurich, 8092 Z\"urich, Switzerland.}
\affil[3]{IBM Research -- Zurich, 8803 R\"uschlikon, Switzerland.}

\date{\today}

\newcommand{\cj}{Choi\-/Ja\-mi\-o\l{}\-kow\-ski\xspace}
\newcommand{\cjrep}{Choi\-/Ja\-mi\-o\l{}\-kow\-ski representation\xspace}
\newcommand{\cjreps}{Choi\-/Ja\-mi\-o\l{}\-kow\-ski representations\xspace}
\newcommand{\strep}{Stine\-spring representation\xspace}
\newcommand{\streps}{Stine\-spring representations\xspace}
\newcommand{\natrep}{natural representation\xspace}

\newcommand{\vacuum}{\ket{\Omega}}
\newcommand{\T}{\cT}
\newcommand{\cut}{\fC(\T)}
\newcommand{\bcut}{\overline{\fC}(\T)}
\newcommand{\completion}{\overline{\fZ}(\T)}
\newcommand{\LtwoOp}{\ell^2}
\newcommand{\Ltwo}[2]{#2 \tensor \LtwoOp\brs{#1}}
\newcommand{\LtwoTC}{\Ltwo{\T}{\C^d}}
\newcommand{\fock}[2][@]{\mathcal{F}\ifthenelse{\equal{#1}{@}}{}{_{#1}}\brs{#2}}
\newcommand{\tcop}[1]{\mathfrak{T}\brs{#1}}
\newcommand{\bdop}[1]{\mathfrak{B}\brs{#1}}
\newcommand{\ctrl}{\mathtt{ctrl}\text{-}}
\newcommand{\QS}{\mathtt{QS}}
\newcommand{\QU}{\mathtt{U}}
\newcommand{\QV}{\mathtt{V}}
\newcommand{\inrm}{\mathrm{in}}

\newcommand{\C}{\complex}
\newcommand{\R}{\reals}
\newcommand{\Q}{\rationals}
\newcommand{\N}{\naturals}
\newcommand{\ports}{\mathsf{ports}}
\newcommand{\connect}[1]{\xleftrightarrow{#1}}

\begin{document}

\pdfbookmark[1]{Title page}{titlepage}

\maketitle

\begin{abstract}
  Complex information\-/processing systems, for example quantum
  circuits, cryptographic protocols, or multi\-/player games, are
  naturally described as networks composed of more basic
  information\-/processing systems. A modular analysis of such systems
  requires a mathematical model of systems that is closed under
  composition, i.e., a network of these objects is again an object of
  the same type. We propose such a model and call the corresponding
  systems \emph{causal boxes}.
 
  Causal boxes capture superpositions of causal structures, e.g.,
  messages sent by a causal box $A$ can be in a superposition of
  different orders or in a superposition of being sent to box $B$ and
  box $C$. Furthermore, causal boxes can model systems whose behavior
  depends on time. By instantiating the Abstract Cryptography
  framework with causal boxes, we obtain the first composable security
  framework that can handle arbitrary quantum protocols and
  relativistic protocols.
\end{abstract}

\clearpage
\phantomsection
\pdfbookmark[1]{\contentsname}{sec:toc}
\tableofcontents
\clearpage

\section{Introduction}
\label{sec:intro}

In this work we are concerned with modeling quantum
information\-/processing systems, i.e., interactive systems that
receive and send quantum messages. Similar formalisms for modeling
such systems were developed by Gutoski and Watrous~\cite{GW07,Gut12},
Chiribella, D'Ariano and Perinotti~\cite{CDP09}, and
Hardy~\cite{Har11,Har12,Har15} (see also Hardy~\cite{Har05,Har07}), to
which we refer in the following using the term from \cite{CDP09},
namely \emph{quantum combs}. Quantum combs are a generalization of
\emph{random systems}~\cite{Mau02,MPR07} to quantum information
theory. A comb is a system with internal memory that processes
messages one at a time as they are received. An example is depicted in
\figref{fig:comb}, where the dashed lines capture the comb structure.

\begin{figure}[htb]
\begin{centering}
\begin{tikzpicture}[scale=1.2,
wire/.style={->,>=stealth,thick},
circle1/.style={draw,circle,fill,thick,inner sep=1.5pt}]

\node[circle1] (u1) at (0,0) {};
\node[circle1] (u2) at (2,0) {};
\node[circle1] (u3) at (4,0) {};
\node[circle1] (u4) at (6,0) {};

\node (v11) at (-.75,-1.1) {};
\node (v12) at (.75,-1.1) {};
\node (v21) at (1.25,-1.1) {};
\node (v22) at (2.75,-1.1) {};
\node (v31) at (3.25,-1.1) {};
\node (v32) at (4.75,-1.1) {};
\node (v41) at (5.25,-1.1) {};
\node (v42) at (6.75,-1.1) {};

\draw[wire] (u1) to (u2);
\draw[wire] (u2) to (u3);
\draw[wire] (u3) to (u4);

\draw[wire] (v11) to (u1);
\draw[wire] (v21) to (u2);
\draw[wire] (v31) to (u3);
\draw[wire] (v41) to (u4);
\draw[wire] (u1) to (v12);
\draw[wire] (u2) to (v22);
\draw[wire] (u3) to (v32);
\draw[wire] (u4) to (v42);

\draw[dashed,thick] (-.35,-1.2) to ++(0,1.6) to ++(6.7,0) to
++(0,-1.6) to ++(-.7,0) to ++(0,.8) to ++(-1.3,0) to ++(0,-.8) to
++(-.7,0) to ++(0,.8) to ++(-1.3,0) to ++(0,-.8) to ++(-.7,0) to
++(0,.8) to ++(-1.3,0) to ++(0,-.8) -- cycle;
\end{tikzpicture}

\end{centering}
\caption[Quantum comb]{\label{fig:comb}A single
  information\-/processing system modeled as a comb. The nodes
  represent an operation and the arrows capture a quantum state. Each
  tooth of the comb corresponds to a pair of an input and an output
  message.}
\end{figure}
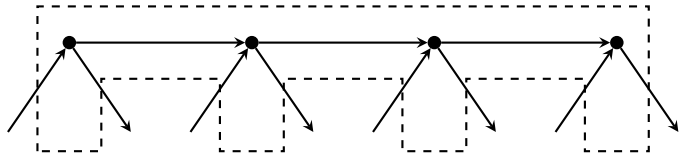

The quantum comb framework~\cite{GW07,Gut12,CDP09,Har11,Har12,Har15}
provides rules for representing these objects independently from their
internal state and composing them when the order of messages is
predefined. It also defines a notion of distance between combs. This
has found applications in modeling two player games~\cite{GW07} and
two player cryptographic protocols~\cite{DFPR14}.

\subsection{Ordering messages}
\label{sec:intro.order}

The composition of systems described as combs is however not always
well\-/defined. Consider the example drawn in \figref{fig:ndsystem}:
two players, Alice and Bob, each send a message to a third player,
Charlie, who outputs the first message he receives and ignores the
second. Each of the systems is a well\-/defined comb. Alice and Bob
just output a single message. When Charlie receives the first message,
$m = (v,p)$ \--- value $v$ from player $p$ \--- he outputs $v$ and
ignores all further inputs. But the composition of all three systems
(depicted as a dashed box in \figref{fig:ndsystem}) is not defined: it
is a system with no input and one output, but this output is
undetermined.

\begin{figure}[htb]
\begin{centering}

\begin{tikzpicture}
  \node[draw,thick,minimum height=1cm,minimum width=1.5cm] (A) at (0,1) {Alice};
  \node[draw,thick,minimum height=1cm,minimum width=1.5cm] (B) at (0,-1) {Bob};
  \node[draw,thick,minimum height=1cm,minimum width=1.5cm] (C) at (3,0)  {Charlie};
  \draw[->,>=stealth,thick] (A) to node[auto,sloped,pos=.02] {\tt Alice!} (C);
  \draw[->,>=stealth,thick] (B) to node[auto,sloped,pos=.85] {\tt Bob!} (C);
  \draw[->,>=stealth,thick] (C) to node[auto] {\tt Alice}
  node[auto,swap] {or {\tt Bob}?} (5.4,0);
   \node[fit=(A)(B)(C),dashed,draw,thick] {};
\end{tikzpicture}

\end{centering}
\caption[Order dependent system]{\label{fig:ndsystem}Alice and Bob
  both send messages to Charlie, who outputs the first message he
  receives. Although each system can be described by a comb, the
  composition of the three, depicted as the dashed box, is undefined.}
\end{figure}
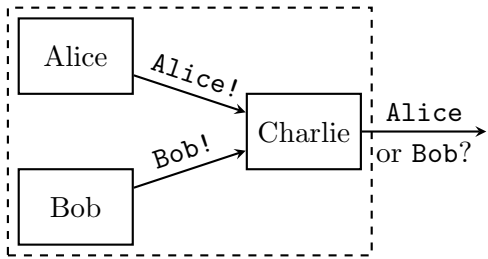

The composition of these three systems is undefined, because Alice's
and Bob's messages are unordered, yet the output of Charlie depends on
this order. However, if one considers physical systems, e.g., an
implementation of \figref{fig:ndsystem}, the composition is a new
well\-/defined physical system. Messages are output at a certain time,
which results in a well\-/defined order. This ordering information was
ignored in the descriptions of the systems given above, and results in
the ill\-/defined composition. Combs are well\-/suited for analysing
systems that have a predefined causal structure, e.g., ordered
networks~\cite{CDP09}, some simple two player games~\cite{GW07} and
two player cryptographic protocols~\cite{DFPR14}. But Chiribella et
al.~\cite{CDPV13} prove that their framework is ill\-/suited for
modeling settings where the causal structure is not predefined, e.g.,
when it is determined by an input or a coin toss.\footnote{To solve
  this, Chiribella et al.~\cite{CDPV13} propose a non\-/causal model
  of quantum information\-/processing systems, which we discuss
  briefly in \secref{sec:intro.related}.}

\subsection{Superpositions of orders}
\label{sec:intro.superposition}

In a quantum framework it is not sufficient to have all messages
(dynamically) ordered, it must also be possible to have messages in
superpositions of different orders. For example, a player might choose
to send a message $\ket{\psi}$, she might choose to send ``nothing''
\--- which we denote by a vacuum state $\vacuum$ \--- or she might
send a superposition of the two, i.e., she prepares and sends the
state
\begin{equation} \label{eq:superposition.1}
  \alpha\ket{\psi}+\beta\vacuum \,.\end{equation} A message could also be
in a superposition of sent to Alice and sent to Bob, i.e.,
\begin{equation} \label{eq:superposition.2} \alpha\ket{\psi}_A \tensor
  \vacuum_B + \beta \vacuum_A \tensor \ket{\psi}_B \,.\end{equation}
This results in a player receiving superpositions of different numbers
of messages in different orders. Since the value of a message can only
be influenced by a message that is ordered before, a superposition of
different orders of messages corresponds to a superposition of the
causal structure between these messages.

Such superpositions are not only a possibility offered by quantum
mechanics, but are also necessary to perform certain quantum
information\-/processing tasks. Consider the problem of designing a
circuit that performs a controlled unitary for an unknown $U$ given
only a single black\-/box access to $U$, which we illustrate in
\figref{fig:ctrlU1}. It has been proven in \cite{TGMV13,AFCB14} that
this task is impossible. And yet, adding control to an unknown unitary
can easily be implemented in practice, and has been done in
\cite{ZRKZPLO11}. As pointed out in \cite{AFCB14,FDDB14}, the
discrepancy between the two results stems from the assumption made in
the impossibility proofs \cite{TGMV13,AFCB14} that a wire in a circuit
models a (non\-/vacuum) state, whereas in a physical model one has the
freedom of not inputing anything to the black box \--- i.e., one does
not use the input wire \--- in which case no output is produced
either. The solution from \cite{ZRKZPLO11} consists in sending a
vacuum state on one of the wires and performing a controlled switch
between the vacuum state and the input. This is illustrated in
\figref{fig:ctrlU2}, where the wires are depicted as arrows to
emphasize that they have a different meaning from the wires in
\figref{fig:ctrlU1}. We refer to \cite{AFCB14,FDDB14} for a further
discussion of this.

\begin{figure}[tb]
\begin{subfigure}{\textwidth}
\begin{centering}
\hfill
\begin{tikzpicture}[
wire/.style={-,thick},
unitary2/.style={draw,thick,minimum width=.7cm,minimum height=2.1cm},
unitary1/.style={draw,thick,minimum width=.7cm,minimum height=.7cm}]

\node (l1) at (1.5,0) {};
\node[above right,yshift=-1] (zero) at (l1) {$\zero$};
\node (l2) at (1.5,-.7) {};
\node (l3) at (1.5,-1.4) {};
\node (r1) at (6.3,0) {};
\node (r2) at (6.3,-.7) {};
\node (r3) at (6.3,-1.4) {};

\node[unitary2] (u1) at (2.7,-.7) {$A$};
\node[unitary1] (u2) at (3.9,-1.4) {$U$};
\node[unitary2] (u3) at (5.1,-.7) {$B$};

\draw[wire] (l1) to (u1.west |- l1);
\draw[wire] (l2) to (u1.west |- l2);
\draw[wire] (l3) to (u1.west |- l3);
\draw[wire] (u1.east |- r1) to (u3.west |- r1);
\draw[wire] (u1.east |- r2) to (u3.west |- r2);
\draw[wire] (u1.east |- r3) to (u2.west |- r3);
\draw[wire] (u2.east |- r3) to (u3.west |- r3);
\draw[wire] (u3.east |- r1) to (r1);
\draw[wire] (u3.east |- r2) to (r2);
\draw[wire] (u3.east |- r3) to (r3);

\node[fit=(zero)(u3)(l3)(r3)] (all) {};
\node[above right] at (all.north west) {Construction:};
\end{tikzpicture}
\hfill
\begin{tikzpicture}[
wire/.style={-,thick},
circle1/.style={draw,circle,fill,thick,inner sep=1.5pt},
unitary1/.style={draw,thick,minimum width=.7cm,minimum height=.7cm}]

\node (l1) at (2.7,0) {};
\node[above right,yshift=-1] (zero) at (l1) {$\zero$};
\node (l2) at (2.7,-.7) {};
\node (l3) at (2.7,-1.4) {};
\node (r1) at (5.1,0) {};
\node (r2) at (5.1,-.7) {};
\node (r3) at (5.1,-1.4) {};

\node[unitary1] (u1) at (3.9,0) {$W_U$};
\node[circle1] (u2) at (3.9,-.7) {};
\node[unitary1] (u3) at (3.9,-1.4) {$U$};

\draw[wire] (l1) to (u1);
\draw[wire] (l2) to (u2);
\draw[wire] (l3) to (u3);
\draw[wire] (u1) to (r1);
\draw[wire] (u2) to (r2);
\draw[wire] (u3) to (r3);
\draw[wire] (u2) to (u3);

\node[fit=(zero)(u3)(l3)(r3)] (all) {};
\node[above right] at (all.north west) {Goal:};
\end{tikzpicture}
\hfill
\end{centering}
\caption[Impossibility]{\label{fig:ctrlU1}Impossibility of a $\ctrl U$: there do not
  exist any operators $A$ and $B$ such that for all $U$ the circuit on
  the left is equivalent to the circuit on the
  right~\cite{TGMV13,AFCB14}.}
\end{subfigure}

\vspace{9pt}

\begin{subfigure}{\textwidth}
\begin{centering}
\hfill
\begin{tikzpicture}[
wire/.style={-,thick},
wArrow/.style={->,>=stealth,thick},
circle1/.style={draw,circle,fill,thick,inner sep=1.5pt},
cross1/.style={inner sep=2pt},
unitary1/.style={draw,thick,minimum width=.7cm,minimum height=.7cm}]

\node (l1) at (1.5,0) {};
\node[above right,yshift=-1] (zero) at (l1) {$\vacuum$};
\node (l2) at (1.5,-.7) {};
\node (l3) at (1.5,-1.4) {};
\node (r1) at (6.3,0) {};
\node (r2) at (6.3,-.7) {};
\node (r3) at (6.3,-1.4) {};

\node[circle1] (c1) at (2.7,-.7) {};
\node[cross1] (x1) at (2.7,0) {};
\node[cross1] (x2) at (2.7,-1.4) {};
\node[unitary1] (u2) at (3.9,-1.4) {$U$};
\node[circle1] (c2) at (5.1,-.7) {};
\node[cross1] (x3) at (5.1,0) {};
\node[cross1] (x4) at (5.1,-1.4) {};

\draw[wArrow] (l1) to (r1);
\draw[wArrow] (l2) to (r2);
\draw[wArrow] (l3) to (u2.west |- r3);
\draw[wArrow] (u2.east |- r3) to (r3);

\draw[wire] (c1) to (x1.center);
\draw[wire] (c1) to (x2.center);
\draw[wire] (c2) to (x3.center);
\draw[wire] (c2) to (x4.center);

\draw[wire] (x1.north west) to (x1.south east);
\draw[wire] (x1.north east) to (x1.south west);
\draw[wire] (x3.north west) to (x3.south east);
\draw[wire] (x3.north east) to (x3.south west);
\draw[wire] (x2.north west) to (x2.south east);
\draw[wire] (x2.north east) to (x2.south west);
\draw[wire] (x4.north west) to (x4.south east);
\draw[wire] (x4.north east) to (x4.south west);

\node at (7,-.7) {\huge $=$};
\end{tikzpicture}
\begin{tikzpicture}[
wire/.style={-,thick},
wArrow/.style={->,>=stealth,thick},
circle1/.style={draw,circle,fill,thick,inner sep=1.5pt},
unitary1/.style={draw,thick,minimum width=.7cm,minimum height=.7cm}]

\node (l1) at (2.7,0) {};
\node[above right,yshift=-1] (zero) at (l1) {$\vacuum$};
\node (l2) at (2.7,-.7) {};
\node (l3) at (2.7,-1.4) {};
\node (r1) at (5.1,0) {};
\node (r2) at (5.1,-.7) {};
\node (r3) at (5.1,-1.4) {};

\node[circle1] (u2) at (3.9,-.7) {};
\node[unitary1] (u3) at (3.9,-1.4) {$U$};

\draw[wArrow] (l1) to (r1);
\draw[wArrow] (l2) to (r2);
\draw[wArrow] (l3) to (u3);
\draw[wArrow] (u3) to (r3);
\draw[wire] (u2) to (u3);
\end{tikzpicture}
\hfill
\end{centering}
\caption[Implementation]{\label{fig:ctrlU2}Implementation of a $\ctrl U$: the circuit
  on the left performs a controlled switch between the top and bottom
  wires before and after the unitary $U$, which is equivalent to the
  controlled unitary depicted on the
  right~\cite{ZRKZPLO11,AFCB14,FDDB14}. When nothing \--- i.e., the
  vacuum state $\vacuum$ \--- is input to $U$, nothing is output
  either.}
\end{subfigure}

\caption[Controlled unitary]{\label{fig:ctrlU}(Im)Possibility of implementing a $\ctrl U$.
  In the circuits in \subref{fig:ctrlU1} a wire \--- depicted as a
  line \--- corresponds to the space on which $U$ acts. In the
  circuits in \subref{fig:ctrlU2} a wire \--- depicted as an arrow
  \--- can also transmit a vacuum state, $\vacuum$.}
\end{figure}

Another example is given by the \emph{quantum switch}~\cite{CDPV13}:
using only a single black\-/box access to each of two systems
implementing some unknown unitaries $U$ and $V$, the quantum switch
applies these gates in (a controlled) superposition of the two
possible orders, i.e., on input
\begin{equation} \label{eq:qswitch.in} \left( \alpha \zero + \beta
    \one \right) \tensor \ket{\varphi}\,, \end{equation}
it outputs
\begin{equation} \label{eq:qswitch.out} \alpha \zero \tensor UV
  \ket{\varphi} + \beta \one \tensor VU \ket{\varphi}\,,\end{equation}
as illustrated in \figref{fig:qswitch}. This particular system \---
which can be physically implemented~\cite{FDDB14,ACB14,PMACADHRBW15},
but cannot be expressed using a classical ordering of
gates~\cite{CDPV13} \--- has been proven to be useful in reducing the
complexity of computational tasks~\cite{CDFP12,ACB14} and increasing
the success probability of discriminating between quantum states and
channels~\cite{Chi12,FAB15}. As an illustration of our framework we
show in \secref{sec:qswitch} how to describe the quantum switch using
the language of causal boxes that we develop in this work.

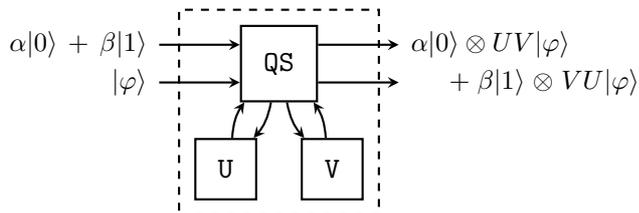
\begin{figure}[tb]
\begin{centering}
\begin{tikzpicture}[
wire/.style={-,thick},
wArrow/.style={->,>=stealth,thick},
unitary2/.style={draw,thick,minimum width=1cm,minimum height=1cm},
unitary1/.style={draw,thick,minimum width=.8cm,minimum height=.8cm}]

\node[inner sep=0] (l1) at (-1.6,.25) {};
\node[left,text width=3cm,align=right] at (l1) {\small $\alpha \zero + \beta \one$};
\node[inner sep=0] (l2) at (-1.6,-.25) {};
\node[left] at (l2) {\small $\ket{\varphi}$};
\node[inner sep=0] (r1) at (1.6,.25) {};
\node[right,text width=3cm,align=left] at (r1)
    {\small $\alpha \zero \tensor UV \ket{\varphi}$};
\node[inner sep=0] (r2) at (1.6,-.25) {};
\node[right,text width=3cm,align=right] at (r2)
    {\small ${}+\beta \one \tensor VU \ket{\varphi}$};

\node[unitary2] (qs) at (0,0) {$\QS$};
\node[unitary1] (u) at (-.7,-1.4) {$\QU$};
\node[unitary1] (v) at (.7,-1.4) {$\QV$};

\node[fit=(qs)(u)(v),thick,dashed,draw,inner sep=.2cm] (box) {};

\draw[wArrow] (l1) to (qs.west |- l1);
\draw[wArrow] (l2) to (qs.west |- l2);
\draw[wArrow,bend left=15] (qs) to (u);
\draw[wArrow,bend right=15] (qs) to (v);
\draw[wArrow,bend left=15] (u) to (qs);
\draw[wArrow,bend right=15] (v) to (qs);
\draw[wArrow] (qs.east |- r1) to (r1);
\draw[wArrow] (qs.east |- r2) to (r2);
\end{tikzpicture}

\end{centering}
\caption[Quantum switch]{\label{fig:qswitch}The quantum switch, $\QS$, queries first
  $\QU$ then $\QV$ or first $\QV$ then $\QU$ depending on the control qubit.}
\end{figure}

It is thus essential that a framework for modeling quantum information
processes can capture superpositions of different numbers of messages
arriving in superpositions of different orders. To the best of our
knowledge, no existing causal framework has such a feature. Physical
systems such as those from \figref{fig:ctrlU2} and
\figref{fig:qswitch}, and constructions from
\cite{ZRKZPLO11,AFCB14,FDDB14,Chi12,FAB15,CDFP12,ACB14,PMACADHRBW15}
cannot be modeled in existing frameworks. Furthermore, impossibility
results proven for combs~\cite{CDPV13} or a restricted circuit
model~\cite{TGMV13,AFCB14} do not hold anymore when we consider the
larger class of systems allowed by quantum mechanics. Similarly,
existing quantum composable security frameworks~\cite{Unr04,Unr10} can
only make security statements for a limited class of quantum
information\-/processing systems. They include an artificial classical
scheduler that systematically measures states such as
\eqref{eq:superposition.1} and \eqref{eq:superposition.2} to determine
who is the recipient of the next message, and can thus only model a
classical ordering of messages.\footnote{The composable security
  framework from \cite{BM04} has an additional restriction on the
  quantum systems that can be modeled in the framework: they must
  correspond to a set of ordered gates, which is essentially
  equivalent to an ordered network captured by a quantum
  comb~\cite{GW07,Gut12,CDP09,Har11,Har12,Har15}.}

\subsection{Causal boxes}
\label{sec:intro.boxes}

The main contribution of this work is to introduce \emph{causal boxes}
as an abstract model for discrete quantum information\-/processing
systems. Causal boxes include quantum
combs~\cite{GW07,Gut12,CDP09,Har11,Har12,Har15} and the models from
composable security frameworks~\cite{BM04,Unr04,Unr10} as special
cases. Crucially, they have the property of being closed under
composition. Naturally, composing two physical systems results in a
new physical system. The main challenge when modeling an abstraction
of physical systems is to identify the minimal characterization that
preserves this closure property without excluding any
systems. 

In a quantum comb~\cite{GW07,Gut12,CDP09,Har11,Har12,Har15}, the teeth
of the comb capture the predefined causal structure of the system. Any
input on a wire at one of the teeth can influence only outputs on
teeth that are further to the right in \figref{fig:comb}. Since we are
interested in systems which do not have such a predefined order,
causal boxes have a flat (unordered) interface, depicted in
\figref{fig:box1}. Instead, the messages themselves are ordered, each
one consisting of a pair $(v,t)$ of a value $v$ and a global order
$t \in \T$ from some partially ordered set $\T$. The parameter $t$
captures the order of the message with respect to other messages,
i.e., its position within $\T$.  A causal box processes inputs
according to this order: an output in position $t$ can only depend on
inputs arriving strictly before $t$ \--- that is, the box must satisfy
\emph{causality}. Having all messages assigned some order $t \in \T$
allows protocols that involve time to be naturally modeled by
interpreting $t$ as the time at which a message is sent or received. A
partial order on $\T$ may also be used to model space\-/time in a
relativistic setting. These applications are discussed further in
\secsref{sec:intro.applications} and \ref{sec:conclusion}.

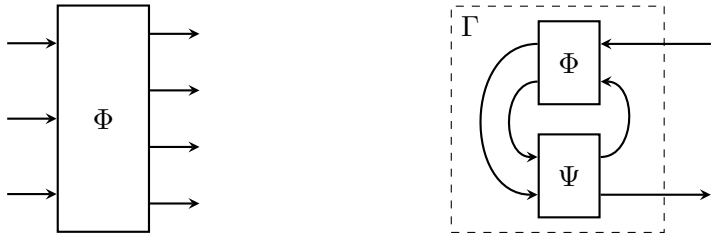
\begin{figure}[tb]
  \subcaptionbox[A single box]{\label{fig:box1}An abstract depiction of a causal
    box with a flat interface. Messages may arrive in any order \---
    or superpositions thereof \--- on the wires.}[.5\textwidth][c]{
\begin{tikzpicture}[
wArrow/.style={->,>=stealth,thick},
unitary2/.style={draw,thick,minimum width=1.2cm,minimum height=3cm}]

\node (l1) at (-1.4,1) {};
\node (l2) at (-1.4,0) {};
\node (l3) at (-1.4,-1) {};
\node (r1) at (1.4,1.125) {};
\node (r2) at (1.4,.375) {};
\node (r3) at (1.4,-.375) {};
\node (r4) at (1.4,-1.125) {};

\node[unitary2] (u) at (0,0) {$\Phi$};

\draw[wArrow] (l1) to (u.west |- l1);
\draw[wArrow] (l2) to (u.west |- l2);
\draw[wArrow] (l3) to (u.west |- l3);
\draw[wArrow] (u.east |- r1) to (r1);
\draw[wArrow] (u.east |- r2) to (r2);
\draw[wArrow] (u.east |- r3) to (r3);
\draw[wArrow] (u.east |- r4) to (r4);
\end{tikzpicture}}
\subcaptionbox[Composing causal boxes]{\label{fig:box2}The two
  causal boxes $\Phi$ and $\Psi$ are connected by their wires. This
  results in a new box $\Gamma$.}[.5\textwidth][c]{
\begin{tikzpicture}[
wArrow/.style={->,>=stealth,thick},
unitary2/.style={draw,thick,minimum width=.8cm,minimum height=1.1cm}]

\node (l1) at (-1.4,1) {};
\node (l2) at (-.9,0.5) {};
\node (l3) at (-.9,-0.5) {};
\node (l4) at (-1.4,-1) {};
\node (r1) at (2,1) {};
\node (r2) at (.9,0.5) {};
\node (r3) at (.9,-0.5) {};
\node (r4) at (2,-1) {};

\node[unitary2] (u1) at (0,.75) {$\Phi$};
\node[unitary2] (u2) at (0,-.75) {$\Psi$};
\node[draw,dashed,minimum width=2.8cm,minimum height=3cm] (u) at (-.15,0) {};
\node[below right] at (u.north west) {$\Gamma$};

\draw[wArrow] (u1.west |- l2) .. controls (l2) and (l3) .. (u1.west |- l3);
\draw[wArrow] (u1.west |- l1) .. controls (l1) and (l4) .. (u1.west |- l4);
\draw[wArrow] (u2.east |- r3) .. controls (r3) and (r2) .. (u1.east |- r2);
\draw[wArrow] (r1) to (u1.east |- r1);
\draw[wArrow] (u2.east |- r4) to (r4);
\end{tikzpicture}}
\caption[Causal boxes]{\label{fig:box}The input wires of causal boxes
  are unordered. Instead, the messages traveling on the wires are
  ordered and processed according to this order, regardless of the
  wire. An output generated in position $t$ can depend on any input
  that arrived strictly before $t$. Cycles are permitted in a network
  of boxes, since a connection between two systems can capture a
  physical link between an output and input port.}
\end{figure}

We provide a characterization of causal boxes which is independent of
their internal state. For example, a box which accepts only classical
inputs and produces only classical outputs would have a classical
description, regardless of whether it internally performs quantum
computation. We show how to compose two boxes, i.e., how to define the
new box resulting from connecting two boxes by their wires. This is
illustrated in \figref{fig:box2}. A connection between two boxes can
be thought of as a physical connection \--- a wire literally plugs one
box into another. Unlike for ordered networks modeled by
combs~\cite{GW07,Gut12,CDP09,Har11,Har12,Har15}, we allow cycles in
networks of causal boxes. In particular, if for a single causal box
$\Phi$ an input on a wire $A$ generates an output on a wire $B$, it is
perfectly legal to put a loop from $B$ to $A$. This does not create
any causality conflicts since messages are ordered and an output can
only depend on inputs that arrived before.

A wire \--- depicted by an arrow in \figref{fig:box} \--- does not
necessarily represent a single message. It is possible to send no
message, multiple messages, or any superposition thereof on a single
wire, e.g., states such as those of \eqnsref{eq:superposition.1} and
\eqref{eq:superposition.2} can be generated by a causal box. A
message can then also be in a superposition of arriving in position
$t_1$ and position $t_2$, and a causal box can process this without
breaking the superposition.

We also define a notion of distance between causal boxes as the
probability that a distinguisher\footnote{For simplicity, a
  distinguisher can be thought of as another causal box. Though the
  exact definition is slightly more powerful, see
  \secref{sec:metric}.} can correctly guess to which of two boxes it
is connected.  Considering different classes of distinguishers results
in different distance measures. For example, computationally bounded
or unbounded distinguishers allow for different notions of security to
be captured in cryptographic applications.

\subsection{Applications}
\label{sec:intro.applications}

\paragraph{Distributed systems.}
Causal boxes allow arbitrary distributed systems to be modeled, e.g.,
where multiple messages are sent simultaneously to different
subsystems and the order of their arrival determines the outcome. The
analysis of simple two player protocols can be performed with quantum
combs~\cite{GW07,DFPR14} as they have a predefined order \--- first
player $A$ sends a message, then player $B$, then $A$, etc. But if the
protocol involves more parties or if these players access some
external ressource, the order of the messages is not predetermined and
causal boxes need to be used to model these systems.

No-go theorems which are proven using a restricted computational model
that cannot capture every quantum information\-/processing system \---
such as those of \cite{CDPV13,TGMV13,AFCB14} \--- merely prove that
the model is inadequate to accomplish the desired task. To show that a
(quantum) information\-/theoretic task is impossible, a general
framework such as the one we propose must be employed.\footnote{Since
  our framework is \emph{causal}, an impossibility proof would
  naturally only be meaningful if non\-/causal systems are unphysical;
  see the discussion in \secref{sec:conclusion}.} For example, the
authors of \cite{FDDB14} suggest that if the provider of a black\-/box
unitary wishes to enforce that its system is not used to implement a
controlled unitary,\footnote{This can be essential in a cryptographic
  context: the unitary could contain sensitive information in the
  global phase, which can be accessed if it is used as a controlled
  unitary.} this may be achieved by including a photon number counter
in the box that breaks the superpositions between the vacuum and a
photon. Causal boxes can be used to prove whether such a construction
does indeed forbid using a black\-/box unitary in a controlled way.

\paragraph{Computation with superposition of causal orders.}
Although in most distributed protocols one might expect the players to
measure whether they receive a message or not \--- and thus break any
superposition between the causal order of events \--- certain
distributed systems cannot be implemented without preserving such a
superposition. A system which transforms any unitary into a controlled
unitary \--- which was discussed in \secref{sec:intro.superposition}
\--- is one such example. Furthermore, as shown in
\cite{CDFP12,ACB14,Chi12,FAB15}, there is a computational advantage in
implementing systems that preserve such a superposition of causal
orders. In those works, the quantum switch (introduced in
\secref{sec:intro.superposition}) was used to illustrate this
advantage. It is rather straightforward to describe the effect of
plugging the quantum switch into the systems $\QU$ and $\QV$ \--- this
corresponds to the dashed box from \figref{fig:qswitch} with
input\-/output behavior given by \eqnsref{eq:qswitch.in} and
\eqref{eq:qswitch.out}. However, the quantum switch itself, i.e., the
box $\QS$ in \figref{fig:qswitch}, cannot be described as a circuit or
comb~\cite{CDPV13}. Possible implementations of the quantum switch
using linear optics have been suggested~\cite{FDDB14,ACB14} and an
experimental realization has been performed~\cite{PMACADHRBW15}, which
effectively provide an implementation dependent description of the
switch. Non\-/causal descriptions of the quantum switch have also been
proposed~\cite{CDPV13,OG16,ABCFGB15}, but as shown by the
implementations~\cite{FDDB14,ACB14,PMACADHRBW15}, the quantum switch
is a causal system, and a non\-/causal description is thus quite
unsatisfactory.  To illustrate our framework, we provide in
\secref{sec:qswitch} an abstract, but mathematical, (implementation
independent) description of the quantum switch as a causal box.


\paragraph{Cryptography.}
Traditional quantum composable security
frameworks~\cite{BM04,Unr04,Unr10} can model only a classical ordering
of messages. The composition theorems of these frameworks are de facto
limited to this restricted class of systems. The Abstract
Cryptography~(AC) framework~\cite{MR11} (see \cite{PR14} for an
introduction to quantum AC) models cryptography as a resource theory:
a protocol constructs a resource (e.g., a secure channel) from another
resource (e.g., a secret key).  The AC framework treats resources as
abstract objects, i.e., it only demands that they have certain
properties, but does not prescribe a specific model for them.  By
instantiating the resources and protocols with causal boxes, one
immediately gets a fully quantum composable security framework that
can handle superpositions of orders of messages.

To model protocols that involve time, the ordered set $\T$ may be
interpreted as the time at which a message is sent or
received. Players can then be synchronized by clocks and use timeouts:
a player aborts or behaves differently if a message has not been
received by time $t$. This has been used, for example, in
authentication protocols to reduce the secret key
consumption~\cite{Mau13,Por14}. Relativistic cryptography uses the
fact that information cannot be transmitted faster than the speed of
light to implement tasks such as bit
commitment~\cite{Ken99,Ken12,KTHW13,LKBHTWZ14,AK15} and position
verification~\cite{BCFGGOS11,Unr14}. However, due to the lack of
cryptography frameworks that can handle time, these works use ad hoc
security definitions.\footnote{A framework that models circuits that
  are located in space\-/time is proposed in~\cite{Unr14}, but it only
  defines security for position\-/based verification and
  authentication. The issue of defining a generic (composable)
  security framework in space\-/time is not adressed in that work.} A
partial order on $\T$ captures relativistic space\-/time in a natural
way. This allows relativistic protocols to be modeled in our
framework. AC instantiated with causal boxes then provides composable
security definitions that are applicable to relativistic
cryptography. This is discussed further in \secref{sec:conclusion}.

If a distributed system does not provide information about the order
in which messages are received, it can be modeled as a family of
systems, where each element in the family corresponds to one possible
behavior of this system \--- which, ultimately, determines the
ordering of the messages (this is discussed further in
\secref{sec:conclusion}). Security then has to hold for every element
in the family. Explicitly modeling the set of possible behaviors of a
system removes the artificial concept of a scheduler that is present
in other frameworks~\cite{Unr04,Unr10}. Requiring that a protocol be
secure for all behaviors of the underlying systems is equivalent to
considering the worst case scheduling and has the advantage of being
sound in a model with multiple adversaries.\footnote{The scheduler is
  typically controlled by the adversary~\cite{Unr04,Unr10}, but if
  multiple (non\-/cooperating) adversaries are present, this type of
  scheduling is ill\-/defined.}  Scenarios involving mutually
distrustful dishonest players can be captured this
way~\cite{UM10,MR11,MM13}.

\subsection{Related work}
\label{sec:intro.related}

As already mentioned in this introduction, the quantum comb
framework~\cite{GW07,Gut12,CDP09,Har11,Har12,Har15} can model any
quantum information\-/processing system in which the order of the
messages is classical and predefined, e.g., a quantum circuit and many
two\-/player games. The systems models from quantum composable
security frameworks~\cite{Unr04,Unr10} allow a dynamical (but still
classical) ordering of messages. This is essential for modeling
cryptographic protocols, where it might be decided at runtime
(depending on an input or coin flip) if a message is sent.

The need for physical models that are not restricted to classical
orderings of messages was \--- to the best of our knowledge \--- first
noticed by Hardy, who developed a framework for probability theories
with indefinite causal structures~\cite{Har05,Har07,Har09,Har10}. This
motivated Oreshkov, Costa and Brukner~\cite{OCB12,Bru14a} to propose a
process matrix formalism for modeling quantum systems with indefinite
causal structures: quantum theory is assumed to be valid in local
laboratories, but no reference is made to any global causal relations
between the operations in the different laboratories. Various aspects
of the process matrix formalism have been developed further in
subsequent works, e.g., the multi-party setting~\cite{OG16,ABCFGB15},
extensible causality~\cite{OG16}, witnesses for causal
non\-/separability~\cite{ABCFGB15,Bra16}, violations of causal
inequalities~\cite{BW14,Bru14b,BAFCB16,BW16b}, classical systems with
indefinite causal structures~\cite{BFW14,BW16a}, and non\-/causal
circuits~\cite{BW16c}. Chiribella et al.~\cite{CDPV13} develop an
equivalent framework in which systems are modeled as higher\-/order
quantum transformations. Oreshkov and Cerf~\cite{OC16} consider an
even more general framework: they introduce a model with no predefined
time, even within the individual laboratories. This encompasses
previous work as special cases.

On one hand, this process matrix formalism and its derivatives are
more general than the current work, because they do not only capture
superpositions of causal structures, but also non\-/causal structures,
e.g., where a message going from $A$ to $B$ is both (causally) before
and after a message from $B$ to $A$. However, this added generality
makes it unsuited for modeling real world applications, e.g., the
behavior of players in a cryptographic protocol, since one might be
assigning them unphysical behaviors. On the other hand, these models
are more restricted than ours, as they do not allow messages to be
dynamically ordered, which is essential for many applications. This is
discussed further in \secref{sec:conclusion}.

Concurrently to this work, the same authors developed a theory of
deterministic systems~\cite{MMPRT16}, which exploits certain
properties specific to the classical case. This allows for a framework
that is both simpler and more general than what one obtains by
restricting the current work to deterministic systems. In particular,
the notion of causality from~\cite{MMPRT16} is less restrictive and
includes systems that can execute an infinite number of causal steps
in a finite amount of time (see \secref{sec:boxes.causality} for a
discussion of the causality definitions).

\subsection{Structure of this paper}
\label{sec:intro.structure}

In \secref{sec:theory} we give an overview of an abstract theory of
systems. The main body of this work consists in formalizing systems
which capture quantum information processes and satisfy this abstract
theory \--- namely causal boxes.  In \secref{sec:space} we model the
state space of a wire. This constitutes the input and output spaces of
causal boxes. Then in \secref{sec:boxes} we define causal boxes as
maps from the input wires to the output wires. In
\secref{sec:representation} we prove some lemmas on the Stinespring
and \cjreps of causal boxes, which serve as a general toolbox in the
next sections. We then treat the composition of causal boxes in
\secref{sec:network}, where we define operations for connecting boxes
by their wires. In \secref{sec:metric} we show how to define a notion
of distance on the space of causal boxes using distinguishers. In
\secref{sec:qswitch} we illustrate our framework by using it to model
the quantum switch. We conclude in \secref{sec:conclusion} with some
final remarks and open questions. An overview of the appendices is
given on \pref{app}.

\section{A theory of systems}
\label{sec:theory}

Our goal is to model discrete quantum information\-/processing
systems, such that a network of these systems is itself a valid
system. In the following we provide a precise formulation of this
desideratum. For this we describe a theory of systems which captures
the required properties on an abstract level. It serves as a guideline
for the main technical contribution of this paper
(Sections~\ref{sec:space} to \ref{sec:metric}), which can be
understood as instantiating the theory with causal boxes. In this, we
follow the top\-/down paradigm of Maurer and Renner~\cite{MR11}, which
consists in modeling objects at the highest possible level of
abstraction, and then proceeding downwards, introducing in each new
lower level only the minimal necessary specializations. An alternative
and more detailed formulation of this abstract theory of systems can
be found in \cite{MMPRT16}. A very similar approach has been used by
Hardy~\cite{Har13} to model composition of abstract physical
objects. Many of the axioms stated in this section and in
\cite{MMPRT16} can also be found in Hardy's work~\cite{Har13}, e.g.,
composition order independence.

On an abstract level an information\-/processing system is an object
that reads inputs and produces outputs. To every such object $\Phi$ we
assign a set of ports, $\ports(\Phi)$, through which messages are
sent and received. Let $\fS$ be the set of all
information\-/processing systems of interest. Two systems can be
composed by connecting some of their ports. Let $\Phi,\Psi \in \fS$
and let $P \subseteq \ports(\Phi) \times \ports(\Psi)$ be a set of
pairs of compatible ports, i.e., they have the same dimensions,
consist of one in- and one out-port and each port appears at most once
in $P$. The composition operation, which we write 
$\Phi \connect{P} \Psi$, must satisfy certain properties.

\paragraph{Closure.} The first property needed is
\emph{closure under composition}, namely, for any $\Phi,\Psi \in \fS$
and compatible pairs of ports
$P \subseteq \ports(\Phi) \times \ports(\Psi)$,
\begin{equation*} \label{eq:closure} \Phi \connect{P} \Psi \in
  \fS \,. \end{equation*}

In \thmref{thm:closure} we prove that this holds for causal boxes.

\paragraph{Composition order independence.} When drawing systems one
typically produces figures as in \figref{fig:associativity}: boxes
connected by wires. Such a picture illustrates the fact that the order
in which the systems are connected should not matter. This property is
called \emph{composition order independence}: for any
$\Phi_1,\Phi_2,\Phi_3 \in \fS$ and compatible pairs of ports
$P_{ij} \subseteq \ports(\Phi_i) \times \ports(\Phi_j)$,
\begin{equation} \label{eq:comp.order.indep} \left(\Phi_1
  \connect{P_{12}} \Phi_2\right) \connect{P_{13} \cup P_{23}} \Phi_3
= \Phi_1 \connect{P_{12} \cup P_{13}} \left( \Phi_2 \connect{P_{23}}
  \Phi_3 \right) \,.\end{equation}

\begin{figure}[tb]
\begin{centering}

\begin{tikzpicture}[scale=1.2]\small

  \def\t{.4}
  \def\u{.8}

  \node[draw,thick,minimum height=2.4cm,minimum width=1.2cm] (r) at (-2,0) {$\Phi_1$};
  \node[draw,thick,minimum height=1.2cm,minimum width=1.2cm] (s) at (0,.5) {$\Phi_2$};
  \node[draw,thick,minimum height=2.4cm,minimum width=1.2cm] (t) at (2,0) {$\Phi_3$};

  \node[inner sep=0] (A) at (r.east |- 0,.7) {};
  \node[inner sep=0] (B) at (r.east |- 0,.3) {};
  \node[inner sep=0] (C) at (r.east |- 0,-.3) {};
  \node[inner sep=0] (D) at (r.east |- 0,-.7) {};

  \draw[<-,>=stealth,thick] (A.center) to (s.west |- A);
  \draw[->,>=stealth,thick] (B.center) to node[auto,yshift=-2] {$P_{12}$} (s.west |- B);
  \draw[<-,>=stealth,thick] (C.center) to (t.west |- C);
  \draw[->,>=stealth,thick] (D.center) to  node[auto,yshift=-2] {$P_{13}$} (t.west |- D);
  \draw[<-,>=stealth,thick] (s.east |- A) to (t.west |- A);
  \draw[->,>=stealth,thick] (s.east |- B) to node[auto,yshift=-2] {$P_{23}$} (t.west |- B);

\end{tikzpicture}

\end{centering}
\caption[Composition order
independence]{\label{fig:associativity}Systems $\Phi_1$ and $\Phi_2$
  are connected by the pairs of ports $P_{12}$, $\Phi_1$ and $\Phi_3$
  are connected by $P_{13}$, and $\Phi_2$ and $\Phi_3$ are connected
  by $P_{23}$. The resulting system must be independent of the order
  of the connections.}
\end{figure}
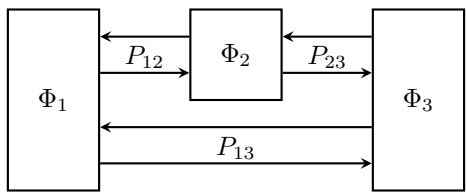

In \thmref{thm:comp.order.indep} we prove that
\eqnref{eq:comp.order.indep} is satisfied by the composition operation
on causal boxes.

\paragraph{Pseudo-metric.} It is often useful for a theory of systems
to provide a pseudo\-/metric on the space of systems $\fS$, namely a
function $d : \fS \times \fS \to \R^+$ such that for any
$\Phi,\Psi,\Gamma \in \fS$,
\begin{align*}
d(\Phi,\Phi) & = 0 \,, \\
d(\Phi,\Psi) & = d(\Psi,\Phi) \,, \\
d(\Phi,\Psi) & \leq d(\Phi,\Gamma) + d(\Gamma,\Psi) \,.
\end{align*}
If additionally  $d(\Phi,\Psi) = 0 \implies \Phi = \Psi$, then $d$ is a metric.

In \secref{sec:metric} we define a pseudo\-/metric on the space of
causal systems: the distinguishing advantage. This is defined in terms
of the probability that a distinguisher,\footnote{Formally, a
  distinguisher is another causal box that may additional put loops on
  the system to which it is plugged and outputs a bit corresponding to
  its guess.}  connected to one of two causal boxes $\Phi$ or $\Psi$,
can successfully guess with which causal box it is interacting. We
prove in \thmref{thm:metric} that this is indeed a pseudo\-/metric
\--- and if the set of all possible distinguishers is considered, the
distinguishing advantage is actually a metric.

\section{The space of partially ordered messages}
\label{sec:space}

\subsection{Ordering messages}
\label{sec:space.ordering}

In a network of information\-/processing devices, messages are
received and sent at a certain time. What is actually received (or
sent) by a system can be thought of as a pair $(v,t)$, where $v$ is
the message and $t$ the time at which it arrives (or the time at which
it is sent). All messages are then naturally ordered. If the
input\-/output behavior of a system is described without any ordering
information, one obtains situations as in \figref{fig:ndsystem}, where
the output of a composed system is undefined. However, time captures
more information than what is necessary to define an
ordering. 
For information\-/processing systems to be closed under composition,
it is sufficient for every message $v$ to be assigned a position in a
partially ordered set. We do this by modeling inputs and outputs of a
system as pairs $(v,t)$, where $v$ is a message and $t$ is an element
of some countable,\footnote{Since we consider discrete
  information\-/processing systems, it is natural that the input and
  output spaces also be discrete. $\T$ can be seen as an abstraction
  of some larger (physical) space, that contains extra (possibly
  uncountably many) points that are not relevant to the discrete
  systems modeled, since they are never used.} partially ordered set
$\T$, denoting the position\footnote{In the following we use
  ``position'' to emphasize the relative meaning of the partial order
  defined on $\T$, though conceptually it may be simpler to think of
  $t \in \T$ as the time at which a message is sent or received.} of
$v$ with respect to other messages.

To compose a system $\Phi$ in an arbitrary environment, where some
part of the environment might run before, after or between any events
occurring in the system $\Phi$, one would need to choose $\T$ such
that for any $t,t' \in \T$ with $t < t'$, there exist
$t_1,t_2,t_3 \in \T$ with $t_1 < t < t_2 < t' < t_3$, e.g., $\T = \Q$.
This could be needed for example in a cryptographic setting. For
modeling synchronized systems which are divided in rounds, $\T = \N$
might be sufficient. For defining the quantum switch depicted in
\figref{fig:qswitch}, if we are not modeling composition with any
other systems, then it is enough to choose $\T = \{1,\dotsc,6\}$,
since exactly $6$ messages are produced by these systems (see
\secref{sec:qswitch} for details, where the quantum switch is modeled
as a causal box). In this framework we do not specify the set $\T$ any
further than by its countable and partially ordered properties. It is
up to each application to define what is needed.

To model quantum messages, we simply define a Hilbert space with a
basis given by $\{(v,t)\}_{v,t}$. This is done in
\secref{sec:space.time}. As explained in
\secref{sec:intro.superposition}, the space of a wire should also
contain an element representing no message, multiple messages or any
superposition thereof. In the classical case this corresponds to the
space of all multisets of pairs $(v,t)$, i.e., a wire might transmit
nothing $\{\}$, one message $\{(v,t)\}$, two messages
$\{(v_1,t_1),(v_2,t_2)\}$, etc. In the quantum case, this corresponds
to the bosonic Fock space. We define the exact space of a wire and its
main properties in \secref{sec:space.wires}.

\begin{rem}[Terminology: wires and systems]
  \label{rem:terminology}
  It is standard in quantum information theory for the word ``system''
  to refer to the label of a Hilbert space, e.g., a (bipartite) system
  $AB$ is in a state $\rho_{AB}$. Since in this work the objects of
  study are not quantum states but objects that receive and send
  quantum states, we denote these interactive objects as
  (information\-/processing) \emph{systems}. We use the terms
  \emph{wire} and \emph{sub-wire} to label the input and output
  Hilbert spaces of a system, e.g., a causal box $\Phi$ might be a
  system that maps states from a (bipartite) input wire $AB$ to a
  (bipartite) output wire $CD$.
\end{rem}

\subsection{Single message space}
\label{sec:space.time}

As introduced in \secref{sec:space.ordering}, an input to a system is
a pair of a message $v \in \cV$ and a position $t \in \T$. The
corresponding quantum state is an element of a Hilbert space with an
orthonormal basis given by $\{\ket{v,t}\}_{v \in \cV, t \in T}$. For a
finite $\cV$ and infinite $\T$, this Hilbert space corresponds to
\begin{equation} \label{eq:hilbertspace}
  \Ltwo{\T}{\C^{|\cV|}}\,,\end{equation} where
$\LtwoOp(\T) = \{ (x_t)_{t \in \T} : x_t \in \C, \norm{x} <
\infty\}$
is the sequence space with bounded $2$\=/norm with
$\norm{x} = \sqrt{\braket{x}{x}}$ and\footnote{For the scalar
  product $\braket{x}{y}$ to be defined for arbitrary sequences
  $(x_t)_{t \in \T}$ and $(y_t)_{t \in \T}$, one would need to specify
  the order of the summation in \eqnref{eq:scalarproduct}. But for the
  set of sequences with bounded $2$\=/norm that define the Hilbert
  space $\LtwoOp(\T)$, all orders result in the same scalar product,
  so we omit
  it.}\begin{equation} \label{eq:scalarproduct}\braket{x}{y} = \sum_{t
    \in \T} \overline{x_t}y_t\,. \end{equation} In the following we
refer to a state $\ket{\psi} \in \LtwoTC$ as a qudit with position
information. The orthonormal basis
$ \{ \ket{v,t}\}_{v \in \cV, t \in \T }$ is given by
$\ket{v,t}= \ket{v} \tensor \ket{t}$, where
$\{\ket{v}\}_{v \in \cV}$ is an orthonormal basis of $\C^{|\cV|}$ and
$\ket{t}$ is the sequence with a $1$ in position $t \in \T$ and all
other elements of the sequence are $0$.

In the following we write $\ket{v_t}$ instead of $\ket{v,t}$ for the
message $\ket{v}$ arriving in position $t$. For example, the state
$\alpha \ket{0_x} + \beta\ket{1_x}$ corresponds to a message in a
superposition of $\zero$ and $\one$ arriving in position $x$. And the
state $\alpha \ket{0_x} + \beta\ket{0_y}$ corresponds to the message
$\zero$ in a superposition of arriving in positions $x$ and $y$.

\subsection{Wires}
\label{sec:space.wires}

To capture multiple qudits being sent on the same wire to some system,
we define the Hilbert space of a wire as a Fock space. For a Hilbert
space $\hilbert$, the corresponding bosonic Fock space is given by
\begin{equation*} 
  \fock{\hilbert} \coloneqq
  \bigoplus_{n = 0}^\infty \vee^n \hilbert\,, \end{equation*} where
$\vee^n \hilbert$ denotes the symmetric subspace of
$\hilbert^{\tensor n}$, and $\hilbert^{\tensor 0}$ is the one
dimensional space containing the vacuum state $\vacuum$. This is
explained in more detail in \appendixref{app:fock}. Applying this to
the Hilbert space of a qudit with position information from
\eqnref{eq:hilbertspace}, we get
\begin{equation} \label{eq:quditfockspace}
\fock{\LtwoTC} = \bigoplus_{n = 0}^\infty \vee^n \left(\LtwoTC\right)\,.
\end{equation}
The orthogonal subspaces $\vee^n \left(\LtwoTC\right)$ for
$n \in \N_0$ capture $n$ messages being sent on a wire. The
restriction to the symmetric space guarantees that there is no order
amongst the qudits other than what might be defined from their state,
e.g., their position in $\T$. For example, if a wire contains two
copies of a state $\zero$ in position $t$, it would be in the state
$\ket{0_t} \tensor \ket{0_t}$. If it contains two copies of $\zero$ at
different positions $t_1$ and $t_2$, it would be in the state
$\frac{1}{\sqrt 2}\left( \ket{0_{t_1}} \tensor \ket{0_{t_2}} +
  \ket{0_{t_2}} \tensor \ket{0_{t_1}} \right)$.
When modeling concrete systems, one may always restrict the model by
adding constraints, e.g., one may consider only systems that produce
at most one message or never produce multiple messages at the same
position $t \in \T$.

Let $A$ denote a wire that carries
$d_A$\=/dimensional messages. We write $\cF^{\T}_A$ for the
corresponding state space, namely
\[\cF^{\T}_A \coloneqq \fock{\Ltwo{\T}{\C^{d_A}}}\,.\]
The subscript allows the spaces of different wires to be
distinguished, e.g., two wires $A$ and $B$ have joint space
$\cF^{\T}_A \tensor \cF^{\T}_B$, which we also write $\cF^{\T}_{AB}$.
For any Hilbert spaces $\hilbert_A$ and
$\hilbert_B$,
\begin{equation} \label{eq:fockisomorphism} \fock{\hilbert_A} \tensor
  \fock{\hilbert_B} \cong \fock{\hilbert_A \oplus
    \hilbert_B}\,,\end{equation} where the isomorphism preserves the
meaning associated with the bases of the Fock spaces, i.e., a tensor
product of two vacuum states on the left in
\eqnref{eq:fockisomorphism} is mapped to a vacuum state on the right,
a tensor product of a vacuum state and one message on the left is
mapped to a single message with the same value and position on the
right, etc.\ (see \remref{rem:fockisomorphism}
in \appendixref{app:fock} for the exact isomorphism). From
\eqnref{eq:fockisomorphism} and the fact that
\[ \left(\Ltwo{\T}{\C^{d_A}}\right) \oplus
\left(\Ltwo{\T}{\C^{d_B}}\right) \cong \Ltwo{\T}{\left(\C^{d_A} \oplus
    \C^{d_B}\right)}\,, \]
one can see that $\cF^{\T}_{AB}$ also satisfies
\eqnref{eq:quditfockspace}, i.e., it can be interpreted as the state
space of a single wire carrying messages of dimension $d_A+d_B$. The
converse also holds: any wire $A$ of messages of dimension $d_A$ can
be split in two sub-wires $A_1$ and $A_2$ of messages of dimensions
$d_{A_1}+d_{A_2} = d_A$. We then have
\begin{equation} \label{eq:splitspace}
\cF^{\T}_{A} \cong \cF^{\T}_{A_1} \tensor \cF^{\T}_{A_2}\,.
\end{equation}
Thus, a (sub-)wire $A_1$ \--- or more precisely, the corresponding
message subspace $\C^{d_{A_1}}$ \---may be regarded as a way of
labeling a subspace of $\C^{d_A}$, which is useful to specify how
causal boxes are connected with each other (e.g., messages in the
sub-wire corresponding to $\C^{d_{A_1}}$ are connected to one causal
box and those in $\C^{d_{A_2}}$ are connected to another, see
\secref{sec:network}). In the following we refer to the dimension of
the messages of a wire as the dimension of the wire, e.g., if we say
that the wire $A$ has dimension $d_A$, then
$\cF^{\T}_A = \fock{\Ltwo{\T}{\C^{d_A}}}$.

For any subset $\cP \subseteq \T$, we write
\begin{equation}
\label{eq:partialtime}
\cF^{\cP}_A \coloneqq \fock{\Ltwo{\cP}{\C^{d_A}}}
\end{equation} for the space of
all states occurring at some position in $\cP$, e.g.,
\[
\cF^{\leq t}_A = \fock{\Ltwo{\T^{\leq t}}{\C^{d_A}}}\,, 
\]
where $\T^{\leq t} \coloneqq \{s \in \T : s \leq t\}$. From
\eqnref{eq:fockisomorphism} and
\[\Ltwo{\T}{\C^{d_A}} \cong \left(\Ltwo{\cP}{\C^{d_A}}\right) \oplus
\left(\Ltwo{\widetilde{\cP}}{\C^{d_A}}\right)\,,\]
where $\widetilde{\cP} \coloneqq \T \setminus \cP$, we have
\begin{equation} \label{eq:splittime}
\cF^{\T}_A \cong \cF^{\cP}_A \tensor \cF^{\widetilde{\cP}}_A \,.
\end{equation}
This can be interpreted as splitting a wire $A$ into two sub-wires
carrying the messages that are in positions in $\cP$ and
$\widetilde{\cP}$, respectively.

A natural embedding of $\cF^{\cP}_A$ in $\cF^{\T}_A$ is obtained by
appending ``nothing'' to $\cF^{\cP}_A$, i.e.,
\begin{equation} \label{eq:notation.embedding} \cF^{\cP}_A \cong
  \cF^{\cP}_A \tensor \vacuum^{\widetilde{\cP}}_A \subseteq \cF^{\T}_A
  \,, \end{equation} where $\vacuum^{\widetilde{\cP}}_A$ denotes the
one dimensional subspace of $\cF^{\widetilde{\cP}}_A$ that contains
the vacuum state. Throughout this work we use $\cF^{\cP}_A$ to denote
both the space defined in \eqnref{eq:partialtime} as well as its
embedding in $\cF^{\T}_A$ given on the right\-/hand side of
\eqnref{eq:notation.embedding}.

\section{Defining causal boxes}
\label{sec:boxes}

Intuitively, a causal box is a transformation from an input message
space to an output message space that satisfies causality. We
introduce these two aspects \--- the transformation and notion of
causality \--- separately. First, in \secref{sec:boxes.set} we
formalize the notion of a transformation from the input to the output
space as a set of maps. In \secref{sec:boxes.cuts} we define
terminology that provides us with a more convenient representation of
these maps. Then in \secref{sec:boxes.causality} we introduce the
notion of causality that these maps must satisfy in order to be a
valid causal box. We put these two aspects together in
\secref{sec:boxes.definition}, where we give the formal definition of
a causal box. Finally, in \secref{sec:boxes.subnormalized} we define
subnormalized causal boxes.

\subsection{A set of maps}
\label{sec:boxes.set}

Let $\cF^\T_X$ and $\cF^\T_Y$ denote the Hilbert spaces of an input
wire $X$ and output wire $Y$. And let $\tcop{\cF^\T_X}$ and
$\tcop{\cF^\T_Y}$ be the corresponding sets of trace class
operators.\footnote{$V \in \tcop{\hilbert}$ if
  $\trnorm{V} = \sum_i \bra{i} \sqrt{\hconj{V}V} \ket{i} < \infty$,
  where $\{\ket{i}\}$ is an orthonormal basis of $\hilbert$. For
  example, a density operator on a wire $A$ is a non\-/negative
  self\-/adjoint operator $\rho \in \tcop{\cF^{\T}_A}$ with
  $\tr \rho = 1$.} A causal box with these input and output wires can
be thought of as a transformation from the message space of $X$ to the
message space of $Y$. But instead of defining it as a map
$\Phi : \tcop{\cF^\T_X} \to \tcop{\cF^\T_Y}$, a causal box is given by
a set of mutually consistent\footnote{\emph{Mutually consistent} means
  that any two maps $\Phi^{\leq t}$ and $\Phi^{\leq u}$ with
  $t \leq u$ must produce the same output on $\T^{\leq t}$, i.e.,
  $\Phi^{\leq t} = {\tr_{\nleq t}} \circ \Phi^{\leq u}$, where
  $\tr_{\nleq t}$ traces out messages in positions
  $\T^{\leq u} \setminus \T^{\leq t}$. This is introduced formally in
  \secref{sec:boxes.definition}.} maps
\begin{equation} \label{eq:boxes.set} \Phi = \left\{\Phi^{\leq t} :
    \tcop{\cF^\T_X} \to \tcop{\cF^{\leq t}_Y}\right\}_{t \in
    \T}\,,\end{equation} where $\cF^{\leq t}_Y$ is the subspace of
$\cF^\T_Y$ that contains only messages in positions
$\T^{\leq t} = \{p \in \T : p \leq t\}$. This allows systems to be
included that produce an unbounded number of messages on the entire
set $\T$. For example, let $\T = \N$ and consider a beacon system that
outputs a state $\zero$ at every point $t \in \N$.  This is
well\-/defined on every subset $\{1,\dotsc,t\}$, but the limit
behavior is not, as it would consist in a box that outputs an
``infinite tensor product'' of $\zero$.

\begin{rem}[Finite causal boxes]
\label{rem:single}
In \appendixref{app:finite} we show how to define the subset of causal
boxes that can be represented by a single map
$\Phi : \tcop{\cF^\T_X} \to \tcop{\cF^\T_Y}$ and which is closed under
composition. This roughly corresponds to the set of systems that stop
processing inputs after some point $t_{\max} \in \T$ \--- or after
some set of unordered points $T_{\max} \subseteq \T$ in the case of a
partial order on $\T$.
\end{rem}

\subsection{Cuts}
\label{sec:boxes.cuts}

A map $\Phi^{\leq t}$ from \eqnref{eq:boxes.set} produces outputs on
the subset of positions $\T^{\leq t} \subseteq \T$, which contains all
points $p \leq t$. This subset allows the behavior of a system up to a
certain position $t$ to be defined. More generally, we are interested
in describing the behavior of systems on subsets of positions
$\cC \subseteq \T$ that are not (necessarily) upper bounded by a
single point $t$, but by a set $\cP$, e.g.,
$\cC = \cup_{p \in \cP} \T^{\leq p}$. We refer to such a subset of
$\T$ as a \emph{cut}, which we illustrate in \figref{fig:cut}.

\begin{figure}[tb]
\begin{centering}

\begin{tikzpicture}[
wire/.style={->,>=stealth,thick},
point/.style={fill,draw,circle,inner sep=1.5pt}]

  \def\x{1.2}

  \node[point] (t1) at (0,0) {};
  \node[point] (t2) at (0,-\x) {};
  \node[point] (t3) at (0,-2*\x) {};
  \node[point] (t4) at (\x,0) {};
  \node[point] (t5) at (\x,-\x) {};
  \node[point] (t6) at (\x,-2*\x) {};
  \node[point] (t7) at (2*\x,-\x/2) {};
  \node[point] (t8) at (2*\x,-3*\x/2) {};
  \node[point] (t9) at (3*\x,-\x) {};
  \node[point] (t10) at (3*\x,-2*\x) {};
  \node[point] (t11) at (4*\x,-\x/2) {};

  \draw[wire] (t1) to (t4);
  \draw[wire] (t2) to (t5);
  \draw[wire] (t3) to (t6);
  \draw[wire] (t4) to (t7);
  \draw[wire] (t5) to (t7);
  \draw[wire] (t5) to (t8);
  \draw[wire] (t6) to (t8);
  \draw[wire] (t7) to (t11);
  \draw[wire] (t8) to (t9);
  \draw[wire] (t8) to (t10);
  \draw[wire] (t9) to (t11);

  \node[inner sep=0] (p1) at (3*\x/2,\x/2) {};
  \node[inner sep=0] (p2) at (3*\x/2,-5*\x/2) {};
  \node[below left] at (p1) {$\cC$};
  \node[inner sep=0] (p3) at (5*\x/2,\x/2) {};
  \node[inner sep=0] (p4) at (5*\x/2,-5*\x/2) {};
  \node[below left] at (p3) {$\cD$};

  \draw[dashed,thick] (p1) to (p2);
  \draw[dashed,thick] (p3) to (p4);

  \node[fit=(t1)(t3)(t11),inner sep=.6cm,thick,draw] (T) {};
  \node[above left] at (T.north east) {$\T$};

\end{tikzpicture}

\end{centering}
\caption[Cuts]{\label{fig:cut}A partially ordered set $\T$ with two
  cuts $\cC \subseteq \cD \subseteq \T$ containing $6$ and $8$ points,
  respectively.}
\end{figure}
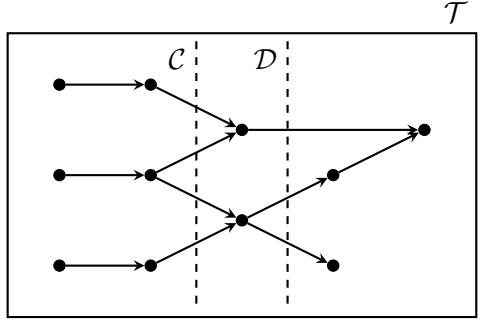

\begin{deff}[Cut]
\label{def:cut}
Let $\T$ be a partially ordered set. We define a \emph{cut of $\T$} as
any subset $\cC \subseteq \T$ such that
\[ \cC = \bigcup_{t \in \cC} \T^{\leq t}\,,\]
where $\T^{\leq t} = \{p \in \T : p \leq t\}$. We say that a cut $\cC$
is \emph{bounded} if there exists a point $t \in \T$ such that
$\cC \subseteq \T^{\leq t}$. We denote the set of all cuts of $\T$ as
$\cut$ and the set of all bounded cuts as $\bcut$.
\end{deff}

Although it is sufficient to describe a causal box by a set of maps as
in \eqnref{eq:boxes.set}, we will actually define it as a set of
(mutually consistent) maps (that respect causality),
\[\Phi = \left\{\Phi^\cC : \tcop{\cF^\T_X} \to \tcop{\cF^{\cC}_Y}\right\}_{\cC \in
  \bcut}\,.\]
These can be derived from \eqnref{eq:boxes.set} by setting
$\Phi^{\cC} \coloneqq {\tr_{\widetilde{\cC}}} \circ \Phi^{\leq t}$,
where $t$ is any point such that $\cC \subseteq \T^{\leq t}$,
$\widetilde{\cC} \coloneqq \T^{\leq t} \setminus \cC$ and
$\tr_{\widetilde{\cC}}$ trace out all messages in positions in
$\widetilde{\cC}$.

\subsection{Causality}
\label{sec:boxes.causality}

To be valid, an information\-/processing system must respect
\emph{causality}: an output can only depend on past inputs. For a
totally ordered set $\T$ one can formalize this by requiring that for
every causal box there exist a monotone function $\chi : \T \to \T$
such that the output up to position $t$ can be computed from the input
up to position $\chi(t) < t$. In the case of a partially ordered $\T$,
there might be many unordered points $p < t$ that are needed to
compute the output up to position $t$. We thus define the causality
function on cuts, $\chi : \cut \to \cut$, and require that an output
on $\cC \in \bcut$ can be computed from the input on
$\chi(\cC) \subsetneq \cC$.

To be consistent, if the output on $\cC$ can be computed from
$\chi(\cC)$ and the output on $\cD$ computed from $\chi(\cD)$, then
the output on $\cC \cup \cD$ can be computed from
$\chi(\cC) \cup \chi(\cD)$. So we require
$\chi(\cC \cup \cD) = \chi(\cC) \cup \chi(\cD)$.\footnote{This follows
  immediately if one defines a function $\hat{\chi} : \T \to \cut$
  such that the output on $\T^{\leq t}$ can be computed from
  $\hat{\chi}(t)$ and one sets
  $\chi(\cC) \coloneqq \bigcup_{t \in \cC} \hat{\chi}(t)$.}
Furthermore, we also expect that if $\cC \subseteq \cD$, then
$\chi(\cC) \subseteq \chi(\cD)$, because if $\chi(\cC)$ is needed to
compute the output on $\cC$, then certainly it is needed for the
output on $\cD \supseteq \cC$.

These conditions are however not sufficient to guarantee that we have
well\-/defined systems. Consider for example a system with
$\T = \Q^+$, that for every input received in position $1-t$, for
$0 < t \leq 1$, produces an output in position $1-t/2$, i.e.,
$\chi([0,1-t/2]) = [0,1-t]$. Furthermore, this system initially
outputs a message in position $0$. If now the messages output are
looped back to the input, this system should produce messages at
points $\{0,1/2,3/4,7/8,15/16,\dotsc\}$. It would effectively output
an infinite number of messages before position $t=1$, which is
ill\-/defined. The problem here is that the gap between every input
and correlated output gets smaller and smaller as $t \to 1$, and this
system executes an infinite number of (causal) steps before reaching
$t=1$. We thus additionally need that every point $t' \in \T$ may be
reached from any point $t \leq t'$ in a finite number of causal
steps. By definition, an output on $\T^{\leq t'}$ may be directly
computed from an input on $\chi(\T^{\leq t'})$. Given an input on
$\chi^2(\T^{\leq t'})$ on may compute the output on $\T^{\leq t'}$ in
two steps: one first obtains the output on $\chi(\T^{\leq t'})$, then
given the input on $\chi(\T^{\leq t'})$ one gets the output on
$\T^{\leq t'}$. In this sense, a point $t'$ is reached from $t$ in a
finite number of steps of there exists an $n$ such that
$t \notin \chi^n(\T^{\leq t'})$, which is captured in
\defref{def:causality.function} by \eqnref{eq:causality.finite}.

\begin{deff}[Causality function]
\label{def:causality.function}
A function $\chi : \cut \to \cut$ is a \emph{causality function} if it
satisfies the following conditions.
\begin{align}
  \forall \cC,\cD \in \cut, \quad & \chi(\cC \cup \cD) = \chi(\cC) \cup \chi(\cD)\,, \label{eq:causality.homomorphism} \\
  \forall \cC,\cD \in \cut, \quad & \cC \subseteq \cD \implies \chi(\cC) \subseteq \chi(\cD)\,, \label{eq:causality.monotone} \\
  \forall \cC \in \bcut \setminus \{\emptyset\}, \quad & \chi(\cC) \subsetneq \cC\,, \label{eq:causality.decreasing} \\
  \forall \cC \in \bcut, \forall t \in \cC, \exists n \in \N, \quad & t \notin \chi^n\left(\cC\right)\,, \label{eq:causality.finite}
\end{align} where $\chi^n$ denotes $n$ compositions of
$\chi$ with itself, $\chi^n = \chi \circ \dotsb \circ \chi$.
\end{deff}

For example, if $\T = \Q$, a minimum delay $\delta$ between every
input and correlated output is sufficient to satisfy
\defref{def:causality.function}. Our definition is thus a
generalization of \emph{delta\-/causality}~\cite{LS98}, which requires
exactly such a delay $\delta$. Our definition is however more
restrictive than the notion of causality used in \cite{MMPRT16}, which
simply requires every output to depend only on inputs received
strictly earlier, thus allowing an infinite number of outputs to be
produced in a finite amount of time.\footnote{This is achieved by
  restricting the systems considered to classical deterministic
  systems, for which an infinite string of bits is well\-/defined (as
  opposed to an infinite tensor product of qubits, which is not).}

\begin{rem}[Causality in totally ordered systems]
\label{rem:infcondition}
In the case of a totally ordered $\T$, instead of
\eqnref{eq:causality.finite} one may alternatively require that for
$\chi : \T \to \T$ the strict inequality $\chi(t) < t$ must also hold
in the limit as $t \to t_0$, i.e.,
$\inf_{t > t_{0}} \chi(t) < \inf_{t > t_{0}} t$ and
$\sup_{t < t_{0}} \chi(t) < \sup_{t < t_{0}} t$. This is discussed
in \appendixref{app:total}, where we prove that it implies
\eqnref{eq:causality.finite}.
\end{rem}

\subsection{The definition of causal boxes}
\label{sec:boxes.definition}

We can now formally define a causal box.

\begin{deff}[Causal box]
  \label{def:quantumbox} 
  A \emph{$(d_X,d_Y)$\-/causal box} $\Phi$ is a system with input
  wire $X$ and output wire $Y$ of dimension $d_X$ and $d_Y$, defined
  by a set of mutually consistent, completely positive,
  trace\-/preserving (CPTP) maps
\[\Phi = \left\{\Phi^\cC : \tcop{\cF^\T_X} \to \tcop{\cF^{\cC}_Y}\right\}_{\cC \in
  \bcut}\,,\] i.e., for all $\cC,\cD \in \bcut$ with $\cC \subseteq \cD$,
\begin{equation} \label{eq:consistency}
  \Phi^\cC = {\tr_{\cD \setminus \cC}} \circ \Phi^\cD\,,
\end{equation} 
where $\tr_{\cD \setminus \cC}$ traces out the messages occurring at
positions in $\cD \setminus \cC$.\footnote{As noted in
  \eqnref{eq:splittime},
  $\cF^{\cD}_Y \cong \cF^{\cC}_Y \tensor \cF^{\cD\setminus\cC}_Y$.}
Furthermore, these maps must respect causality: there must exist a function
$\chi : \cut \to \cut$ satisfying \defref{def:causality.function} such
that for all $\cC \in \bcut$,\footnote{\eqnref{eq:causality} uses the
  embedding of $\cF^{\chi(\cC)}_X$ in $\cF^{\T}_X$ so that the output
  space of $\tr_{\T \setminus \chi(\cC)}$ matches the input space of
  $\Phi^\cC$.}
  \begin{equation}
    \label{eq:causality} \Phi^\cC = \Phi^\cC \circ \tr_{\T \setminus \chi(\cC)}\,.
  \end{equation}
\end{deff}

\eqnref{eq:causality} says that the output on positions in $\cC$ can
be computed from the input on positions $\chi(\cC) \subsetneq \cC$,
i.e., the output on $\cC$ can only depend on inputs that are before,
namely $\chi(\cC)$. This equation may be alternatively expressed as
requiring that $\Phi^\cC = \hat{\Phi}^\cC \tensor \tr$, where
$\hat{\Phi}^\cC$ is some map
$\tcop{\cF^{\chi(\cC)}_X} \to \tcop{\cF^{\cC}_Y}$ and $\tr$ acts on
the space $\cF^{\T \setminus \chi(\cC)}_X$. Using this,
\defref{def:quantumbox} can be written more compactly as a set of CPTP
maps
\[\Phi = \left\{\Phi^\cC : \tcop{\cF^{\chi(\cC)}_X} \to
  \tcop{\cF^{\cC}_Y} \right\}_{\cC \in \bcut}\,,\]
such that for all $\cC,\cD \in \bcut$, $\cC \subseteq \cD$,
\begin{equation} 
  \label{eq:combined}
  {\tr_{\cD \setminus \cC}} \circ \Phi^\cD = \Phi^\cC \circ \tr_{\T \setminus \chi(\cC)}\,.
\end{equation} In the following we will often use this form.

\begin{rem}[Ports]
\label{rem:ports}
From \eqnref{eq:splitspace} we know that a wire can be split into a
tensor product of sub-wires. It is thus sufficient to define causal
boxes with one index $X$ ($Y$) for the input (output) wires, as these
can be subdivided as needed. For example, $Y$ could consist of two
sub-wires $Y_1$ and $Y_2$, which are then connected to different
systems $\Psi$ and $\Gamma$. The ports of an information\-/processing
system, $\ports(\Phi)$, introduced in \secref{sec:theory} to define
composition of systems, correspond to a predefined partition of the
input and output wires in sub-wires.
\end{rem}

\subsection{Subnormalized boxes}
\label{sec:boxes.subnormalized}

\defref{def:quantumbox} only considers trace\-/preserving causal
boxes. In \defref{def:subbox} we generalize this to include
trace\-/decreasing maps. This can be used to model, for example, the
subnormalized box resulting from post-selecting states on some
measurement outcome.

\begin{deff}[Subnormalized causal box]
  \label{def:subbox} A set of completely positive (CP) maps,
  \[ \Phi = \left\{\Phi^\cC : \tcop{\cF^{\T}_X} \to \tcop{\cF^{\cC}_{Y}} \right\}_{\cC \in \bcut} \]
  defines a \emph{$(d_X,d_Y)$\-/subnormalized causal box} if there
  exists a normalized causal box
  \[ \hat{\Phi} = \left\{\hat{\Phi}^\cC : \tcop{\cF^{\T}_X} \to
    \tcop{\cF^{\cC}_{RY}} \right\}_{\cC \in \bcut}\]
  such that for all ${\cC \in \bcut}$,
  \[ \Phi^\cC = P^{\Omega}_{R} \circ \hat{\Phi}^\cC\,,\]
  where
  $P^{\Omega}_{R} : \tcop{\cF^{\T}_{RY}} \to \tcop{\cF^{\T}_{Y}}$
  projects the $R$ sub-wire on the vacuum state, namely
\[ P^{\Omega}_{R}(\rho) \coloneqq \left( \bra{\Omega}_{R} \tensor
I_{Y} \right) \rho \left( \vacuum_{R} \tensor I_{Y} \right)\,. \]
\end{deff}

\section{Alternative representations of causal boxes}
\label{sec:representation}

In this section we provide some alternative characterizations of
causal boxes using the Stinespring\footnote{For a detailed treatment
  of the \strep of quantum operators on infinite dimensional spaces,
  we refer to the textbook~\cite{Pau03}. A more accessible
  introduction to the finite dimensional case \--- which shares many
  of the essential properties with the infinite dimensional setting
  \--- can be found in~\cite{Wat16}.} and \cj\footnote{Since the
  \cjrep of infinite dimensional operators differs significantly from
  the finite dimensional case, we have summarized the most important
  points in \appendixref{app:rep.cj} and refer the interested reader
  to \cite{Hol11}.} representations. These are then used as technical
tools in the proofs of the main theorems in the following sections.

\subsection{Stinespring and Choi-Jamio\l{}kowski}
\label{sec:representation.cj}

In \secref{sec:boxes.definition} causal boxes are defined as a set of
maps $\left\{\Phi^\cC \right\}_{\cC \in \bcut}$ subject to a
causality constraint. In this section we show how to represent
causal boxes using the Stinespring and \cjreps. The \strep of a
causal box can easily be drawn as a circuit, which we do in
\figref{fig:causality} \--- \corref{cor:representation} here below
states that \figref{fig:causality} is indeed equivalent to
\eqnref{eq:combined}.

\begin{figure}[tb]
\begin{centering}
\begin{tikzpicture}[
wire/.style={-,thick},
unitary2/.style={draw,thick,minimum width=1cm,minimum height=1.8cm}]

\node (l1) at (.5,0) {};
\node[above right] at (l1) {\small $\cF^{\chi(\cC)}_X$};
\node (l3) at (.5,-2) {};
\node[above right,xshift=-.2cm] at (l3) {\small $\cF^{\chi(\cD)\setminus\chi(\cC)}_X$};
\node (r1) at (7.5,0) {};
\node[above left] at (r1) {\small $\cF^\cC_Y$};
\node (r2) at (7.5,-1) {};
\node[above left] at (r2) {\small $\cF^{\cD\setminus\cC}_Y$};
\node (r3) at (7.5,-2) {};
\node[above left,xshift=-2] at (r3) {\small $\hilbert_R$};

\node[unitary2] (u1) at (3,-.5) {$U^{\cC}_\Phi$};
\node[unitary2] (u2) at (5,-1.5) {$V$};
\node[fit=(u1)(u2),thick,dashed,draw,inner sep=.3cm] (u) {};
\node[above right] at (u.north west) {$U^{\cD}_\Phi$};

\draw[wire] (l1) to (u1.west |- l1);
\draw[wire] (l3) to (u2.west |- l3);
\draw[wire] (u1.east |- r2) to node[auto,yshift=-2] {$\hilbert_Q$} (u2.west |- r2);
\draw[wire] (u1.east |- r1) to (r1);
\draw[wire] (u2.east |- r2) to (r2);
\draw[wire,-|] (u2.east |- r3) to (r3);
\end{tikzpicture}

\end{centering}
\caption[The causality constraint]{\label{fig:causality}For any
  $\cC, \cD \in \bcut$ with $\cC \subseteq \cD$, a map
  $\Phi^\cD : \tcop{\cF^{\chi(\cD)}_X} \to \tcop{\cF^{\cD}_Y}$ of a
  causal box \--- with \strep $U^\cD_\Phi$ \--- can be decomposed
  into a sequence of two isometries $U^\cC_\Phi$ and $V$, where
  $U^\cC_\Phi$ is a \strep of
  $\Phi^\cC : \tcop{\cF^{\chi(\cC)}_X} \to \tcop{\cF^{\cC}_Y}$.}
\end{figure}
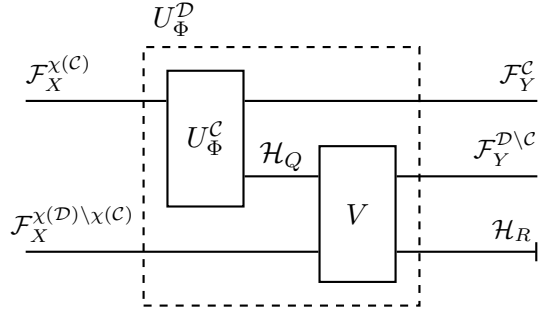

\begin{lem} \label{lem:representation}
Let $\Phi : \tcop{\hilbert_{AB}} \to \tcop{\hilbert_{CD}}$ be a CPTP
map for which the output on $C$ does not depend on the input to $B$,
i.e., there exists a CPTP map $\Psi : \tcop{\hilbert_A} \to
\tcop{\hilbert_C}$ such that
\begin{equation} \label{eq:equiv} {\tr_D} \circ \Phi = \Psi \circ
  \tr_B\,. \end{equation} The following two conditions are equivalent to
\eqnref{eq:equiv}.
\begin{enumerate}
\item For any states $\psi_A,\varphi_A \in \hilbert_A$,
  $\psi_B,\varphi_B \in \hilbert_B$,
  $\psi_C,\varphi_C \in \hilbert_C$, and any basis $\{\ket{j}\}_j$ of
  $\hilbert_D$,
  \begin{multline} \label{eq:equiv.cj}\sum_j R_\Phi \left( \psi_C
      \tensor j_D \tensor \psi_A \tensor \psi_B ; \varphi_C \tensor
      j_D \tensor \varphi_A \tensor \varphi_B \right) \\ = R_\Psi
    \left( \psi_C \tensor \psi_A ; \varphi_C \tensor \varphi_A \right)
    \braket{\psi_B}{\varphi_B}\,, \end{multline} where the sesquilinear
  forms\footnote{$R(\cdot;\cdot)$ is a sesquilinear form if it is
    antilinear in the first argument and linear in the second,
    see \appendixref{app:rep.cj}.}
\begin{align*}
  &  R_\Phi : \left(\hilbert_{CD} \times \hilbert_{AB} \right) \times
    \left(\hilbert_{CD} \times \hilbert_{AB} \right) \to \C\\
  \text{and} \qquad & R_\Psi : \left(\hilbert_{C} \times \hilbert_{A} \right) \times
                      \left(\hilbert_{C} \times \hilbert_{A} \right) \to \C
\end{align*}
are the \cjreps of the maps $\Phi$ and $\Psi$
(see \appendixref{app:rep.cj}).
\item There exists an isometry $V : \hilbert_{QB} \to \hilbert_{DR}$
  such that \begin{equation} \label{eq:equiv.st} U_\Phi = \left( I_C
      \tensor V \right) \left( U_\Psi \tensor I_B
    \right)\,,\end{equation} where $U_{\Phi} : \hilbert_{AB} \to
  \hilbert_{CDR}$ and $U_{\Psi} : \hilbert_{A} \to \hilbert_{CQ}$ are
  minimal \streps of $\Phi$ and $\Psi$.
\end{enumerate}
\end{lem}

Setting $\hilbert_A = \cF^{\chi(\cC)}_X$,
$\hilbert_B = \cF^{\widetilde{\chi(\cC)}}_X$,
$\hilbert_C = \cF^{\cC}_Y$ and $\hilbert_D = \cF^{\widetilde{\cC}}_Y$ in this
lemma, where
$\widetilde{\chi(\cC)} \coloneqq \chi(\cD)\setminus\chi(\cC)$ and
$\widetilde{\cC} \coloneqq \cD\setminus\cC$, we get the following
corollary, which gives an alternative characterization of
\eqnref{eq:combined} and is illustrated in \figref{fig:causality}.

\begin{cor} \label{cor:representation} Let
  $\{\Phi_\cC : \tcop{\cF^{\chi(\cC)}_X} \to \tcop{\cF^{\cC}_Y}\}_{\cC
    \in \bcut}$
  be a set of CPTP maps. For any $\cC,\cD \in \bcut$ with
  $\cC \subseteq \cD$, the following three conditions are equivalent:
\begin{enumerate}
\item These maps satisfy \eqnref{eq:combined},
  namely \begin{equation} \label{eq:rep.equiv} {\tr_{\widetilde{\cC}}} \circ \Phi^{\cD}
    = \Phi^{\cC} \circ \tr_{\widetilde{\chi(\cC)}}\,. \end{equation}
\item For any states $\psi^{\chi(\cC)}_X,\varphi^{\chi(\cC)}_X \in \cF^{\chi(\cC)}_X$,
  $\psi^{\widetilde{\chi(\cC)}}_X,\varphi^{\widetilde{\chi(\cC)}}_X \in \cF^{\widetilde{\chi(\cC)}}_X$,
  $\psi^{\cC}_Y,\varphi^{\cC}_Y \in \cF^{\cC}_Y$, and any
  basis $\{\ket{j}\}_j$ of $\cF^{\widetilde{\cC}}_Y$, where
  $\cF^{\chi(\cD)}_X \cong \cF^{\chi(\cC)}_X \tensor \cF^{\widetilde{\chi(\cC)}}_X$ and
  $\cF^{\cD}_Y \cong \cF^{\cC}_Y \tensor \cF^{\widetilde{\cC}}_Y$,
  \begin{multline} \label{eq:rep.equiv.cj} \sum_{j}
    R^{\cD}_{\Phi}\left(\psi^{\cC}_Y \tensor j \tensor \psi^{\chi(\cC)}_X
      \tensor \psi^{\widetilde{\chi(\cC)}}_X ; \varphi^{\cC}_Y \tensor j \tensor
      \varphi^{\chi(\cC)}_X \tensor \varphi^{\widetilde{\chi(\cC)}}_X \right) \\ =
    R^{\cC}_{\Phi}\left(\psi^{\cC}_Y \tensor \psi^{\chi(\cC)}_X ;
      \varphi^{\cC}_Y \tensor \varphi^{\chi(\cC)}_X\right)
    \braket{\psi^{\widetilde{\chi(\cC)}}_X}{\varphi^{\widetilde{\chi(\cC)}}_X}\,, \end{multline} where
  $R^{\cD}_\Phi$ and $R^{\cC}_\Phi$ are the \cjreps of $\Phi^{\cD}$ and
  $\Phi^{\cC}$.
\item There exists an isometry
  $V : \hilbert_{Q} \tensor \cF^{\widetilde{\chi(\cC)}}_X \to \cF^{\widetilde{\cC}}_Y
  \tensor \hilbert_{R}$
  such that \begin{equation} \label{eq:rep.equiv.st} U^{\cD}_\Phi
    = \left( I_Y^{\cC} \tensor V \right) \left(
      U^{\cC}_\Phi \tensor I_X^{\widetilde{\chi(\cC)}} \right) \,,\end{equation}
  where
  $U^{\cD}_{\Phi} : \cF^{\chi(\cD)}_X \to \cF^{\cD}_Y \tensor \hilbert_{R}$
  and
  $U^{\cC}_{\Phi} : \cF^{\chi(\cC)}_X \to \cF^{\cC}_Y \tensor \hilbert_{Q}$
  are minimal \streps of $\Phi^{\cD}$ and $\Phi^{\cC}$.
\end{enumerate}
\end{cor}

Since the three conditions are equivalent, this corollary not only
states that a causal box can be described by the \cjrep of the maps
$\Phi^\cC$, but also that any set of positive semi\-/definite
sesquilinear forms satisfying \eqnref{eq:rep.equiv.cj} is a valid
causal box.

We now provide a proof of \lemref{lem:representation}. As
visual aid, we draw the corresponding system in \figref{fig:equiv}, which
is a copy of \figref{fig:causality} with different labels.

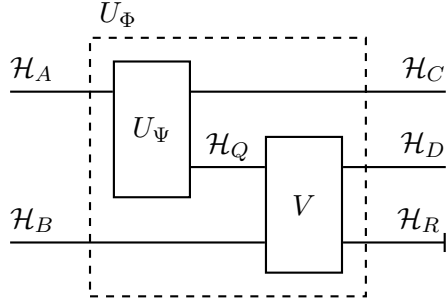
\begin{figure}[tb]
\begin{centering}
\begin{tikzpicture}[
wire/.style={-,thick},
unitary2/.style={draw,thick,minimum width=1cm,minimum height=1.8cm}]

\node (l1) at (1,0) {};
\node[above right] at (l1) {$\hilbert_A$};
\node (l3) at (1,-2) {};
\node[above right] at (l3) {$\hilbert_B$};
\node (r1) at (7,0) {};
\node[above left] at (r1) {$\hilbert_C$};
\node (r2) at (7,-1) {};
\node[above left] at (r2) {$\hilbert_D$};
\node (r3) at (7,-2) {};
\node[above left,xshift=-2] at (r3) {$\hilbert_R$};

\node[unitary2] (u1) at (3,-.5) {$U_\Psi$};
\node[unitary2] (u2) at (5,-1.5) {$V$};
\node[fit=(u1)(u2),thick,dashed,draw,inner sep=.3cm] (u) {};
\node[above right] at (u.north west) {$U_\Phi$};

\draw[wire] (l1) to (u1.west |- l1);
\draw[wire] (l3) to (u2.west |- l3);
\draw[wire] (u1.east |- r2) to node[auto,yshift=-2] {$\hilbert_Q$} (u2.west |- r2);
\draw[wire] (u1.east |- r1) to (r1);
\draw[wire] (u2.east |- r2) to (r2);
\draw[wire,-|] (u2.east |- r3) to (r3);
\end{tikzpicture}

\end{centering}
\caption[An illustration of
\lemref{lem:representation}]{\label{fig:equiv}The \strep from
  \lemref{lem:representation} drawn as a circuit.}
\end{figure}

\begin{proof}[Proof of \lemref{lem:representation}]We first show that \eqref{eq:equiv} $\iff$
\eqref{eq:equiv.cj}. Note that \eqnref{eq:equiv} is equivalent to
requiring that for all states
$\psi_A,\psi_B,\psi_C,\varphi_A,\varphi_B,\varphi_C$,
\begin{equation} \label{eq:equiv.1} \bra{\psi_C} \trace[D]{\Phi \left(
      \ketbra{\bar{\psi}_A, \bar{\psi}_B} {\bar{\varphi}_A,
        \bar{\varphi}_B} \right)} \ket{\varphi_C} = \bra{\psi_C} \Psi
  \left( \trace[B]{ \ketbra{\bar{\psi}_A, \bar{\psi}_B}
      {\bar{\varphi}_A, \bar{\varphi}_B} } \right) \ket{\varphi_C}\,.
\end{equation}
The left-hand side of \eqnref{eq:equiv.1} is equal to
\[ \sum_j \bra{\psi_C,j_D} \Phi \left( \ketbra{\bar{\psi}_A,
    \bar{\psi}_B} {\bar{\varphi}_A, \bar{\varphi}_B} \right)
\ket{\varphi_C,j_D}\,, \]
which, when rewritten with the \cjrep, corresponds to the left-hand
side of \eqnref{eq:equiv.cj}. And since
\[\trace[B]{\ketbra{\bar{\psi}_A, \bar{\psi}_B} {\bar{\varphi}_A,
    \bar{\varphi}_B}} = \ketbra{\bar{\psi}_A} {\bar{\varphi}_A}
\braket{\psi_B}{\varphi_B}\,,\]
the right-hand side of \eqnref{eq:equiv.1} is equal to the right-hand
side of \eqnref{eq:equiv.cj}.

Next we show that \eqref{eq:equiv.st} $\implies$
\eqref{eq:equiv}. From the definition of the \strep and
\eqnref{eq:equiv.st} we have for any $\rho_{AB} \in \tcop{\hilbert_{AB}}$,
\begin{align*}
  \trace[D]{\Phi(\rho_{AB})} & = \trace[DR]{( I_C \tensor V ) ( U_\Psi
  \tensor I_B ) \rho_{AB} (\hconj{U}_\Psi \tensor I_B ) ( I_C \tensor
  \hconj{V} )} \\
  & = \trace[QB]{ ( U_\Psi \tensor I_B ) \rho_{AB} (\hconj{U}_\Psi \tensor I_B)} \\
  & = \trace[Q]{U_\Psi \trace[B]{\rho_{AB}} \hconj{U}_\Psi} \\
  & = \Psi \left( \trace[B]{\rho_{AB}} \right)\,,
\end{align*}
where to obtain the second line we used that since
$V : \hilbert_{QB} \to \hilbert_{DR}$ maps $QB$ to $DR$ and the registers $DR$ are
traced out, $V$ and $\hconj{V}$ have no effect on the outcome (see
\figref{fig:equiv} for an illustration of this case).

Finally, we show that \eqref{eq:equiv} $\implies$
\eqref{eq:equiv.st}. To prove this, we need to find the isometry $V$
which satisfies \eqnref{eq:equiv.st}. Consider the map
$\Phi' \coloneqq {\tr_D} \circ \Phi$. Since
$\Phi'(\rho_{AB}) = \trace[DR]{U_\Phi \rho_{AB} \hconj{U}_\Phi}$, the
operator $U_{\Phi}$ \--- which by construction is a \strep of $\Phi$
\--- is also a \strep of $\Phi'$ with ancilla registers $DR$. From
\eqnref{eq:equiv} we have $\Phi' = \Psi \circ \tr_B$, hence
$\Phi'(\rho_{AB}) = \trace[QB]{(U_\Psi \tensor I_B) \rho_{AB}
  (\hconj{U}_\Psi \tensor I_B)}$,
and $U_\Psi \tensor I_B$ is also a \strep of $\Phi'$, with ancilla
$QB$. Furthermore, since $U_\Psi$ is a minimal representation of
$\Psi$, $U_\Psi \tensor I_B$ is a minimal representation of
$\Phi'$. And because any \strep is related to a minimal one by an
isometry on the ancilla~\cite{Pau03}, there must exist
$V : \hilbert_{QB} \to \hilbert_{DR}$ that satisfies
\eqnref{eq:equiv.st}.
\end{proof}

\subsection{A sequence of operators}
\label{sec:representation.isometries}

If we apply \lemref{lem:representation} recursively, we can decompose
any map $\Phi^{\cC}$ into a (finite) sequence of isometries, e.g., 
\begin{equation*} 
  \left\{ V_i : \hilbert_{Q_{i+1}} \tensor
    \cF^{\chi(C_{i})\setminus\chi(\cC_{i+1})}_X \to
    \cF^{\cC_{i}\setminus\cC_{i+1}}_Y \tensor \hilbert_{Q_{i}}
  \right\}_{i = 1}^n\,,\end{equation*}
for any sequence of cuts $\emptyset = \cC_{n+1} \subseteq \cC_{n} \subseteq
\dotsb  \subseteq \cC_2 \subseteq \cC_1 = \cC$.
Here, $\hilbert_{Q_{i+1}}$ can be thought of as the Hilbert space of the
internal memory of the system before unitary $V_i$ has been applied,
and each $V_i$ processes this internal memory as well as the input on
the wire $X$ in positions $\chi(\cC_{i})\setminus\chi(\cC_{i+1})$
to produce an output on the wire $Y$ in positions
$\cC_{i}\setminus\cC_{i+1}$ and the new updated internal memory $\hilbert_{Q_{i}}$.

In this section we show that we can do this an infinite number of
times in such a way that the input and output sets of positions of the
unitaries $V_i$ are disjoint, i.e.,
$\left(\chi(C_{i})\setminus\chi(\cC_{i+1})\right) \cap \left(
  \cC_{i}\setminus\cC_{i+1} \right) = \emptyset$.
We call this a \emph{sequence representation} and illustrate it in
\figref{fig:isometries}.

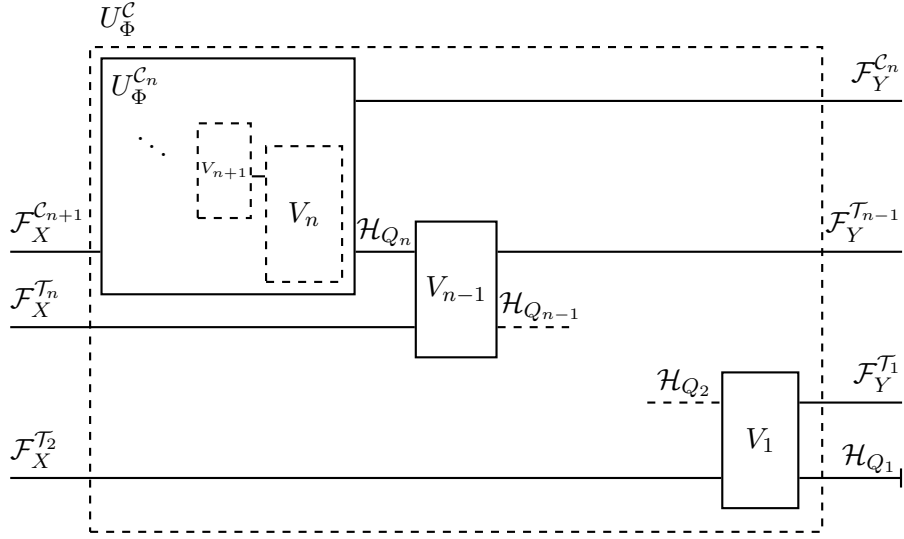
\begin{figure}[tb]
\begin{centering}
\begin{tikzpicture}[
wire/.style={-,thick},
unitary2/.style={draw,thick,minimum width=1cm,minimum height=1.8cm},
unitary2a/.style={draw,inner sep=0,thick,minimum width=\a*28.35,minimum height=1.8*\a*28.35},
invisible2/.style={thick,minimum width=1cm,minimum height=1.8cm}]

\def\a{.7}

\node (l1) at (1,0) {};
\node (l2) at (1,-.5) {};
\node (l3) at (1,-1) {};
\node (l4) at (1,-1.5) {};
\node (l5) at (1,-2) {};
\node[above right] at (l5) {$\cF^{\cC_{n+1}}_X$};
\node (l6) at (1,-2.5) {};
\node (l7) at (1,-3) {};
\node[above right] at (l7) {$\cF^{\T_{n}}_X$};
\node (l8) at (1,-3.5) {};
\node (l9) at (1,-4) {};
\node (l10) at (1,-4.5) {};
\node (l11) at (1,-5) {};
\node[above right] at (l11) {$\cF^{\T_{2}}_X$};

\node (r1) at (13,0) {};
\node[above left] at (r1) {$\cF^{\cC_{n}}_Y$};
\node (r2) at (13,-.5) {};
\node (r3) at (13,-1) {};
\node (r4) at (13,-1.5) {};
\node (r5) at (13,-2) {};
\node[above left] at (r5) {$\cF^{\T_{n-1}}_Y$};
\node (r6) at (13,-2.5) {};
\node (r7) at (13,-3) {};
\node (r8) at (13,-3.5) {};
\node (r9) at (13,-4) {};
\node[above left] at (r9) {$\cF^{\T_1}_Y$};
\node (r10) at (13,-4.5) {};
\node (r11) at (13,-5) {};
\node[yshift=-2,above left,xshift=-2] at (r11) {$\hilbert_{Q_1}$};

\node[invisible2] (u1) at (3,-.5) {$\ddots$};
\node[unitary2a,dashed] (u2a) at (5-.7-.5*\a,-1.5+.4+.25*\a) {\tiny $V_{n+1}$};
\node[unitary2,dashed] (u2) at (5,-1.5) {$V_{n}$};
\node[unitary2] (u3) at (7,-2.5) {$V_{n-1}$};
\node[invisible2] (u4) at (9,-3.5) {};
\node[unitary2] (u5) at (11,-4.5) {$V_1$};
\node[fit=(u1)(u2),thick,draw,inner sep=.15cm] (un) {};
\node[below right] at (un.north west) {$U^{\cC_n}_\Phi$};
\node[fit=(u1)(u5),thick,dashed,draw,inner sep=.3cm] (u) {};
\node[above right] at (u.north west) {$U_\Phi^{\cC}$};

\draw[wire] (l5) to (un.west |- l5);
\draw[wire] (l7) to (u3.west |- l7);
\draw[wire] (l11) to (u5.west |- l11);

\draw[wire] (u2a.east |- r3) to (u2.west |- r3);
\draw[wire] (un.east |- r5) to node[auto,yshift=-2] {$\hilbert_{Q_{n}}$} (u3.west |- r5);
\draw[wire,dashed] (u3.east |- r7) to node[auto,yshift=-2,pos=.6] {$\hilbert_{Q_{n-1}}$} (u4.west |- r7);
\draw[wire,dashed] (u4.east |- r9) to node[auto,yshift=-2] {$\hilbert_{Q_{2}}$} (u5.west |- r9);

\draw[wire] (un.east |- r1) to (r1);
\draw[wire] (u3.east |- r5) to (r5);
\draw[wire] (u5.east |- r9) to (r9);
\draw[wire,-|] (u5.east |- r11) to (r11);


\end{tikzpicture}

\end{centering}
\caption[Decomposing a causal box]{\label{fig:isometries}A map
  $\Phi^{\cC}$ is decomposed into a sequence of isometries $V_i$.}
\end{figure}

\begin{deff}[Sequence representation]
\label{def:sequencerepresentation}
Let $\dotsb \subseteq \cC_i \subseteq \dotsb \subseteq \cC_1 = \cC$ be
an infinite sequence of cuts such that
$\bigcap_{i=1}^\infty \cC_i = \emptyset$, and let
$\T_{i} \coloneqq \cC_i \setminus \cC_{i+1}$. A \emph{sequence
  representation} of a map
$\Phi^\cC : \tcop{\cF^{\chi(\cC)}_X} \to \tcop{\cF^{\cC}_Y}$ is given
by such a set of cuts $\{\cC_i\}_{i=1}^\infty$ along with a set of
operators
\begin{equation*}
 \left\{ V_i : \hilbert_{Q_{i+1}} \tensor \cF^{\T_{i+1}}_X \to
  \cF^{\T_i}_Y \tensor \hilbert_{Q_{i}} \right\}_{i = 1}^\infty \,, 
\end{equation*}
such that for all $n \geq 2$,
  \begin{equation} \label{eq:sequencerepresentation}
  U^{\cC_{1}}_\Phi = \left( \prod_{i=1}^{n-1}I^{\cC_{i+1}}_Y \tensor
    V_i \tensor I^{\cC_2 \setminus \cC_{i+1}}_X\right)
  \left( U^{\cC_{n}}_{\Phi} \tensor I^{\cC_2 \setminus \cC_{n+1}}_X \right) \,,
\end{equation} where $U^{\cC_i}_\Phi$ is a minimal
\strep of $\Phi^{\cC_i}$.
\end{deff}

A sequence representation for any map $\Phi^\cC$ can be obtained by
defining $\cC_i \coloneqq \chi^{i-1}(\cC)$. We then immediately have
that the set of input positions to the operator $V_i$ has an empty
intersection with the output positions, since
$\chi(C_i)\setminus\chi(\cC_{i+1}) = \T_{i+1}$ and
$\T_{i+1} \cap \T_i = \emptyset$.

\begin{prop} \label{prop:sequencerepresentation} For every causal box
  $\Phi = \left\{\Phi^\cC\right\}_{\cC \in \bcut}$ and every
  $\cC \in \bcut$, there exists a sequence representation of
  $\Phi^\cC$.
\end{prop}

\begin{proof} We fix $\Phi$ and $\cC$, and define
  $\cC_i \coloneqq \chi^{i-1}(\cC)$ and
  $\T_i \coloneqq \cC_i \setminus \cC_{i+1}$. Note that $\cC_1 = \cC$,
  $\T_i \cap \T_j = \emptyset$ for $i \neq j$, and from
  \eqnref{eq:causality.finite} we have
  $\bigcap_{i=1}^\infty \cC_i = \emptyset$.  Since by construction of
  the cuts, $\chi(\cC_n) = \cC_{n+1}$, the output at all positions
  $t \in \cC_n$ can be computed from the input in positions
  $\cC_{n+1}$, i.e., $\Phi^{\cC_n}$ is a map
  \[ \Phi^{\cC_n} : \tcop{\cF^{\cC_{n+1}}_X} \to \tcop{\cF^{\cC_{n}}_Y}.\]

  Plugging $\Phi^{\cC_1}$ and $\Phi^{\cC_{2}}$ in
  \lemref{lem:representation}, i.e, $\hilbert_A = \cF^{\cC_{3}}_X$,
  $\hilbert_B = \cF^{\T_{2}}_X$, $\hilbert_C = \cF^{\cC_{2}}_Y$
  and $\hilbert_D = \cF^{\T_1}_Y$, we obtain
  \[ U^{\cC_{1}}_\Phi = \left( I^{\cC_{2}}_Y \tensor V_1 \right)
  \left( U^{\cC_{2}}_{\Phi} \tensor I^{\T_2}_X \right) \,.\]
  Repeating this recursively for $i$ going from $2$ to $n$ results in
  the decomposition on the right-hand side of
  \eqnref{eq:sequencerepresentation}.
\end{proof}

This proposition can easily be extended to subnormalized causal boxes
by appending to every isometry $V_i$ a projector on the vacuum state
$\vacuum$ of the additional wire $R$. This is illustrated in
\figref{fig:isometries.sub}.

\begin{figure}[tb]
\begin{centering}
\begin{tikzpicture}[
wire/.style={-,thick},
unitary2a/.style={draw,thick,inner sep=1pt,minimum width=2.16cm,minimum height=2.16cm},
unitary2/.style={draw,thick,inner sep=1pt,minimum width=1cm,minimum height=2.16cm},
unitary1/.style={draw,thick,inner sep=1pt,minimum width=1cm,minimum height=.72cm},
invisible2/.style={thick,inner sep=1pt,minimum width=1cm,minimum height=2.16cm}]

\def\t{.72}

\small

\node (l1) at (1,0) {};
\node (l2) at (1,-\t) {};
\node[above right,yshift=-2] at (l2) {$\cF^{\cC_{n+1}}_X$};
\node (l3) at (1,-2*\t) {};
\node (l4) at (1,-3*\t) {};
\node[above right,yshift=-2] at (l4) {$\cF^{\T_{n}}_X$};
\node (l5) at (1,-4*\t) {};
\node[above right,yshift=-2] at (l5) {};
\node (l6) at (1,-5*\t) {};
\node (l7) at (1,-6*\t) {};
\node (l8) at (1,-7*\t) {};
\node[above right,yshift=-2] at (l8) {$\cF^{\T_{2}}_X$};
\node (l9) at (1,-8*\t) {};

\node (r1) at (13.5,0) {};
\node[above left,yshift=-2] at (r1) {$\vacuum^{\cC_{n}}_R$};
\node (r2) at (13.5,-\t) {};
\node[above left,yshift=-2] at (r2) {$\cF^{\cC_{n}}_Y$};
\node (r3) at (13.5,-2*\t) {};
\node[above left,yshift=-2] at (r3) {$\vacuum^{\T_{n-1}}_R$};
\node (r4) at (13.5,-3*\t) {};
\node[above left,yshift=-2] at (r4) {$\cF^{\T_{n-1}}_Y$};
\node (r5) at (13.5,-4*\t) {};
\node (r6) at (13.5,-5*\t) {};
\node (r7) at (13.5,-6*\t) {};
\node[above left,yshift=-2] at (r7) {$\vacuum^{\T_1}_R$};
\node (r8) at (13.5,-7*\t) {};
\node[above left,yshift=-2] at (r8) {$\cF^{\T_1}_Y$};
\node (r9) at (13.5,-8*\t) {};
\node[yshift=-2,above left,xshift=-2] at (r9) {$\hilbert_{Q_1}$};

\node[unitary2a] (u1) at (3.4,-\t) {\normalsize$\hat{U}^{\cC_n}_{\hat{\Phi}}$};
\node[unitary2] (u2) at (5.8,-3*\t) {$\hat{V}_{n-1}$};
\node[invisible2] (u3) at (7.7,-5*\t) {};
\node[unitary2] (u4) at (9.6,-7*\t) {$\hat{V}_1$};
\node[invisible2] (u5) at (11.5,-7*\t) {};
\node[unitary1] (p1) at (u2 |- r1) {$P^\Omega_{R_{\geq n}}$};
\node[unitary1] (p2) at (u3 |- r3) {$P^\Omega_{R_{n-1}}$};
\node[unitary1] (p4) at (u5 |- r7) {$P^\Omega_{R_1}$};
\node[fit=(u1)(u5),thick,dashed,draw,inner sep=.2cm] (u) {};
\node[above right] at (u.north west) {$U_{\Phi}^{\cC}$};

\draw[wire] (l2) to (u1.west |- l2);
\draw[wire] (l4) to (u2.west |- l4);
\draw[wire] (l8) to (u4.west |- l8);

\draw[wire] (u1.east |- r3) to node[auto,yshift=-2] {$\hilbert_{Q_n}$} (u2.west |- r3);
\draw[wire,dashed] (u2.east |- r5) to node[auto,yshift=-2,pos=.6] {$\hilbert_{Q_{n-1}}$} (u3.west |- r5);
\draw[wire,dashed] (u3.east |- r7) to node[auto,yshift=-2] {$\hilbert_{Q_{2}}$} (u4.west |- r7);
\draw[wire] (u1.east |- r1) to node[auto,yshift=-2] {$\cF^{\cC_n}_{R}$} (p1.west |- r1);
\draw[wire] (u2.east |- r3) to node[auto,yshift=-2] {$\cF^{\T_{n-1}}_{R}$} (p2.west |- r3);
\draw[wire] (u4.east |- r7) to node[auto,yshift=-2] {$\cF^{\T_1}_{R}$} (p4.west |- r7);

\draw[wire] (p1.east |- r1) to (r1);
\draw[wire] (u1.east |- r2) to (r2);
\draw[wire] (p2.east |- r3) to (r3);
\draw[wire] (u2.east |- r4) to (r4);
\draw[wire] (p4.east |- r7) to (r7);
\draw[wire] (u4.east |- r8) to (r8);
\draw[wire,-|] (u4.east |- r9) to (r9);
\end{tikzpicture}

\end{centering}
\caption[Decomposing a subnormalized causal
box]{\label{fig:isometries.sub}A trace\-/decreasing map $\Phi^\cC$ is
  decomposed into a sequence of isometries $\hat{V}_i$ followed
  by a projection $P^\Omega_{R_i}$ on the vacuum state of the
  wire $R$ in the subset $\T_i$.}
\end{figure}
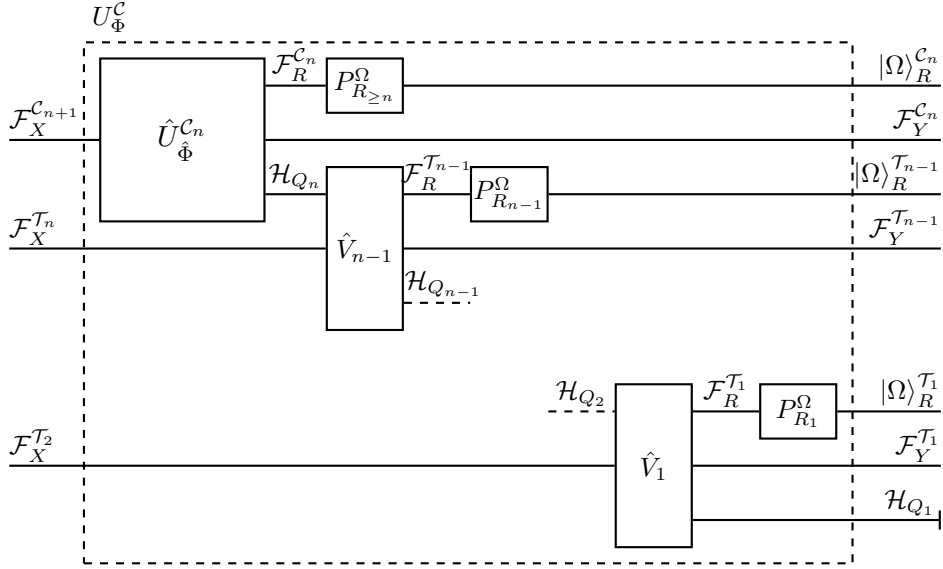

\begin{cor} \label{cor:sequencerepresentation} For every subnormalized
  causal box $\Phi = \left\{\Phi^\cC\right\}_{\cC \in \bcut}$ and
  every $\cC \in \bcut$, there exists a sequence representation of
  $\Phi^\cC$.
\end{cor}

\begin{proof}
  Let
  $\hat{\Phi} = \{\hat{\Phi}^\cC : \tcop{\hilbert_X} \to
  \tcop{\hilbert_{RY}}\}_{\cC \in \bcut}$
  be the corresponding normalized box, and let
  $\{\hat{V}_i\}_{i =1}^\infty$ be the sequence representation for
  $\hat{\Phi}^\cC$. \eqnref{eq:sequencerepresentation} is satisfied
  for
\begin{equation*}
V_i = \left(\bra{\Omega}^{\T_i}_R \tensor I^{\T_{i}}_Y \tensor I_{Q_i}\right)
\hat{U}_i\,. \qedhere
\end{equation*}
\end{proof}

\section{Composing causal boxes}
\label{sec:network}

Two causal boxes can be combined by ``plugging'' output wires into
input wires, resulting in a new causal box. Unlike for ordered
networks captured by combs~\cite{GW07,Gut12,CDP09,Har11,Har12,Har15},
where cycles are forbidden, here two systems $\Phi$ and $\Psi$ can be
connected with wires going both from $\Phi$ to $\Psi$ and $\Psi$ to
$\Phi$, as illustrated in \figref{fig:box2}.

Connecting systems can be decomposed in two steps. First $\Phi$ and
$\Psi$ are composed in parallel, resulting in a new system $\Gamma =
\Phi \| \Psi$, whose input and output ports are the union of the ports
from $\Phi$ and $\Psi$. Then loops are applied to $\Gamma$, connecting
its own output and input ports. These two steps are defined in
\secsref{sec:network.parallel} and \ref{sec:network.loops},
respectively.

In \secref{sec:network.composition} we combine these two steps to
define the composition of two systems, $\Phi \connect{P} \Psi$,
introduced in \secref{sec:theory}. We then prove in
\thmref{thm:closure} and \thmref{thm:comp.order.indep} that it
satisfies closure and composition order independence.

\subsection{Parallel composition}
\label{sec:network.parallel}

We first define parallel composition in \defref{def:parallel}. Then we
prove in \propref{prop:parallel} that the resulting system is still a
valid causal box.

\begin{deff}[Parallel composition]
\label{def:parallel}
Let $\Phi = \left\{\Phi^\cC\right\}_{\cC \in \bcut}$ and
$\Psi = \left\{\Psi^\cC\right\}_{\cC \in \bcut}$ be (possibly subnormalized)
$(d_A,d_C)$- and $(d_B,d_D)$\-/causal boxes. The parallel composition
of the two is defined as the $(d_A+d_B,d_C+d_D)$\-/causal
box\footnote{Recall that the dimension of a wire $d_A$ is the
  dimension of the messages on the wire, not the dimension of the wire
  Hilbert space (which is infinite). By \eqnref{eq:fockisomorphism}
  the dimension of the tensor product of two wires is the \emph{sum}
  of their dimensions.}
\[\Gamma \coloneqq \left\{ \Phi^\cC \tensor \Psi^\cC\right\}_{\cC \in
  \bcut}\,,\] which we denote $\Gamma = \Phi \| \Psi$.
\end{deff}

\begin{prop}
  \label{prop:parallel} If
  $\Phi = \{ \Phi^\cC : \tcop{\cF^\T_A} \to \tcop{\cF^{\cC}_C} \}_{\cC
    \in \bcut}$
  and
  $\Psi = \{ \Psi^\cC : \tcop{\cF^\T_B} \to \tcop{\cF^{\cC}_D}\}_{\cC
    \in \bcut}$
  are two (sub)normalized causal boxes, then so is
  $\Gamma = \Phi \| \Psi$.
\end{prop}

\begin{proof}
  We first consider normalized causal boxes. To prove that
  $\Phi \| \Psi$ is a valid causal box, we need to find a causality
  function $\chi_\Gamma : \cut \to \cut$ such that
\[ \Phi^{\cC} \tensor \Psi^{\cC} = \left( \Phi^{\cC} \tensor
  \Psi^{\cC} \right) \circ \tr_{\T \setminus \chi_\Gamma(\cC)}\,.\]
We prove in \lemref{lem:causality.parallel} that the function
$\chi_\Gamma(\cC) \coloneqq \chi_\Phi(\cC) \cup \chi_\Psi(\cC)$
satisfies the requirements of a causality function given in
\defref{def:causality.function}. And we have
\begin{align*}
  \left( \Phi^{\cC} \tensor \Psi^{\cC} \right) \circ
  \tr_{\T \setminus \chi_{\Gamma}(\cC)}
  & = \left( \Phi^{\cC} \circ
    \tr_{A^{\T \setminus \chi_{\Gamma}(\cC)}} \right) \tensor
    \left( \Psi^{\cC} \circ \tr_{B^{\T \setminus \chi_{\Gamma}(\cC)}} \right) \\
  & = \left( \Phi^{\cC} \circ \tr_{A^{\T \setminus \chi_{\Phi}(\cC)}} \right) \tensor
    \left( \Psi^{\cC} \circ \tr_{B^{\T \setminus \chi_{\Psi}(\cC)}} \right) \\
  & = \Phi^{\cC} \tensor \Psi^{\cC}\,,
\end{align*}
where $\tr_{A^{\T \setminus \chi_{\Gamma}(\cC)}}$ and
$\tr_{B^{\T \setminus \chi_{\Gamma}(\cC)}}$ trace out the inputs in positions
$\T \setminus \chi_\Gamma(\cC)$ on the $A$ and $B$ wires, respectively.

For subnormalized boxes $\Phi$ and $\Psi$, let
$\hat{\Phi} = \left\{ \hat{\Phi}^{\cC} : \tcop{\cF^\T_A} \to \tcop{\cF^{\cC}_{RC}}
\right\}$ and $\hat{\Psi} = \left\{ \hat{\Psi}^{\cC} : \tcop{\cF^\T_B} \to \tcop{\cF^{\cC}_{SD}}
\right\}$ be their normalized counterparts. Then
\begin{align*}
  P^{\Omega}_{RS} \circ \left(\hat{\Phi}^{\cC} \tensor \hat{\Psi}^{\cC}\right)
  & = \left(P^{\Omega}_{R} \circ \hat{\Phi}^{\cC}\right) \tensor
    \left(P^{\Omega}_{S} \circ \hat{\Psi}^{\cC}\right) \\
  & = \Phi^\cC \tensor \Psi^\cC\,,
\end{align*}
where $P^{\Omega}_R$ projects the wire $R$ on the vacuum state
$\vacuum$.
\end{proof}

\subsection{Loops}
\label{sec:network.loops}

In this section we first give an intuitive explanation of what it
means to put a loop from an output wire to an input wire of a causal
box and show how to capture this mathematically. We then provide the
formal definition in \defref{def:loop}.  In \propref{prop:loop} we
prove that the resulting system after the loop has been applied is a
new valid causal box.

Before considering the case of a loop on an arbitrary causal box, we
first look at the simpler case of a map
$\Phi : \tcop{\hilbert_{AB}} \to \tcop{\hilbert_{CD}}$ for which the
output on $C$ does not depend on the input on $B$. A \strep of such a
system is illustrated on the left in \figref{fig:connecting}. Here, it
is clear what happens when $C$ is looped back to $B$: first $U_\Psi$
is applied to the input on wire $A$ to obtain the output on $C$, then
$V$ is applied to both $C$ and the internal state of the system to
produce the output on $D$, as depicted on the right in
\figref{fig:connecting}.

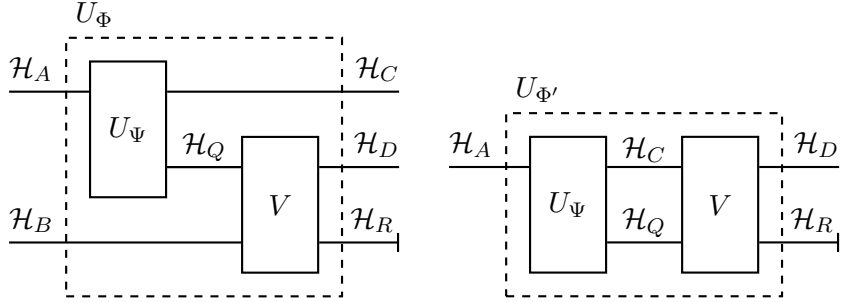
\begin{figure}[tb]
\begin{centering}
\begin{tikzpicture}[
wire/.style={-,thick},
unitary2/.style={draw,thick,minimum width=1cm,minimum height=1.8cm}]

\node (l1) at (1.3,0) {};
\node[above right] at (l1) {$\hilbert_A$};
\node (l3) at (1.3,-2) {};
\node[above right] at (l3) {$\hilbert_B$};
\node (r1) at (6.7,0) {};
\node[above left] at (r1) {$\hilbert_C$};
\node (r2) at (6.7,-1) {};
\node[above left] at (r2) {$\hilbert_D$};
\node (r3) at (6.7,-2) {};
\node[above left,xshift=-2] at (r3) {$\hilbert_R$};

\node[unitary2] (u1) at (3,-.5) {$U_\Psi$};
\node[unitary2] (u2) at (5,-1.5) {$V$};
\node[fit=(u1)(u2),thick,dashed,draw,inner sep=.3cm] (u) {};
\node[above right] at (u.north west) {$U_\Phi$};

\draw[wire] (l1) to (u1.west |- l1);
\draw[wire] (l3) to (u2.west |- l3);
\draw[wire] (u1.east |- r2) to node[auto,yshift=-2] {$\hilbert_Q$} (u2.west |- r2);
\draw[wire] (u1.east |- r2) to (u2.west |- r2);
\draw[wire] (u1.east |- r1) to (r1);
\draw[wire] (u2.east |- r2) to (r2);
\draw[wire,-|] (u2.east |- r3) to (r3);
\end{tikzpicture}
\begin{tikzpicture}[
wire/.style={-,thick},
unitary2/.style={draw,thick,minimum width=1cm,minimum height=1.8cm}]

\node (l1) at (1.3,0) {};
\node[above right] at (l1) {$\hilbert_A$};
\node (r1) at (6.7,0) {};
\node[above left] at (r1) {$\hilbert_D$};
\node (r2) at (6.7,-1) {};
\node[above left,xshift=-2] at (r2) {$\hilbert_R$};

\node[unitary2] (u1) at (3,-.5) {$U_\Psi$};
\node[unitary2] (u2) at (5,-.5) {$V$};
\node[fit=(u1)(u2),thick,dashed,draw,inner sep=.3cm] (u) {};
\node[above right] at (u.north west) {$U_{\Phi'}$};

\draw[wire] (l1) to (u1.west |- l1);
\draw[wire] (u1.east |- r1) to node[auto,yshift=-2] {$\hilbert_C$} (u2.west |- r1);
\draw[wire] (u1.east |- r2) to node[auto,yshift=-2] {$\hilbert_Q$} (u2.west |- r2);
\draw[wire] (u2.east |- r1) to (r1);
\draw[wire,-|] (u2.east |- r2) to (r2);
\end{tikzpicture}

\end{centering}
\caption[A loop on a causally ordered
system]{\label{fig:connecting}Connecting $C$ to $B$ in the system
  depicted on the left results in the system depicted on the right.}
\end{figure}

For the more general case, we make use of
\propref{prop:sequencerepresentation}, which shows that all causal
boxes admit a sequence representation, i.e., every map $\Phi^\cC$ can
be decomposed into an infinite sequence of isometries whose input and
output positions are disjoint. In \figref{fig:loop1} we reproduce
\figref{fig:isometries}, but with two input and output wires. Here,
the isometries $V_i$ for $i \geq n$ have been grouped together as one
operator $U^{\cC_n}_\Phi$ so that $U^{\cC}_\Phi$ consists of a finite
sequence of operators. We apply the same reasoning as above to this
system, i.e., we connect the output on wire $C$ in positions $\T_i$
for $i \leq n$ to the input on wire $B$ of $V_{i-1}$. This results in
a map $\Psi^{\cC}_n$, which corresponds to the system $\Phi^\cC$ with
all outputs on wire $C$ in positions $\bigcup_{i=2}^n \T_i$ looped back
to wire $B$, and has the \strep drawn in \figref{fig:loop2}.

\begin{figure}[tbp]
\begin{subfigure}{\textwidth}
\begin{centering}
\begin{tikzpicture}[
wire/.style={-,thick},
unitary2/.style={draw,thick,minimum width=1cm,minimum height=2.16cm},
invisible2/.style={thick,minimum width=1cm,minimum height=2.16cm},
unitary2a/.style={draw,thick,minimum width=2cm,minimum height=2.16cm}]

\def\t{.72}

\small

\node (l1) at (1,0) {};
\node (l2) at (1,-\t) {};

\def\t{.72}

\node (l1) at (.1,0) {};
\node (l2) at (.1,-\t) {};
\node[above right,yshift=-2] at (l2) {\small$\cF^{\cC_{n+1}}_A$};
\node (l3) at (.1,-2*\t) {};
\node[above right,yshift=-2] at (l3) {\small$\cF^{\cC_{n+1}}_B$};
\node (l4) at (.1,-3*\t) {};
\node[above right,yshift=-2] at (l4) {\small$\cF^{\T_n}_A$};
\node (l5) at (.1,-4*\t) {};
\node[above right,yshift=-2] at (l5) {\small$\cF^{\T_n}_B$};
\node (l6) at (.1,-5*\t) {};
\node (l7) at (.1,-6*\t) {};
\node (l8) at (.1,-7*\t) {};
\node[above right,yshift=-2] at (l8) {\small$\cF^{\T_{3}}_A$};
\node (l9) at (.1,-8*\t) {};
\node[above right,yshift=-2] at (l9) {\small$\cF^{\T_{3}}_B$};
\node (l10) at (.1,-9*\t) {};
\node[above right,yshift=-2] at (l10) {\small$\cF^{\T_{2}}_A$};
\node (l11) at (.1,-10*\t) {};
\node[above right,yshift=-2] at (l11) {\small$\cF^{\T_{2}}_B$};

\node (r1) at (12.3,0) {};
\node[above left,yshift=-2] at (r1) {\small$\cF^{\cC_n}_C$};
\node (r2) at (12.3,-\t) {};
\node[above left,yshift=-2] at (r2) {\small$\cF^{\cC_n}_D$};
\node (r3) at (12.3,-2*\t) {};
\node[above left,yshift=-2] at (r3) {\small$\cF^{\T_{n-1}}_C$};
\node (r4) at (12.3,-3*\t) {};
\node[above left,yshift=-2] at (r4) {\small$\cF^{\T_{n-1}}_D$};
\node (r5) at (12.3,-4*\t) {};
\node (r6) at (12.3,-5*\t) {};
\node (r7) at (12.3,-6*\t) {};
\node[above left,yshift=-2] at (r7) {\small$\cF^{\T_2}_C$};
\node (r8) at (12.3,-7*\t) {};
\node[above left,yshift=-2] at (r8) {\small$\cF^{\T_2}_D$};
\node (r9) at (12.3,-8*\t) {};
\node[yshift=-2,above left] at (r9) {\small$\cF^{\T_1}_C$};
\node (r10) at (12.3,-9*\t) {};
\node[yshift=-2,above left] at (r10) {\small$\cF^{\T_1}_D$};
\node (r11) at (12.3,-10*\t) {};
\node[yshift=-2,above left,xshift=-2] at (r11) {\small$\hilbert_{Q_1}$};

\node[unitary2a] (u1) at (2.5,-\t) {$U^{\cC_n}_\Phi$};
\node[unitary2] (u2) at (4.85,-3*\t) {$V_{n-1}$};
\node[invisible2] (u3) at (6.7,-5*\t) {};
\node[unitary2] (u4) at (8.55,-7*\t) {$V_2$};
\node[unitary2] (u5) at (10.4,-9*\t) {$V_1$};
\node[fit=(u1)(u5),thick,dashed,draw,inner sep=.2cm] (u) {};
\node[above right] at (u.north west) {$U_{\Phi}^{\cC}$};

\draw[wire] (l2) to (u1.west |- l2);
\draw[wire] (l3) to (u1.west |- l3);
\draw[wire] (l4) to (u2.west |- l4);
\draw[wire] (l5) to (u2.west |- l5);
\draw[wire] (l8) to (u4.west |- l8);
\draw[wire] (l9) to (u4.west |- l9);
\draw[wire] (l10) to (u5.west |- l10);
\draw[wire] (l11) to (u5.west |- l11);

\draw[wire] (u1.east |- r3) to node[auto,yshift=-2] {\small$\hilbert_{Q_n}$} (u2.west |- r3);
\draw[wire,dashed] (u2.east |- r5) to node[auto,yshift=-2,pos=.6] {\small$\hilbert_{Q_{n-1}}$} (u3.west |- r5);
\draw[wire,dashed] (u3.east |- r7) to node[auto,yshift=-2] {\small$\hilbert_{Q_{3}}$} (u4.west |- r7);
\draw[wire] (u4.east |- r9) to node[auto,yshift=-2] {\small$\hilbert_{Q_2}$} (u5.west |- r9);

\draw[wire] (u1.east |- r1) to (r1);
\draw[wire] (u1.east |- r2) to (r2);
\draw[wire] (u2.east |- r3) to (r3);
\draw[wire] (u2.east |- r4) to (r4);
\draw[wire] (u4.east |- r7) to (r7);
\draw[wire] (u4.east |- r8) to (r8);
\draw[wire] (u5.east |- r9) to (r9);
\draw[wire] (u5.east |- r10) to (r10);
\draw[wire,-|] (u5.east |- r11) to (r11);
\end{tikzpicture}

\end{centering}
\caption[Before]{\label{fig:loop1}By
  \propref{prop:sequencerepresentation} any map
  $\Phi^{\cC} : \tcop{\cF^\T_{AB}} \to \tcop{\cF^\cC_{CD}}$ can be
  decomposed into a sequence of isometries $V_i$ with disjoint input
  and output positions, as depicted here.}
\end{subfigure}

\vspace{9pt}

\begin{subfigure}{\textwidth}
\begin{centering}
\begin{tikzpicture}[
wire/.style={-,thick},
unitary2/.style={draw,thick,inner sep=2pt,minimum width=1cm,minimum height=2.16cm},
invisible2/.style={thick,minimum width=1cm,minimum height=2.16cm},
unitary2a/.style={draw,thick,minimum width=2cm,minimum height=2.88cm}]

\def\t{.72}

\node (l1) at (.15,0) {};
\node (l2) at (.15,-\t) {};
\node (l3) at (.15,-2*\t) {};
\node[above right,yshift=-2] at (l3) {\small$\cF^{\cC_{n+1}}_A$};
\node (l4) at (.15,-3*\t) {};
\node[above right,yshift=-2] at (l4) {\small$\cF^{\cC_{n+1}}_B$};
\node (l5) at (.15,-4*\t) {};
\node[above right,yshift=-2] at (l5) {\small$\cF^{\T_n}_A$};
\node (l6) at (.15,-5*\t) {};
\node (l7) at (.15,-6*\t) {};
\node[above right,yshift=-2] at (l7) {\small$\cF^{\T_{3}}_A$};
\node (l8) at (.15,-7*\t) {};
\node[above right,yshift=-2] at (l8) {\small$\cF^{\T_{2}}_A$};

\node (r1) at (12.35,0) {};
\node[above left,yshift=-2] at (r1) {\small$\cF^{\cC_{n+1}}_C$};
\node (r2) at (12.35,-\t) {};
\node[above left,yshift=-2] at (r2) {\small$\cF^{\cC_n}_D$};
\node (r3) at (12.35,-2*\t) {};
\node[above left,yshift=-2] at (r3) {\small$\cF^{\T_{n-1}}_D$};
\node (r4) at (12.35,-3*\t) {};
\node (r5) at (12.35,-4*\t) {};
\node[yshift=-2,above left] at (r5) {\small$\cF^{\T_{2}}_D$};
\node (r6) at (12.35,-5*\t) {};
\node[above left,yshift=-2] at (r6) {\small$\cF^{\T_{1}}_D$};
\node (r7) at (12.35,-6*\t) {};
\node[yshift=-2,above left,xshift=-2] at (r7) {\small$\hilbert_{Q_1}$};
\node (r8) at (12.35,-7*\t) {};
\node[yshift=-2,above left,xshift=-2] at (r8) {\small$\cF^{\T_1}_C$};

\node[unitary2a] (u1) at (2.5,-1.5*\t) {$U^{\cC_n}_\Phi$};
\node[unitary2] (u2) at (4.85,-3*\t) {$V_{n-1}$};
\node[invisible2] (u3) at (6.7,-4*\t) {};
\node[unitary2] (u4) at (8.55,-5*\t) {$V_2$};
\node[unitary2] (u5) at (10.4,-6*\t) {$V_1$};
\node[fit=(u1)(u5),thick,dashed,draw,inner sep=.2cm] (u) {};
\node[above right] at (u.north west) {$U_{\Psi_n}^{\cC}$};

\draw[wire] (l3) to (u1.west |- l3);
\draw[wire] (l4) to (u1.west |- l4);
\draw[wire] (l5) to (u2.west |- l5);
\draw[wire] (l7) to (u4.west |- l7);
\draw[wire] (l8) to (u5.west |- l8);

\draw[wire] (u1.east |- r3) to node[auto,yshift=-2] {\small$\hilbert_{Q_n}$} (u2.west |- r3);
\draw[wire] (u1.east |- r4) to node[auto,yshift=-2] {\small$\cF^{\T_n}_C$} (u2.west |- r4);
\draw[wire,dashed] (u2.east |- r4) to node[auto,yshift=-2,pos=.6] {\small$\hilbert_{Q_{n-1}}$} (u3.west |- r4);
\draw[wire,dashed] (u2.east |- r5) to node[auto,yshift=-2,pos=.6] {\small$\cF^{\T_{n-1}}_C$} (u3.west |- r5);
\draw[wire,dashed] (u3.east |- r5) to node[auto,yshift=-2] {\small$\hilbert_{Q_{3}}$} (u4.west |- r5);
\draw[wire,dashed] (u3.east |- r6) to node[auto,yshift=-2] {\small$\cF^{\T_{3}}_C$} (u4.west |- r6);
\draw[wire] (u4.east |- r6) to node[auto,yshift=-2] {\small$\hilbert_{Q_{2}}$} (u5.west |- r6);
\draw[wire] (u4.east |- r7) to node[auto,yshift=-2] {\small$\cF^{\T_{2}}_C$} (u5.west |- r7);

\draw[wire] (u1.east |- r1) to (r1);
\draw[wire] (u1.east |- r2) to (r2);
\draw[wire] (u2.east |- r3) to (r3);
\draw[wire] (u4.east |- r5) to (r5);
\draw[wire] (u5.east |- r6) to (r6);
\draw[wire,-|] (u5.east |- r7) to (r7);
\draw[wire,-|] (u5.east |- r8) to (r8);
\end{tikzpicture}

\end{centering}
\caption[After]{\label{fig:loop2}To (partially) connect $C$ to $B$,
  the output from $V_i$ in positions $\T_i$ on wire $C$ is input to
  $V_{i+1}$.}
\end{subfigure}

\caption[A partial loop on a decomposed causal
box]{\label{fig:loop}Putting a loop from the output wire $C$ to the
  input wire $B$ in positions $\bigcup_{i=1}^n \T_i$ in the system
  depicted in \figref{fig:loop1} results in the system drawn in
  \figref{fig:loop2}. Looping all of $C$ back to $B$ corresponds to
  taking the limit as $n \to \infty$.}
\end{figure}
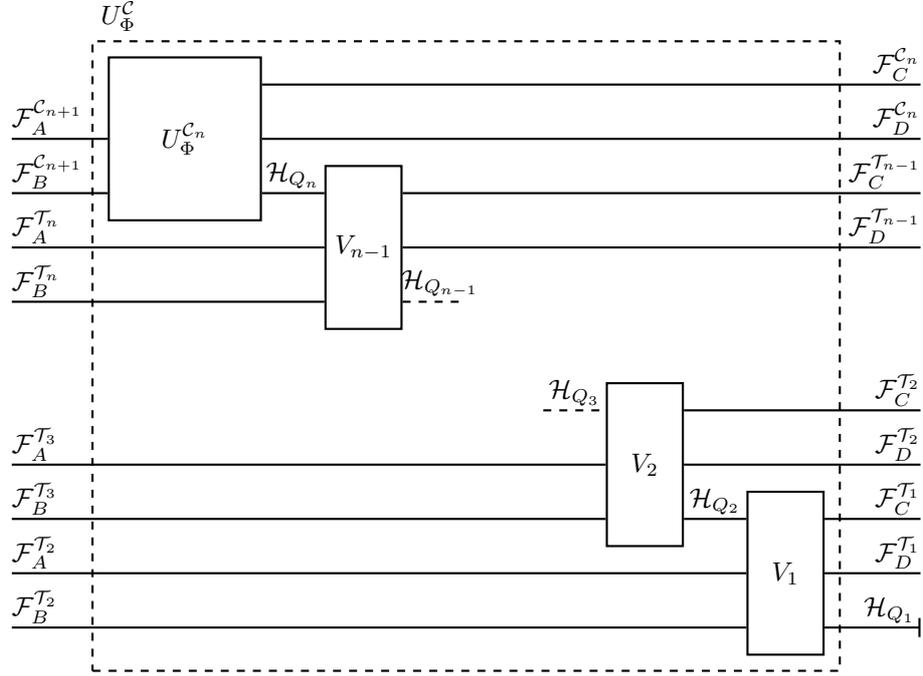
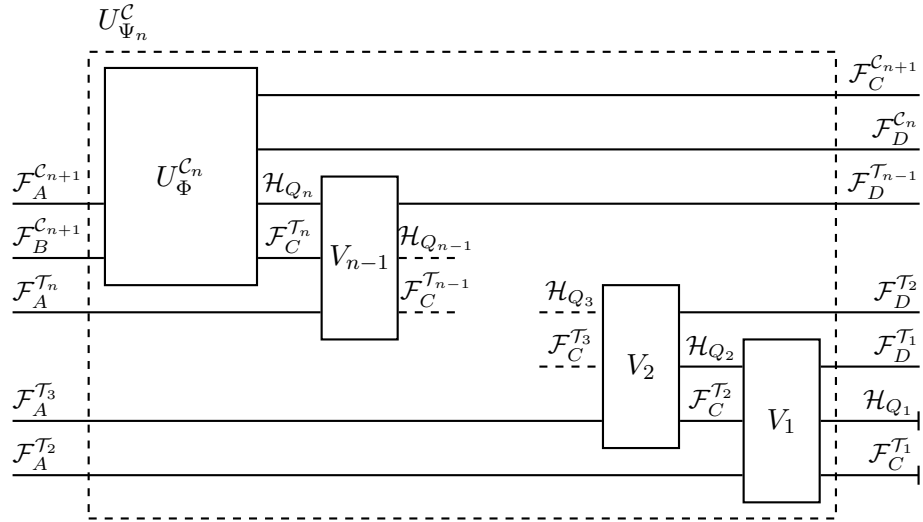

The causal box $\Psi = \Phi^{(C \hookrightarrow B)}$, resulting from
looping all of $C$ to $B$, can be defined as the limit of these maps
when $n \to \infty$, i.e.,
$\Psi \coloneqq \left\{\Psi^{\cC}\right\}_{\cC \in \bcut}$ where
\begin{equation}
\label{eq:loop.limit}
\Psi^{\cC} \coloneqq \lim_{n \to \infty} \Psi^{\cC}_n \,.
\end{equation} This is however a rather
inconvenient definition with which to work. Instead, we provide a
closed formula for the maps $\{\Psi^{\cC}\}_{\cC}$ in
\defref{def:loop}, \eqnref{eq:loop}. We prove
in \appendixref{app:loops} that this definition is equivalent
to \eqnref{eq:loop.limit}.

In the case of a classical causal box given
by a set of conditional probability distributions $\{P^{\cC}_{CD|AB}\}_{\cC}$, this closed
formula for a loop reduces to
\[ Q^\cC_{D|A}(d|a) = \sum_c P^{\cC}_{CD|AB}(c,d|a,c)\,,\]
i.e., the new system with conditional probability distributions
$\{Q^\cC_{D|A}\}_{\cC}$ is obtained from the old system by inputing on
wire $B$ the value $c$ that is output on wire $C$. \eqnref{eq:loop} is
a generalization of this to the quantum case, which uses the
\cjrep.\footnote{The terms in \eqnref{eq:loop} are absolutely
  convergent (as proven in \appendixref{app:loops}). The order of the
  summation is thus not relevant and hence not specified.} In
\remref{rem:nat.rep} we give an equivalent formulation of a loop as a
partial trace using the \natrep.\footnote{We define the \natrep of a
  map in \appendixref{app:rep.natural}.}

\begin{deff}[Loop]
\label{def:loop}
Let
$\Phi = \{\Phi^{\cC} : \tcop{\cF^{\T}_{AB}} \to
\tcop{\cF^{\cC}_{CD}}\}_{\cC \in \bcut}$
be a $(d_A+d_B,d_C+d_D)$\-/causal box with $d_B = d_C$. Let
$R^{\cC}_{\Phi}(\cdot;\cdot)$ be the \cjrep of $\Phi^{\cC}$. Let
$\{\ket{k_C}\}_h$ and $\{\ket{\ell_C}\}_h$ be any orthonormal bases of
$\cF^{\cC}_C$ and let $\{\ket{k_B}\}_k$ and $\{\ket{\ell_B}\}_\ell$
denote the corresponding bases of $\cF^{\cC}_B$, i.e., for all $k$ and
$\ell$, $\ket{k_C} \cong \ket{k_B}$ and
$\ket{\ell_C} \cong \ket{\ell_B}$. The new system resulting from putting
a loop from the output wire $C$ to the input wire $B$,
$\Psi = \Phi^{(C \hookrightarrow B)}$, is given by the set of maps
\[\left\{\Psi^{\cC} : \tcop{\cF^{\T}_{A}} \to \tcop{\cF^{\cC}_{D}}\right\}_{\cC \in \bcut} \ \] that have \cjrep 
\begin{multline} \label{eq:loop} R^{\cC}_{\Psi}(\psi_D \tensor \psi_A
  ; \varphi_D \tensor \varphi_A) \\ = \sum_{k,\ell} R^{\cC}_{\Phi}(k_C
  \tensor \psi_D \tensor \psi_A \tensor \bar{k}_B ; \ell_C \tensor
  \varphi_D \tensor \varphi_A \tensor \bar{\ell}_B )\,, \end{multline}
where
$\ket{\bar{k}_B} = \sum_{i = 1}^{\infty}\ket{i}
\overline{\braket{i}{k}}$
for the basis $\{\ket{i_B}\}_i$ of $\cF^{\T}_B$ used in the \cjrep of
$\Phi^\cC$.
\end{deff}

\begin{rem}[The \natrep of a loop]
  \label{rem:nat.rep}
  In \appendixref{app:rep.natural} we define the \natrep of a map
  $\Phi : \tcop{\hilbert_{AB}} \to \tcop{\hilbert_{DC}}$ as a linear
  operator $K_\Phi : \hilbert_{A\bar{A}B\bar{B}} \to \hilbert_{D\bar{D}C\bar{C}}$.
  Relabelling $B$ and $\bar{B}$ with $C$ and $\bar{C}$ according to the
  isomorphism between the spaces used to define the loop, we can
  write $K_\Phi$ as an operator $K_\Phi : \hilbert_{A\bar{A}C\bar{C}} \to
  \hilbert_{D\bar{D}C\bar{C}}$. We show in \lemref{lem:nat.rep} that a
  map with \cjrep given by  \eqnref{eq:loop} may equivalently be
  defined by its \natrep
  \[ K_\Psi = \trace[C\bar{C}]{K_\Phi}\,.\]
  Since the partial trace is basis independent, the choice of bases in
  \eqnref{eq:loop}, \defref{def:loop}, is not relevant.
\end{rem}

It follows from \propref{prop:loopfromsequence} that if $\Phi^\cC$ is
CPTP, then so is the map $\Psi^\cC$ with \cjrep given by
\eqnref{eq:loop}. We still need to show that the resulting set
$\{\Psi^\cC\}_{\cC \in \bcut}$ is a valid causal box.

\begin{prop}
\label{prop:loop}
If $\Phi = \{\Phi^\cC : \tcop{\cF^{\T}_{AB}} \to \tcop{\cF^{\cC}_{CD}}\}$
is a (sub)normalized $(d_A+d_B,d_C+d_D)$\-/causal box with
$d_B = d_C$, then $\Psi = \Phi^{(C \hookrightarrow B)}$ is a (sub)\-nor\-mal\-ized
$(d_A,d_D)$\-/causal box.
\end{prop}

\begin{proof}
  We need to show that the maps $\{\Psi^\cC\}_{\cC \in \bcut}$ satisfy the
  consistency (\eqnref{eq:consistency}) and causality
  (\eqnref{eq:causality}) conditions. We only need to prove this for
  normalized boxes, since putting a loop on a subnormalized box is
  equivalent to putting the loop on the underlying normalized box then
  projecting the ancilla wire $R$ on the vacuum state $\vacuum_R$,
  which results in a valid subnormalized box if putting the loop on
  the underlying normalized box is valid.

  To prove this, we use the \cjrep of these conditions, given in
  \eqnref{eq:rep.equiv.cj}. What we need to prove is that there exists
  a causality function $\chi$ such that for any $\cC,\cD \in \bcut$ with $\cC \subseteq \cD$,
  and any $\psi^{\cC}_{D}$, $\psi^{\chi(\cC)}_{A}$,
  $\psi^{\widetilde{\chi(\cC)}}_{A}$, $\varphi^{\cC}_{D}$,
  $\varphi^{\chi(\cC)}_{A}$, $\varphi^{\widetilde{\chi(\cC)}}_{A}$,
\begin{multline} \label{eq:loop.prop} \sum_{j}
    R^{\cD}_{\Psi}\left(\psi^{\cC}_{D} \tensor j^{\widetilde{\cC}}_{D} \tensor \psi^{\chi(\cC)}_{A}
      \tensor \psi^{\widetilde{\chi(\cC)}}_{A} ; \varphi^{\cC}_{D} \tensor j^{\widetilde{\cC}}_{D} \tensor
      \varphi^{\chi(\cC)}_{A} \tensor \varphi^{\widetilde{\chi(\cC)}}_{A} \right) \\ =
    R^{\cC}_{\Psi}\left(\psi^{\cC}_{D} \tensor \psi^{\chi(\cC)}_{A} ;
      \varphi^{\cC}_{D} \tensor \varphi^{\chi(\cC)}_{A}\right)
    \braket{\psi^{\widetilde{\chi(\cC)}}_{A}}{\varphi^{\widetilde{\chi(\cC)}}_{A}}\,, \end{multline} 
where $R^{\cD}_{\Psi}$ and $R^{\cC}_{\Psi}$ are the \cjreps of the maps $\Psi^\cD :
\tcop{\cF^{\chi(\cD)}_A} \to \tcop{\cF^{\cD}_D}$ and $\Psi^{\cC} :
\tcop{\cF^{\chi(\cC)}_A} \to \tcop{\cF^{\cC}_D}$,
respectively. We prove this for the same function $\chi$ for which $\Phi$
satisfies causality.

We denote $\hat{\Psi}^{\cC} \coloneqq {\tr_{\widetilde{\cC}}} \circ
\Psi^{\cD}$. 
The left-hand side of \eqnref{eq:loop.prop} corresponds to the \cjrep
of $\hat{\Psi}^{\cC}$. Using the loop formula from \eqnref{eq:loop}, we
get
\begin{align*} 
& R^{\cC}_{\hat{\Psi}} \left(\psi^{\cC}_{D} \tensor \psi^{\chi(\cD)}_{A} ; \varphi^{\cC}_{D} \tensor
      \varphi^{\chi(\cD)}_{A}\right) \\
& \quad = \sum_{j} R^{\cD}_{\Psi}\left(\psi^{\cC}_{D} \tensor j^{\widetilde{\cC}}_{D} \tensor \psi^{\chi(\cD)}_{A}
      ; \varphi^{\cC}_{D} \tensor j^{\widetilde{\cC}}_{D} \tensor
      \varphi^{\chi(\cD)}_{A} \right) \\ & \quad =
    \sum_{j,k,\ell} R^{\cD}_{\Phi}\left(k_C
      \tensor \psi^{\cC}_{D} \tensor j^{\widetilde{\cC}}_{D} \tensor \psi^{\chi(\cD)}_{A}
      \tensor \bar{k}_B ; \ell_C \tensor \varphi^{\cC}_{D} \tensor j^{\widetilde{\cC}}_{D} \tensor
      \varphi^{\chi(\cD)}_{A} \tensor \bar{\ell}_B\right)\,,
\end{align*} 
where $R^{\cD}_{\Phi}$ is the \cjrep of
$\Phi^{\cD} : \tcop{\cF^{\T}_{AB}} \to \tcop{\cF^{\cD}_{CD}}$. Since
$\Phi$ satisfies causality, the state of the wire $B$ in positions
$\cD \setminus \cC$ cannot modify the output on the wire $D$ in
positions $\cC$. Hence there is no need to loop all states on $C$ back
to $B$, one could transmit only the states in positions $\cC$ and
trace out the others without changing the output, i.e.,
\begin{multline*} 
  R^{\cC}_{\hat{\Psi}} \left(\psi^{\cC}_{D} \tensor \psi^{\chi(\cD)}_{A} ; \varphi^{\cC}_{D} \tensor
    \varphi^{\chi(\cD)}_{A}\right)  = \\
   \sum_{i,j,k,\ell} R^{\cD}_{\Phi}\left(k^{\cC}_C \tensor i^{\widetilde{\cC}}_C
    \tensor \psi^{\cC}_{D} \tensor j^{\widetilde{\cC}}_{D} \tensor \psi^{\chi(\cD)}_{A}
    \tensor \bar{k}^{\cC}_B; \right. \\
 \left. \ell^{\cC}_C
   \tensor i^{\widetilde{\cC}}_C\tensor \varphi^{\cC}_{D} \tensor j^{\widetilde{\cC}}_{D}
   \tensor \varphi^{\chi(\cD)}_{A} \tensor \bar{\ell}^{\cC}_B
   \right)\,.
\end{multline*} 
Applying \eqnref{eq:rep.equiv.cj} we get
\begin{multline*} 
  R^{\cC}_{\hat{\Psi}} \left(\psi^{\cC}_{D} \tensor \psi^{\chi(\cD)}_{A} ; \varphi^{\cC}_{D} \tensor
    \varphi^{\chi(\cD)}_{A}\right) = \\
    \sum_{i,k,\ell} R^{\cC}_{\Phi}\left(k^{\cC}_C
    \tensor \psi^{\cC}_{D} \tensor \psi^{\chi(\cC)}_{A}
    \tensor \bar{k}^{\chi(\cC)}_B; \ell^{\cC}_C
   \tensor \varphi^{\cC}_{D}
   \tensor \varphi^{\chi(\cC)}_{A} \tensor \bar{\ell}^{\chi(\cC)}_B \right) \\
\braket{\psi^{\widetilde{\chi(\cC)}}_A
  \tensor \bar{k}^{\cC\setminus\chi(\cC)}_B}{\varphi^{\widetilde{\chi(\cC)}}_A
  \tensor \bar{\ell}^{\cC\setminus\chi(\cC)}_B}\,,
\end{multline*} 
where $R^{\cC}_{\Phi}$ is the \cjrep of
$\Phi^{\cC} : \tcop{\cF^{\chi(\cC)}_{AB}} \to \tcop{\cF^{\cC}_{CD}}$.
Let
$\hat{\Phi}^{\cC} \coloneqq \Phi^{\cC} \tensor \tr$ for a trace
operator $\tr$ acting on $\cF^{\cC\setminus\chi(\cC)}_{B}$. Then the previous
equation may be written
\begin{multline*} 
  R^{\cC}_{\hat{\Psi}} \left(\psi^{\cC}_{D} \tensor \psi^{\chi(\cD)}_{A} ; \varphi^{\cC}_{D} \tensor
    \varphi^{\chi(\cD)}_{A}\right) = \\
    \sum_{k,\ell} R^{\cC}_{\hat{\Phi}}\left(k^{\cC}_C
    \tensor \psi^{\cC}_{D} \tensor \psi^{\chi(\cC)}_{A}
    \tensor \bar{k}^{\cC}_B; \ell^{\cC}_C
   \tensor \varphi^{\cC}_{D}
   \tensor \varphi^{\chi(\cC)}_{A} \tensor \bar{\ell}^{\cC}_B
 \right) \braket{\psi^{\widetilde{\chi(\cC)}}_A}{\varphi^{\widetilde{\chi(\cC)}}_A}\,.
\end{multline*} 
Finally, using \eqnref{eq:loop}, we get
\begin{multline*} 
  R^{\cC}_{\hat{\Psi}} \left(\psi^{\cC}_{D} \tensor \psi^{\chi(\cD)}_{A} ; \varphi^{\cC}_{D} \tensor
    \varphi^{\chi(\cD)}_{A}\right) = \\
R^{\cC}_{\Psi}\left(\psi^{\cC}_{D} \tensor \psi^{\chi(\cC)}_{A}
   ; \varphi^{\cC}_{D} \tensor \varphi^{\chi(\cC)}_{A}
 \right) \braket{\psi^{\widetilde{\chi(\cC)}}_A}{\varphi^{\widetilde{\chi(\cC)}}_A}\,. \qedhere
\end{multline*}
\end{proof}

\subsection{The composition operation}
\label{sec:network.composition}

Now that we have defined parallel composition and loops, we can
instantiate the operation for composing two causal boxes by their
wires, as introduced in \secref{sec:theory}.

\begin{deff}[Composition operation]
\label{def:composition}
Let $\Phi$ and $\Psi$ be two (sub)\-nor\-mal\-ized causal boxes, let
$\ports(\Phi)$ and $\ports(\Psi)$ consist of a partition of the input
and output wires into sub-wires (see \remref{rem:ports}), and let the set
$P = \{(A^\Phi_1,A^\Psi_1),\dotsc,(A^\Phi_n,A^\Psi_n)\}$ consist of pairs of
ports of $\Phi$ and $\Psi$, such that each pair consists of an output
and input sub-wire of the same dimension, and each sub-wire appears at
most once. Then
\[\Phi \connect{P} \Psi \coloneqq \left(\Phi \|
  \Psi\right)^{(A^{\Phi/\Psi}_1 \hookrightarrow A^{\Psi/\Phi}_1)
  \dotsb (A^{\Phi/\Psi}_n \hookrightarrow A^{\Psi/\Phi}_n)}\,, \]
where $A^{\Phi/\Psi}_i \hookrightarrow A^{\Psi/\Phi}_i$ denotes either
$A^{\Phi}_i \hookrightarrow A^{\Psi}_i$ or $A^{\Psi}_i \hookrightarrow
A^{\Phi}_i$ depending on which is an output and input wire.
\end{deff}

\begin{thm}
\label{thm:closure}
Let $\fS$ be the set of all (sub)normalized causal boxes. For any
$\Phi,\Psi \in \fS$ and any set of valid sub-wires $P$,
\[ \Phi \connect{P} \Psi \in \fS\,.\]
\end{thm}

\begin{proof}
Immediate from \propref{prop:parallel} and \propref{prop:loop}.
\end{proof}

The next step is to prove that \defref{def:composition} satisfies
composition order independence. To do this, we will prove some lemmas
on the commutativity and associativity of parallel composition and
loops.

\begin{lem}
\label{lem:parallel.associativity}
For any $\Phi,\Psi,\Gamma \in \fS$,
\[\left(\Phi \| \Psi \right) \| \Gamma = \Phi \| \left( \Psi \| \Gamma \right)\,.\]
\end{lem}

\begin{proof}
This follows from the associativity of the tensor product, namely,
\[(\Phi^\cC \tensor \Psi^\cC) \tensor \Gamma^\cC = \Phi^\cC \tensor (\Psi^\cC
  \tensor \Gamma^\cC)\,. \qedhere \]
\end{proof}

\begin{lem}
\label{lem:parallelloop.commute}
For any $\Phi,\Psi \in \fS$ and any pair
$(B,C) \in \ports(\Psi) \times \ports(\Psi)$ of in- and out\-/ports of
the same dimension,
  \begin{align*}
    \Psi^{(C \hookrightarrow B)} \| \Phi & = \left( \Psi \| \Phi \right)^{(C \hookrightarrow B)}\,, \\
    \Phi \| \Psi^{(C \hookrightarrow B)} & = \left( \Phi \| \Psi \right)^{(C \hookrightarrow B)}\,.
  \end{align*}
\end{lem}

\begin{proof}
  The \cjrep of a product of maps is given by
\begin{multline*}
  R_{\Phi \tensor \Psi}(\psi_Y \tensor \psi_D \tensor \psi_X \tensor
  \psi_A ; \varphi_Y \tensor
  \varphi_D \tensor \varphi_X \tensor \varphi_A ) = \\
  R_{\Phi}(\psi_Y \tensor \psi_X ; \varphi_Y \tensor \varphi_X )
  R_{\Psi}( \psi_D \tensor \psi_A ; \varphi_D \tensor \varphi_A
  )\,. \end{multline*} Hence for
$\Phi^\cC : \tcop{\cF^{\T}_{X}} \to \tcop{\cF^{\cC}_{Y}}$ and
$\Psi^\cC : \tcop{\cF^{\T}_{AB}} \to \tcop{\cF^{\cC}_{CD}}$, both
$ \Phi \| \Psi^{(C \hookrightarrow B)}$ and
$\left( \Phi \| \Psi \right)^{(C \hookrightarrow B)}$ have \cjrep
\begin{multline*}
  R^{\cC}_{\Gamma}(\psi_{Y} \tensor \psi_D \tensor \psi_{X} \tensor
  \psi_{A}; \varphi_{Y} \tensor \varphi_D \tensor \varphi_{X} \tensor
  \varphi_A) = \\
  R^{\cC}_{\Phi}(\psi_Y \tensor \psi_X ; \varphi_Y \tensor \varphi_X) 
  \sum_{k,\ell} R^{\cC}_{\Psi}(k_C \tensor \psi_{D} \tensor \psi_{A}
  \tensor \bar{k}_B ; \ell_C \tensor \varphi_D \tensor \varphi_{A}
  \tensor \bar{\ell}_B)\,.
\end{multline*}
The proof for $\Psi^{(C \hookrightarrow B)} \| \Phi$ is identical.
\end{proof}

\begin{lem}
\label{lem:loop.commute}
For any $\Phi \in \fS$ and any pairs
$(A,D),(B,C) \in \ports(\Phi) \times \ports(\Phi)$ of in- and
out\-/ports of the same dimension,
\begin{equation} \label{eq:loop.commute} \left(\Phi^{(C \hookrightarrow B)}\right)^{(D
    \hookrightarrow A)} = \left(\Phi^{(D \hookrightarrow
      A)}\right)^{(C \hookrightarrow B)}. \end{equation}
\end{lem}

\begin{proof}
Let $\Phi^\cC : \tcop{\cF^{\T}_{ABX}} \to \tcop{\cF^{\cC}_{CDY}}$. The \cjrep of the left-hand side of
\eqnref{eq:loop.commute} is
 \begin{multline*}
  R^{\cC}_{\Psi}(\psi_Y \tensor \psi_X ; \varphi_Y \tensor \varphi_X) = \\
  \sum_{e,h} \sum_{k,\ell} R^{\cC}_{\Phi}(k_C \tensor e_D \tensor \psi_Y
  \tensor \bar{e}_A \tensor \bar{k}_B \tensor \psi_X ; \ell_C \tensor
  h_D \tensor \varphi_Y \tensor \bar{h}_A \tensor \bar{\ell}_B \tensor
  \varphi_X )\,.
  \end{multline*}
Whereas for the right-hand side we get
 \begin{multline*}
  R^{\cC}_{\Psi}(\psi_Y \tensor \psi_X ; \varphi_Y \tensor \varphi_X) = \\
  \sum_{k,\ell} \sum_{e,h} R^{\cC}_{\Phi}(k_C \tensor e_D \tensor \psi_Y
  \tensor \bar{e}_A \tensor \bar{k}_B \tensor \psi_X ; \ell_C \tensor
  h_D \tensor \varphi_Y \tensor \bar{h}_A \tensor \bar{\ell}_B \tensor
  \varphi_X )\,.
  \end{multline*}
  The two are equal if the limits implicit in the sums commute. This
  is the case, because a third expression obtained by first merging
  the wires $A$ with $B$ and $C$ with $D$ then looping $CD$ to $AB$
  results in a valid map
\[\sum_{e,h,k,\ell} R^{\cC}_{\Phi}(k_C \tensor e_D \tensor \psi_Y
\tensor \bar{e}_A \tensor \bar{k}_B \tensor \psi_X ; \ell_C \tensor
h_D \tensor \varphi_Y \tensor \bar{h}_A \tensor \bar{\ell}_B \tensor
\varphi_X )\,,\]
and from \lemref{lem:limit.swap} we know that if all three expressions
converge, then they must converge to the same value. This can also be
seen from the partial trace representation of a loop, namely
  \[ K^{\cC}_{\Psi} = \tr_{D\bar{D}} \trace[C\bar{C}]{K^{\cC}_\Phi} =
  \trace[C\bar{C}D\bar{D}]{K^{\cC}_\Phi} = \tr_{C\bar{C}}
  \trace[D\bar{D}]{K^{\cC}_\Phi}\,. \qedhere\] 
\end{proof}

These three lemmas allow the order in which systems are put in
parallel and wires are connected to be changed, which results in
composition order independence.

\begin{thm}
\label{thm:comp.order.indep}
For any causal boxes
$\Phi_1,\Phi_2,\Phi_3 \in \fS$ and compatible pairs of sub-wires
$P_{ij} \subseteq \ports(\Phi_i) \times \ports(\Phi_j)$,
\begin{equation*} 
\left(\Phi_1
    \connect{P_{12}} \Phi_2\right) \connect{P_{13} \cup P_{23}} \Phi_3
  = \Phi_1 \connect{P_{12} \cup P_{13}} \left( \Phi_2 \connect{P_{23}}
    \Phi_3 \right)\,.\end{equation*}
\end{thm}

\begin{proof}
Immediate by combining Lemmas~\ref{lem:parallel.associativity},
\ref{lem:parallelloop.commute} and \ref{lem:loop.commute}.
\end{proof}

\section{Distance}
\label{sec:metric}

In practical applications it is often useful to define a notion of
distance between systems, e.g., to measure how close a real system
constructed from a noisy channel and an error correcting code is from
an ideal channel that has no noise. In this section we introduce a
pseudo\-/metric on the set of causal boxes, the \emph{distinguishing
  advantage}. This is defined with the help of a distinguisher:
another system $\aD$ that, when connected to one of two causal boxes,
$\Phi$ or $\Psi$, outputs a bit corresponding to its best guess of the
system to which it is connected. A distance between these causal boxes
is then defined as the statistical distance between the
distinguisher's outputs.

Although one could define a distinguisher to be another causal box, we
consider a slightly more general definition in this work: a
distinguisher may also put loops on the system being tested, i.e.,
from an output of $\Phi$ back to an input of $\Phi$ \--- this is more
powerful, because a loop ``takes no time'', but plugging in a causal
box would incur a delay since the identity is not a causal
operation. Such a distinguisher is illustrated in
\figref{fig:distinguisher}. Since it only makes sense to compare two
systems $\Phi$ and $\Psi$ that have inputs and outputs of the same
dimensions, we also parametrize the distinguishers by these input and
output dimensions.

\begin{figure}[tb]
\begin{centering}
\begin{tikzpicture}[
wArrow/.style={->,>=stealth,thick}]

\node[inner sep=0] (A) at (0,-1.5) {};
\node[inner sep=0] (T) at (0,.5) {};
\node[inner sep=0] (B) at (0,-.5) {};
\node[draw,thick,minimum width=2cm,minimum height=2cm] (phi) at (0,0) {$\Phi$};

\draw[wArrow] (-2,0 |- T) to (-1,0 |- T);
\draw[wArrow] (1,0 |- T) to (2,0 |- T);

\draw[thick] (1,0 |- B) .. controls (2,0 |- B) and (2,0 |- A) .. (A.center);
\draw[wArrow] (A.center) .. controls (-2,0 |- A) and (-2,0 |- B) .. (-1,0 |- B);

\draw[thick] (-3,1) -- ++(1,0) -- ++(0,-3) -- ++(4,0) -- ++(0,3)
-- ++(1,0) -- ++(0,-4) -- ++(-6,0) -- cycle;

\node (D) at (0,-2.5) {$\hat{\aD}$};
\draw[wArrow] (3,-1) to node[auto] {$0,1$} (4,-1);

\end{tikzpicture}

\end{centering}
\caption[Distinguisher]{\label{fig:distinguisher}A distinguisher
  plugged into a causal box $\Phi$ can loop some of the outputs back
  to the inputs as well as process other outputs before providing them
  as inputs to $\Phi$. It then outputs a bit corresponding to its
  decision.}
\end{figure}
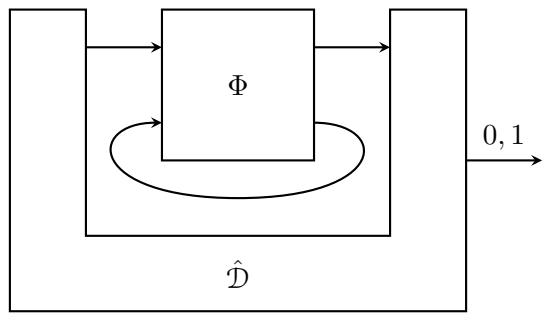

Although one traditionally depicts distinguishers as outputting a bit
$0$ or $1$, we formally model a distinguisher as having a
$1$\=/dimensional output wire, i.e., it outputs a state
$\rho \in \tcop{\cF^\T}$. We then define the distinguisher's decision
to be $0$ if it outputs a vacuum state, and $1$ otherwise. According
to our definition of causal boxes, an output only needs to be defined
on all bounded cuts $\cC \subseteq \T$, it is not necessarily defined
on all of $\T$. If the set $\T$ were totally ordered, one could define
the output value of the distinguisher to be the limit of the output
value on $\T^{\leq t}$ as $t \to +\infty$.  However, for a partially
ordered set, it is no clear how to take such a limit. Instead, we
define the output of a distinguisher $\aD$ to correspond to the value
produced on a bounded cut $\T^{\leq t_\aD}$, for some fixed
$t_\aD \in \T$. By considering a set of distinguishers
$\fD = \{\aD_i\}_i$ that have different bounds $t_{\aD_i} \in \T$ and
taking the supremum of the distances for each $\aD_i \in \fD$, we get
a notion of distinguishability that covers all of $\T$.

In \appendixref{app:alternative} we discuss some alternative
distinguisher definitions. We rewrite the definitions from this
section using subnormalized boxes, which simplifies them
considerably. Then we consider distinguishers that are not constrained
to producing an output within a fixed cut $\T^{\leq t_\aD}$, and show
that this results in an equivalent notion of distance.

\begin{deff}[Distinguisher] \label{def:distinguisher} A
  \emph{$(m,n)$\-/distinguisher}
  $\aD = \{{\id^\cC} \tensor \hat{\aD}^\cC\}_{\cC \in \bcut}$ consists
  of a $(\hat{n},\hat{m})$\-/causal box
  $\hat{\aD} = \{\hat{\aD}^\cC\}_{\cC \in \bcut}$ with
  $m+1-\hat{m} = n - \hat{n}$, a termination time $t_\aD \in \T$, and
  a specification of how the distinguisher is connected to an
  $(m,n)$\-/dimensional system \--- i.e., which input and output
  sub-wires are connected to $\hat{\aD}$ and which are directly
  connected by a loop. We refer to the additional output sub-wire of
  dimension $1$ of the distinguisher that is never connected to any
  other system as the distinguisher's output wire. For an
  $(m,n)$\-/distinguisher $\aD$ and an $(m,n)$\-/causal box $\Phi$,
  let $\aD\Phi$ denote the causal box with no input wire and a
  $1$\=/dimensional output wire resulting from connecting the systems
  as specified. We define $\aD[\Phi]$ to be the binary random variable
  on $\{0,1\}$ obtained by projecting the output of $\aD\Phi$ within
  $\T^{\leq t_\aD}$ on $P^{\leq t_\aD}_0 = \proj{\Omega}^{\leq t_\aD}$
  and
  $P^{\leq t_\aD}_1 = I^{\leq t_\aD} - \proj{\Omega}^{\leq t_\aD}$.
\end{deff}

\begin{deff}[Distance] \label{def:metric} Given a set of
  $(m,n)$\-/distinguishers $\fD$, the distance between two
  $(m,n)$\-/causal boxes $\Phi$ and $\Psi$ is defined as
\[d^{\fD}(\Phi,\Psi) \coloneqq \sup_{\aD \in \fD}
\delta(\aD[\Phi],\aD[\Psi])\,,\] where $\delta(\cdot,\cdot)$ is the
statistical or total variation distance, i.e., 
\[ 
\delta(X_1,X_2) = \frac{1}{2} \sum_{x \in \{0,1\}} \left| \Pr[X_1 = x] - \Pr[X_2 = x]
    \right|\,.
\]
\end{deff}

We now show that the the function $d^{\fD} : \fS \times \fS \to [0,1]$
defined above is a pseudo\-/metric for arbitrary $\fD$ and a metric if
$\fD$ is the set of all distinguishers.

\begin{thm}
\label{thm:metric}
The distance measure given in \defref{def:metric} is a
pseudo\-/metric. Furthermore, if $\fD$ is the set of all
distinguishers, then this distance is a metric.
\end{thm}

\begin{proof}
We need to show that the three conditions given in
\secref{sec:theory} hold, namely, for any $\Phi,\Psi,\Gamma
\in \fS$,
\begin{align*}
d^{\fD}(\Phi,\Phi) & = 0\,, \\
d^{\fD}(\Phi,\Psi) & = d(\Psi,\Phi)\,, \\
d^{\fD}(\Phi,\Psi) & \leq d(\Phi,\Gamma) + d(\Gamma,\Psi)\,.
\end{align*}

The first two conditions are immediate because the statistical
distance is a metric, hence $\delta(\aD[\Phi],\aD[\Phi]) = 0$ and
$\delta(\aD[\Phi],\aD[\Psi]) = \delta(\aD[\Psi],\aD[\Phi])$. The last
one follows from
\begin{align*}
d^{\fD}(\Phi,\Psi) & = \sup_{\aD \in \fD} \delta(\aD[\Phi],\aD[\Psi]) \\
& \leq \sup_{\aD \in \fD} \left(\delta(\aD[\Phi],\aD[\Gamma])  +
  \delta(\aD[\Gamma],\aD[\Psi])\right) \\
& \leq \sup_{\aD \in \fD} \delta(\aD[\Phi],\aD[\Gamma])  +
  \sup_{\aD' \in \fD} \delta(\aD'[\Gamma],\aD'[\Psi]) \\
& = d^{\fD}(\Phi,\Gamma) + d^{\fD}(\Gamma,\Psi)\,.
\end{align*}

Finally to show that $d^{\fD}$ is a metric if $\fD$ is the set of all
distinguishers, we need to prove that
\[\Phi \neq \Psi \implies \exists \aD \in \fD \text{ s.t. }
\delta(\aD[\Phi],\aD[\Psi]) > 0\,. \]
Since $\Phi \neq \Psi$ there must exist a $\cC \in \bcut$ for which
$\Phi^{\cC} \neq \Psi^{\cC}$, i.e., there exists a
$\rho \in \tcop{\cF^{\cC_\inrm}_X}$ such that
$\Phi^\cC(\rho) \neq \Psi^\cC(\rho)$, where
$\cC_\inrm = \chi_\Phi(\cC) \cup \chi_\Psi(\cC)$. A distinguisher
$\aD$ that distinguishes $\Phi$ from $\Psi$ can be constructed as the
parallel composition of two systems. The first prepares the state
$\rho$ and sends it to the system ($\Phi$ or $\Psi$). The second
gathers the output of the system until some position $t$ such that
$\cC \subseteq \T^{\leq t}$. Then it performs an optimal measurement to distinguish
$\Phi^\cC(\rho)$ from $\Psi^\cC(\rho)$. The result of this measurement
is output in a position $t_\aD \geq t$.
\end{proof}

\section{Example: the quantum switch}
\label{sec:qswitch}

The quantum switch \--- which was introduced in
\secref{sec:intro.superposition} \--- is a system that, given
black\-/box access to two unitaries $U$ and $V$, applies them in a
controlled superposition of different orders. We reproduce in
\figref{fig:qswitch.labeled} a figure illustrating the quantum switch
with some extra labels for the different wires.

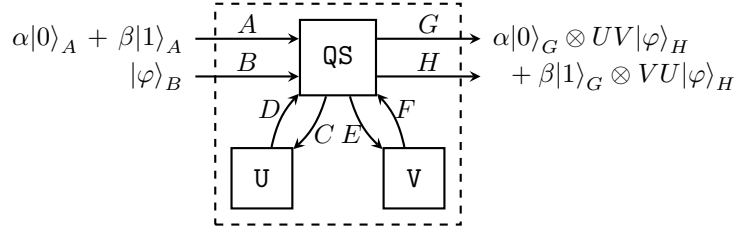
\begin{figure}[tb]
\begin{centering}
\begin{tikzpicture}[
wire/.style={-,thick},
wArrow/.style={->,>=stealth,thick},
unitary2/.style={draw,thick,minimum width=1cm,minimum height=1cm},
unitary1/.style={draw,thick,minimum width=.8cm,minimum height=.8cm}]

\node[inner sep=0] (l1) at (-1.9,.25) {};
\node[left,text width=3.2cm,align=right] at (l1) {\small $\alpha \zero_A + \beta \one_A$};
\node[inner sep=0] (l2) at (-1.9,-.25) {};
\node[left] at (l2) {\small $\ket{\varphi}_B$};
\node[inner sep=0] (r1) at (1.9,.25) {};
\node[right,text width=3.2cm,align=left] at (r1)
    {\small $\alpha \zero_G \tensor UV \ket{\varphi}_H$};
\node[inner sep=0] (r2) at (1.9,-.25) {};
\node[right,text width=3.2cm,align=right] at (r2)
    {\small ${}+\beta \one_G \tensor VU \ket{\varphi}_H$};

\node[unitary2] (qs) at (0,0) {$\QS$};
\node[unitary1] (u) at (-1,-1.6) {$\QU$};
\node[unitary1] (v) at (1,-1.6) {$\QV$};

\node[fit=(qs)(u)(v),thick,dashed,draw,inner sep=.2cm] (box) {};

\draw[wArrow] (l1) to node[auto,yshift=-2] {\small $A$} (qs.west |- l1);
\draw[wArrow] (l2) to node[auto,yshift=-2] {\small $B$} (qs.west |- l2);
\draw[wArrow,bend left=15] (qs) to node[auto,xshift=-4,yshift=4] {\small $C$} (u);
\draw[wArrow,bend right=15] (qs) to node[auto,swap,xshift=4,yshift=4] {\small $E$} (v);
\draw[wArrow,bend left=15] (u) to node[auto,xshift=4,yshift=-4] {\small $D$} (qs);
\draw[wArrow,bend right=15] (v) to node[auto,swap,xshift=-4,yshift=-4] {\small $F$} (qs);
\draw[wArrow] (qs.east |- r1) to node[auto,yshift=-2] {\small $G$} (r1);
\draw[wArrow] (qs.east |- r2) to node[auto,yshift=-2] {\small $H$} (r2);
\end{tikzpicture}

\end{centering}
\caption[Quantum switch with labels]{\label{fig:qswitch.labeled}The
  quantum switch, $\QS$, queries first $\QU$ then $\QV$ or first $\QV$
  then $\QU$ depending on the control qubit.}
\end{figure}

As shown in \cite{CDFP12,ACB14,Chi12,FAB15} the quantum switch can
decrease the complexity of computational tasks and increase the
distinguishability of quantum states and channels.  Although it is
rather straightforward to describe the effect of plugging the quantum
switch into the systems $\QU$ and $\QV$ \--- this corresponds to the
dashed box from \figref{fig:qswitch.labeled} \--- the quantum switch
itself, i.e., the box $\QS$ in \figref{fig:qswitch.labeled}, cannot be
described using combs or circuits~\cite{CDPV13}. In this section we use our
framework to provide a description of the system $\QS$ as a
causal box.

The total system consists of $3$ boxes, $\QU$, $\QV$ and $\QS$. We
first model these $3$ sub-systems, then prove that the composition of
the $3$ results in a causal box which applies either $UV$ or $VU$ to
an input state according to a control qubit. The entire system can be
executed in $6$ steps, so it is sufficient to choose
$\T = \{1,\dotsc,6\}$. The system can easily be extended by adding
more points to $\T$ and defining the behavior of the boxes on these
points.

The sub\-/system $\QU$ applies the corresponding unitary to the input
it receives, and outputs the result a step later. For simplicity we
assume that if multiple messages are received simultaneously, $U$ is
applied to every one of them. We also provide the box with an
(internal) counter, that keeps track of the number of times that $U$
is applied. Let $\ket{\psi^n_t} \in \vee^n (\C^d \tensor \ket{t})$ be
an element of the symmetric subspace of $n$ qudits all arriving in
position $t$, for $n \geq 1$. Then
\begin{align*}
\QU \vacuum_C \tensor \ket{i}_{\QU} & = \vacuum_D \tensor \ket{i}_{\QU}\,, \\
\QU \ket{\psi^n_t}_C \tensor \ket{i}_{\QU} & = \ket{\left( U^{\tensor
      n} \psi^n\right)_{t+1}}_D \tensor \ket{i+n}_{\QU}\,, \end{align*}
where the register $C$ contains the input to $\QU$, $D$ contains the
output, and the register denoted $\QU$ is the internal counter of
the system. The box $\QV$ is defined similarly.

As can be seen in \figref{fig:qswitch.labeled}, the box $\QS$ has $8$
wires. To provide a complete description of $\QS$ we need to define
how inputs on the entire Fock space of each input wire are mapped to
the output space for all points in $\T$. However, many of these inputs
are ``invalid'', e.g., the initial input is expected to be received in
step $1$, and $\QS$ is plugged into systems $\QU$ and $\QV$ that
behave as described above. Hence, we describe the behavior of $\QS$ on
valid inputs only, and assume that any state orthogonal to these are
simply ignored by $\QS$.

In the first step, upon receiving the input \[\left(\alpha \zero_A +
  \beta \one_A\right) \tensor \ket{\psi}_B \,,\] $\QS$ moves the control
qubit to its internal register and queries either $\QU$ or $\QV$
conditioned on this state. Let $\QS_1$ denote this operation; it
is defined as
\begin{align*}
\QS_1 \zero_A \tensor \ket{\psi}_B & = \zero_\QS \tensor \vacuum_C
                                     \tensor \ket{\psi}_E \,, \\
\QS_1 \one_A \tensor \ket{\psi}_B & = \one_\QS \tensor \ket{\psi}_C
                                    \tensor \vacuum_E \,.
\end{align*}
When $t = 3$, $\QS$ forwards what it received from $\QU$ to $\QV$ and
from $\QV$ to $\QU$, i.e.
\begin{equation*}
\QS_3 \ket{\psi}_D \tensor \ket{\varphi}_F = \ket{\varphi}_C \tensor \ket{\psi}_E \,.
\end{equation*}
And finally, in the last step, conditioned on the value of the control
qubit, $\QS$ outputs either the message from $\QU$ or from $\QV$ along
with the control qubit.
\begin{align*}
\QS_5 \zero_\QS \tensor \ket{\psi}_D \tensor \vacuum_F & = \zero_G \tensor \ket{\psi}_H \,, \\
\QS_5 \one_\QS \tensor \vacuum_D \tensor \ket{\psi}_F & = \one_G \tensor \ket{\psi}_H \,.
\end{align*}

\begin{lem}
\label{lem:qswitch}
The composition of $\QS$, $\QU$ and $\QV$ results in a system which
performs a controlled switch between the orders of $U$ and
$V$. Furthermore, the boxes $\QU$ and $\QV$ are queried only once
each.
\end{lem}

\begin{proof}
To prove this we put together all the steps described above. The wires
that do not appear in the equations contain a vacuum state in the
corresponding step.
\begin{align*}
  & t = 1 & & \left(\alpha \zero_A + \beta \one_A\right)
              \ket{\psi}_B \zero_\QU \zero_\QV \,, \\
  & t = 2 & & \left( \alpha \zero_\QS  \vacuum_C 
              \ket{\psi}_E + \beta \one_\QS 
              \ket{\psi}_C  \vacuum_E \right) 
              \zero_\QU  \zero_\QV \,, \\ 
  & t = 3 & & \alpha \zero_\QS  \vacuum_D 
              \ket{V\psi}_F 
              \zero_\QU  \one_\QV + \beta \one_\QS 
              \ket{U\psi}_D  \vacuum_F 
              \one_\QU  \zero_\QV \,, \\
  & t = 4 & & \alpha \zero_\QS   \ket{V\psi}_C \vacuum_E 
              \zero_\QU  \one_\QV + \beta \one_\QS 
              \vacuum_C \ket{U\psi}_E
              \one_\QU  \zero_\QV \,, \\
  & t = 5 & & \left( \alpha \zero_\QS   \ket{UV\psi}_D \vacuum_F 
              + \beta \one_\QS 
              \vacuum_D \ket{VU\psi}_F \right)
              \one_\QU  \one_\QV \,, \\
  & t = 6 & & \left( \alpha \zero_G   \ket{UV\psi}_H 
              + \beta \one_G 
              \ket{VU\psi}_H \right)
              \one_\QU  \one_\QV \,.
\end{align*}
In the final step, the wires $G$ and $H$ contain the desired output,
and the counters of $\QU$ and $\QV$ are set to $1$.
\end{proof}

\begin{rem}[Alternative quantum switch]
  \label{rem:qswitch}
  The quantum switch defined in this section stores the control qubit in
  its internal memory. The number of inputs it may treat
  simultaneously is thus limited by this memory. An alternative
  construction was proposed independently in \cite{FDDB14} and
  \cite{ACB14} that does not involve any memory: the quantum switch
  sends the control qubit along with the target qubit to the systems
  $\QU$ and $\QV$. This allows an unbounded number of switches to be
  performed simultaneously (and even do them in superposition), but it
  only works if one can assume that the systems $\QU$ and $\QV$
  perform their operation on a subspace of the input received and
  identity on the remaining input.
\end{rem}

\section{Concluding remarks}
\label{sec:conclusion}

Many fundamental notions used to define quantum computation are
derived by generalizing classical notions, e.g., quantum Turing
machines~\cite{Ben80,Deu85,Yao93,BV97}, quantum
circuits~\cite{Deu89,Yao93}, or the machine model used in quantum
composable security~\cite{Unr04,Unr10}. This can result in some
classical concepts being hard\-/coded in the quantum model, e.g.,
quantum composable security~\cite{Unr04,Unr10} can only model a
classical scheduling of messages and quantum
circuits~\cite{Deu89,Yao93} do not consider states in a superposition
of being on one wire or another.\footnoteremember{fn:complexity}{If
  one interprets the circuit model as forbidding the presence of
  vacuum states on a wire, then one can prove impossibility results
  for certain systems which are physically
  valid~\cite{CDPV13,TGMV13,AFCB14}. However, if one adopts the view
  introduced in this work that a vacuum state should be modeled
  explicitly, then one can argue that a message in a superposition
  between two wires is just an entangled state between two wires,
  $\alpha\ket{\psi}_A \tensor \vacuum_B + \beta \vacuum_A \tensor
  \ket{\psi}_B$,
  and that this is captured by the circuit model. And in this sense,
  the circuit model is (still) a universal model of quantum
  computation. Nonetheless, the distinction between vacuum and
  non\-/vacuum state is vital for computation and query complexity,
  since if one sends nothing to an oracle, this oracle has not been
  queried.}

In this work we propose a model of quantum systems that is not derived
from a classical model, but uses a genuine quantum approach in which
the space of quantum messages is a Fock space. The general approach
follows the top\-/down paradigm~\cite{MR11}, which consists in
starting at the highest possible level of abstraction, and proceeds
downwards, introducing in each new lower level only the minimal
necessary specializations. The theory of systems introduced in
\secref{sec:theory} does not make any assumptions about the nature of
the underlying systems \--- it applies equally to a classical or
quantum model of systems. We instantiate these with causal boxes,
which are defined as black boxes. Only the input\-/output behavior of
the system is relevant, the internals \--- e.g., the memory or model
of computation \--- are not needed and thus not specified. This work
leaves open a multitude of questions on modeling discrete quantum
systems, some of which we discuss in the following paragraphs.

\paragraph{Quantum complexity.} One of the most fundamental open
questions is how to evaluate the complexity of a quantum circuit. A
straightforward adaptation of the classical concept of counting gates
in a circuit does not seem to be meaningful if two circuits are run in
superposition, or if a circuit consists of many gates used in
superposition. If one does not distinguish between a gate being used
or not (see \footnoteref{fn:complexity}), then the complexity of the
circuit would correspond to counting all the gates needed to draw the
entire circuit. In the case of the quantum switch from
\secref{sec:qswitch}, one would find that either $\QU$ or $\QV$ has to
be queried at least twice~\cite{CDPV13}, even though one of the
queries is the vacuum state. The gap between the number of gates
needed to draw a circuit and the number of gates actually used can be
quite large: Ara\'ujo et al.~\cite{ACB14} found a problem that cannot
be solved with a circuit containing less than $\Omega(n^2)$ gates, but
only superpositions of $O(n)$ subsets of these gates are ever
used. Note that to implement the algorithm of Ara\'ujo et al., it is
sufficient to construct $O(n)$ gates, since the circuit contains only
$O(n)$ \emph{different} gates and the quantum switch can be used to
swap the order of these gates~\cite{ACB14}. However, such a coherent
counting method cannot be applied in general, since this could result
in a measurement of whether the gate is used or not, which could
change the outcome of the computation. We would thus evaluate the
complexity of a different circuit altogether.\footnote{The example of
  Ara\'ujo et al.~\cite{ACB14} is constructed so that coherently
  counting which gates are used does not change the computation. This
  is possible when the counter factors out at the end of the
  computation. This is also the case in the example of the quantum
  switch from \secref{sec:qswitch}, where the internal counters of
  $\QU$ and $\QV$ are in a superposition of $\zero$ and $\one$ in the
  middle of the protocol, but factor out at the end.} A more physical
approach to measuring complexity may be needed, e.g., by evaluating
the time or energy needed to perform a
computation~\cite{GHRdRS16,Fai16}.

\paragraph{Indefinite causal structures.} Our framework models
messages in superpositions of different orders by explicitly defining
a set of positions $\T$ and assigning different positions (or
superpositions of positions) to different messages. This allows
superpositions of different causal structures to be captured by the
framework. A different line of research, inspired by Hardy's work on
probability theories with indefinite causal
structures~\cite{Har05,Har07,Har09,Har10}, has developed a process
matrix formalism to model indefinite causal
structures~\cite{OCB12,Bru14a,OG16,ABCFGB15,Bra16,BW14,Bru14b,BAFCB16,BW16b,BFW14,BW16a}. As
already mentioned in \secref{sec:intro.related}, the process matrix
framework does not only capture superpositions of causal structures,
but also unphysical causal structures. In fact, it is currently still
unclear what systems modeled by the process matrix framework are
physically implementable. Additionally, the various players connected
to the process matrix are restricted to processing inputs and
producing outputs exactly once, preventing any dynamical behavior of
the parties, e.g., they cannot change at runtime the number of
messages they read or send. This framework is thus not suited for
modeling the possible behaviors of players in a multi\-/player game or
cryptographic protocol.

It is nonetheless interesting to compare what systems can be modeled
by either framework. A system that allows a causal inequality to be
violated cannot be captured by causal boxes, since a global notion of
order is hard\-/coded in our framework. Hence causal boxes can only
model process matrices that are \emph{extensibly causal} \--- those
that cannot be used to violate a causal
inequality~\cite{OG16}. However, no violation of a causal inequality
has ever been physically realized. In fact, to date, realizable
physical experiments have a global notion of order, and thus causal
boxes might well capture the subset of process matrices that can be
implemented. It remains open whether this corresponds to extensibly
causal processes, or whether there are examples of such processes that
cannot be physically realized~\cite{FAB16}.

\paragraph{Relativistic cryptography.} An immediate application of our
causal boxes framework is in modeling cryptographic protocols that
involve time, e.g., relativistic
protocols~\cite{Ken99,Ken12,KTHW13,LKBHTWZ14,AK15,BCFGGOS11,Unr14}. Here,
(honest) parties are positioned in precise locations, and the time
taken by messages to travel form one location to another is used to
ensure that a dishonest party cannot cheat. To model such a setting
one can take $\T$ to be a countable subset of space\-/time. Two points
$t$ and $t'$ are ordered, $t \leq t'$, if $t'$ is in the future light
cone of $t$. All messages are then assigned some point $t \in \T$
corresponding to their location.

\paragraph{Non\-/deterministic systems.} Our framework may also be
used to model protocols that do not involve time explicitly. In many
settings, a player executes a set of instructions, but the time it
takes to do so is not determined by the player. One can think of such
a system as being \emph{non\-/deterministic}. It does not correspond
to one system, but to a set of systems, each one capturing one
possible behavior. For example, a simple identity channel could be a
set of systems which forward the message received with a delay
$\delta$, where the set is taken over all possible $\delta$, e.g.,
$\id = \{\id_\delta\}_{\delta \in \Q^+}$. The composition of two such
non\-/deterministic systems, $\Pi = \{\pi_x\}_x$ and
$\Psi = \{\psi_y\}_y$, is defined as the set resulting from composing
all pairs of elements from each system,
$\Pi \connect{P} \Psi = \{\pi_x \connect{P} \psi_y\}_{x,y}$. Since any
possible element $\pi \in \Pi$ could occur if a protocol $\Pi$ is
executed, we require security to hold for all elements $\pi \in \Pi$.
A main challenge in modeling systems this way is to determine what set
of behaviors correspond to a simple and intuitive, but incomplete,
description of a system as, e.g., a protocol given by pseudo\-/code or
a random system~\cite{Mau02,MPR07}. Furthermore, for such
non\-/deterministic systems to be usable, one would need to develop
formalism that allows them to be described compactly and which
satisfies the axioms from \secref{sec:theory}.

\section*{Acknowledgments}
\addcontentsline{toc}{section}{Acknowledgments}

CP is grateful to Gus Gutoski for discussions on quantum combs and how
to model composable security with them, to Vedran Dunjko for the
example of the controlled black\-/box unitary used in the introduction
as well as proofreading a preliminary version, and to Mateus Ara\'ujo,
\v{C}aslav Brukner, Lucien Hardy and Ognyan Oreshkov for providing
valuable comments. He is particularly indebted to Volkher Scholz for
having suggested the Hilbert space used to capture orders of messages,
provided invaluable advice on modeling physical systems, walked him
through the mathematics of infinite\-/dimensional Hilbert spaces, and
patiently answered his flow of questions and problems.

CP and RR are supported by the European Commission FP7 Project RAQUEL
(grant No.~323970), US Air Force Office of Scientific Research (AFOSR)
via grant~FA9550-16-1-0245, the Swiss National Science Foundation (via
the National Centre of Competence in Research `Quantum Science and
Technology') and the European Research Council -- ERC (grant
No.~258932). UM was supported by the Swiss National Science Foundation
(SNF), project No.~200020-132794. BT was supported by the Swiss
National Science Foundation (SNF) via Fellowship No.~P2EZP2_155566 and
in part by the NSF grants CNS-1228890 and CNS-1116800.

\appendix
\appendixpage
\phantomsection
\label{app}
\addcontentsline{toc}{section}{Appendices}

In \appendixref{app:fock} we define a Fock space in more
detail. In \appendixref{app:rep} we give a brief overview of
representations of quantum operators on infinite dimensional
systems. We discuss the \cjrep in \appendixref{app:rep.cj} and the
\natrep in
\appendixref{app:rep.natural}. 
In \appendixref{app:finite} we define the subset of causal boxes that
can be represented as a single map and which is closed under
composition.  In \appendixref{app:total} we discuss the special case
of the causality function for systems with a totally ordered set
$\T$. We prove that in this case, for the composition of two systems
to be well\-/defined, it is sufficient to require that the systems are
not arbitrarily fast around any point. In \appendixref{app:loops} we
prove some additional results about the loop operation. We show that
the alternative definition of a loop from \secref{sec:network.loops}
is equivalent to \defref{def:loop}. This simultaneously proves that
\defref{def:loop} always yields normalized maps.
In \appendixref{app:alternative} we give some alternative
distinguisher definitions, and show that they result in the same
pseudo\-/metric as in \secref{sec:metric}. In particular, we show how
this distance can be defined with subnormalized
distinguishers. Finally, \appendixref{app:lemmas} contains technical
lemmas.

\section{Fock space}
\label{app:fock}

As introduced in \secref{sec:space.wires}, for a Hilbert
space $\hilbert$, the corresponding bosonic Fock space is given by
\begin{equation} \label{eq:generalfockspace} \fock{\hilbert} \coloneqq
  \bigoplus_{n = 0}^\infty \vee^n \hilbert\,, \end{equation} where
$\vee^n \hilbert$ denotes the symmetric subspace of
$\hilbert^{\tensor n}$, and $\hilbert^{\tensor 0}$ is the one
dimensional space containing the vacuum state $\vacuum$. Let
$\{\ket{i} : i \in \cB\}$ denote a basis of $\hilbert$, where $\cB$ is
a strict totally ordered set. For
\begin{equation} \label{eq:symspacebasis1} \cB^{\vee n} \coloneqq
  \left\{ \left(i_1,\dotsc,i_n\right) : i_j \in \cB, j < k \implies
    i_j \leq i_k\right\}\,,\end{equation} a basis of $\vee^n \hilbert$
is then given by $\{\ket{(i_1,\dotsc,i_n)} : (i_1,\dotsc,i_n) \in \cB^{\vee n}\}$ where
$\ket{(i_1,\dotsc,i_n)}$ denotes the superposition over all
different states obtained by permuting the positions in the vector:
for $(i_1,\dotsc,i_n) \in \cB^{\vee n}$, let
\[\cB^{\times n}_{\left(i_1,\dotsc,i_n\right)} \coloneqq \bigcup_{\pi
  \in S_n} \{(i_{\pi(1)},\dotsc,i_{\pi(n)})\}\]
denote the subset of $\cB^{\times n}$ obtained by permuting the
elements of $\left(i_1,\dotsc,i_n\right)$, then
\begin{equation} \label{eq:symspacebasis2} \ket{(i_1,\dotsc,i_n)} \coloneqq
  \frac{1}{\sqrt{\left|\cB^{\times n}_{(i_1,\dotsc,i_n)}\right|}} \sum_{(j_1,\dotsc,j_n) \in
    \cB^{\times n}_{(i_1,\dotsc,i_n)}} \ket{j_1,\dotsc,j_n}\,.\end{equation} And a basis of
$\fock{\hilbert}$ is given by
\[\left\{\ket{x} : x \in \bigcup_{n = 0}^\infty \cB^{\vee n}\right\}\,.\]

For a wire $A$ with Hilbert space $\cF^{\T}_A = \fock{\Ltwo{\T}{\C^{d_A}}}$, a
state $\ket{\Psi} \in \cF^{\T}_A$ can thus be in a superposition
consisting of any number of qudits, e.g.,
\[\ket{\Psi} = \sum_{n = 0}^\infty \alpha_n \ket{\Psi_n}\,,\] where
$\ket{\Psi_n} \in \vee^n \left(\Ltwo{\T}{\C^{d_A}}\right)$ is an
$n$-qudit (symmetric) state with position information (where
$\ket{\Psi_0} = \vacuum$). An orthonormal basis of $\cF^\T_A$ is given
by the union of a basis for each
$\vee^n \left(\Ltwo{\T}{\C^{d_A}}\right)$, and an orthonormal basis
for $\vee^n \left(\Ltwo{\T}{\C^{d_A}}\right)$ is given by all
multisets\footnote{Due to the restriction to a symmetric subspace,
  only sets with repetition (multisets) not sequences of basis
  elements of $\Ltwo{\T}{\C^{d_A}}$ are relevant.} of cardinality $n$
of elements from $\{ \ket{v_t} : v \in \cV, t \in \T\}$.

\begin{rem}\label{rem:fockisomorphism}
Throughout this paper we use \eqnref{eq:fockisomorphism} to split a
wire into the tensor product of two wires, or to merge the tensor
product of two wires into one. Here we explicitly provide the
isomorphism used in \eqnref{eq:fockisomorphism}.

Let $\{\ket{x} : x \in \cA^{\vee n}\}$ and
$\{\ket{x} : x \in \cB^{\vee n}\}$ denote bases of $\vee^n\hilbert_A$
and $\vee^n\hilbert_B$, respectively, as defined in
\eqnsref{eq:symspacebasis1} and \eqref{eq:symspacebasis2}. A basis of
$\vee^n\left( \hilbert_A \oplus \hilbert_B \right)$ is given by
$\{\ket{x} : x \in \left(\cA \oplus \cB \right)^{\vee n} \}$
for \begin{multline*}\left(\cA \oplus \cB \right)^{\vee n} = \left\{
    \left(i_1,\dotsc,i_{n_1},j_{1},\dotsc,j_{n_2}\right) : n_1 +n_2 =
    n, i_k \in \cA, j_k \in \cB, \right. \\ \left. k < \ell \implies
    i_k \leq i_\ell \text{ and } j_k \leq j_\ell
  \right\}\,. \end{multline*} We define the isomorphism between
$\fock{\hilbert_A} \tensor \fock{\hilbert_B}$ and
$ \fock{\hilbert_A \oplus \hilbert_B}$ by setting
  \[\ket{\left(i_1,\dotsc,i_{n_1}\right)} \tensor
  \ket{\left(j_1,\dotsc,j_{n_2}\right)} \cong
  \ket{\left(i_1,\dotsc,i_{n_1},j_1,\dotsc,j_{n_2}\right)}\,.\]
\end{rem}

\section{Infinite dimensional representations of quantum operators}
\label{app:rep}

\subsection{The Choi-Jamio\l{}kowski representation}
\label{app:rep.cj}

Let $\hilbert_A$ and $\hilbert_B$ be finite dimensional Hilbert
spaces, and let $\Phi : \lo{A} \to \lo{B}$ be a CPTP map, where
$\lo{}$ is the set of linear operators on $\hilbert$. The \cjrep of
$\Phi$ is given by the operator $R_\Phi \in \lo{BA}$ \--- the
\emph{Choi operator} \--- defined as
\begin{equation} \label{eq:finitecj} R_\Phi = \sum_{i,j}
  \Phi(\ketbra{i}{j}) \tensor \ketbra{i}{j}\,.\end{equation} One can
think of $R_\Phi$ as capturing the image of a basis of $\lo{A}$,
namely $\{\ketbra{i}{j}\}_{i,j}$.  $R_\Phi$ is positive
semi\-/definite and satisfies \begin{equation} \label{eq:cjtp} \tr_B
  R_\Phi = I_A\,.\end{equation} In fact, any positive semi\-/definite
operator satisfying \eqnref{eq:cjtp} is the \cjrep of some CPTP
map~\cite{Wat16}.

For infinite dimensional spaces $\hilbert_A$ and $\hilbert_B$, the
Choi operator of a map
$\Phi : \tcop{\hilbert_A} \to \tcop{\hilbert_B}$ can be
unbounded. Instead, the \cjrep is defined as the sesquilinear positive
semi\-/definite form\footnote{$R(\cdot;\cdot)$ is a sesquilinear form
  if it is antilinear in the first argument and linear in the second.}
$R_\Phi$ on
\[ \hilbert_B \times \hilbert_A = \operatorname{span}\left\{\psi_B
  \tensor \psi_A : \psi_B \in \hilbert_B, \psi_A \in \hilbert_A
\right\}\,,\] satisfying
  \begin{equation} \label{eq:cjform} R_\Phi(\psi_B \tensor \psi_A ;
    \varphi_B \tensor \varphi_A) \coloneqq \bra{\psi_B} \Phi \left(
      \ketbra{\bar{\psi}_A}{\bar{\varphi}_A} \right)
    \ket{\varphi_B}\,, \end{equation} where $\ket{\bar{\psi}} = \sum_{i
    = 1}^{\infty}\ket{i} \overline{\braket{i}{\psi}}$ for some fixed
  basis $\{\ket{i}\}_i$ of $\hilbert_A$~\cite{Hol11}.

  If the domain of the sesquilinear form $R_\Phi$ is the whole of
  $\hilbert_B \otimes \hilbert_A$, then the corresponding operator is
  bounded\footnote{We denote the set of bounded operators on
    $\hilbert$ by $\bdop{\hilbert}$. $U \in \bdop{\hilbert}$ if there
    exists some $c > 0$ such that for all $\psi \in \hilbert$,
    $\norm{U\psi}/\norm{\psi} < c$.} (as in the case of finite
  spaces), and can be recovered as the operator
  $\hat{R}_\Phi\in \bdop{\hilbert_{BA}}$
  satisfying \begin{equation} \label{eq:boundedcj} \bra{\psi_B}
    \tensor \bra{\psi_A} \hat{R}_\Phi \ket{\varphi_B} \tensor
    \ket{\varphi_A} = R_\Phi(\psi_B \tensor \psi_A ; \varphi_B \tensor
    \varphi_A)\,.\end{equation}

\eqnref{eq:cjtp} can be rewritten
as \begin{equation} \label{eq:infcjtp} \sum_j R_\Phi(j_B \tensor
  \psi_A ; j_B \tensor \varphi_A) =
  \braket{\psi_A}{\varphi_A} \end{equation} for any basis
$\{\ket{j}\}_j$ of $\hilbert_B$, and holds in the infinite dimensional
case as well. The converse also holds in the infinite dimensional
case: any positive semi\-/definite sesquilinear form satisfying
\eqnref{eq:infcjtp} uniquely defines a CPTP map $\Phi :
\tcop{\hilbert_A} \to \tcop{\hilbert_B}$~\cite{Hol11}.

If the map $\Phi : \tcop{\hilbert_A} \to \tcop{\hilbert_B}$ is
completely positive, but not trace\-/preserving, then the \cjrep of
$\Phi$ is still defined as the sesquilinear positive semi\-/definite
from $R_\Phi$ given by \eqnref{eq:cjform}. But this form does not
satisfy \eqnref{eq:infcjtp}. Similarly, in the finite case the Choi
operator of a non\-/trace\-/preserving map is still given by the
positive operator from \eqnref{eq:finitecj}, but it does not satisfy
\eqnref{eq:cjtp}.

\subsection{The \natrep}
\label{app:rep.natural}

Let $\hilbert_A$ and $\hilbert_B$ be finite dimensional Hilbert
spaces. A linear map $\Phi : \lo{A} \to \lo{B}$ can be represented as
a linear map
$K_\Phi : \hilbert_A \tensor \hilbert_{\bar{A}} \to \hilbert_B \tensor
\hilbert_{\bar{B}}$ defined by
\begin{equation}
  \label{eq:natural.rep}
  K_\Phi \ket{i}\ket{j} = \sum_{k,\ell}\alpha_{k\ell}
  \ket{k}\ket{\ell} \coloniff  \Phi\left(\ketbra{i}{j}\right) =
  \sum_{k,\ell}\alpha_{k\ell} \ketbra{k}{\ell}
\end{equation} for some bases $\{\ket{i_A}\}_i$ and $\{\ket{k_B}\}_k$ of
$\hilbert_A$ and $\hilbert_B$. The transformation $\Phi \mapsto K_\Phi$ is a
bijection, and is referred to as the \emph{\natrep}~\cite{Wat16}. 

In the infinite dimensional case, one may define the \natrep of a map
$\Phi : \tcop{\hilbert_A} \to \tcop{\hilbert_B}$ in the same way. Note
however that \eqnref{eq:natural.rep} is not well\-/defined on the
entire space of linear maps
$\Phi : \bdop{\hilbert_A} \to \bdop{\hilbert_B}$. The inner product on
$\hilbert_A \tensor \hilbert_{\bar{A}}$ is isomorphic to the
Hilbert\-/Schmidt inner product\footnote{The Hilbert\-/Schmidt inner
  product is defined as
  $\langle A,B \rangle \coloneqq \trace{\hconj{A}B}$.}  of operators
on $\hilbert_A$. The transformation $\Phi \mapsto K_\Phi$ given by
\eqnref{eq:natural.rep} is thus a bijection between the set of maps
$\fS(\hilbert_A) \to \fS(\hilbert_B)$ and the set of maps
$\hilbert_A \tensor \hilbert_{\bar{A}} \to \hilbert_B \tensor
\hilbert_{\bar{B}}$,
where $\fS(\hilbert)$ is the space of all operators on $\hilbert$ with
a bounded Hilbert\-/Schmidt norm. Since
$\tcop{\hilbert} \subseteq \fS(\hilbert)$, this representation is in
particular well\-/defined for linear transformations between trace
class operators, $\Phi : \tcop{\hilbert_A} \to \tcop{\hilbert_B}$.

The natural and \cjreps are related as follows.
\begin{align}
\bra{\psi_B}\bra{\varphi_{\bar{B}}} K_\Phi
  \ket{\psi_A}\ket{\varphi_{\bar{A}}} & = \bra{\psi_{B}}
  \Phi\left(\ketbra{\psi_A}{\bar{\varphi}_A}\right) \ket{\bar{\varphi}_B} \notag\\
 & = R_\Phi(\psi_B \tensor \bar{\psi}_A ; \bar{\varphi}_B \tensor \varphi_A ) \,,
\label{eq:rep.cj.nat}
\end{align}
where, as previously,
$\ket{\bar{\psi}} = \sum_{i = 1}^{\infty}\ket{i}
\overline{\braket{i}{\psi}}$ for some fixed basis $\{\ket{i}\}_i$.

We now prove that the formula for a loop from \eqnref{eq:loop} may be
equivalently written using the \natrep, as noted in
\remref{rem:nat.rep}.

\begin{lem}
\label{lem:nat.rep}
Let $\Phi : \tcop{\hilbert_{AC}} \to \tcop{\hilbert_{BC}}$ be a linear
map. Let $\Psi : \tcop{\hilbert_{A}} \to \tcop{\hilbert_{B}}$ be
another linear map such that for some bases $\{k_C\}_k$ and
$\{\ell_C\}_\ell$ of $\hilbert_C$,
\begin{multline*}
  R_\Psi (\psi_B \tensor \psi_A; \varphi_B \tensor \varphi_A) = \\
  \sum_{k,\ell} R_{\Phi}(\psi_B \tensor k_C \tensor
  \psi_A  \tensor \bar{k}_C ; \varphi_B \tensor \ell_C \tensor
  \varphi_A \tensor \bar{\ell}_C  ) \,.
\end{multline*}
Then \[K_\Psi = \trace[C\bar{C}]{K_\Phi}\,.\]
\end{lem}

\begin{proof}
From \eqnref{eq:rep.cj.nat} we have
\[ R_\Psi(\psi_B \tensor \psi_A ; \varphi_B \tensor \varphi_A ) =
\bra{\psi_B}\bra{\bar{\varphi}_{\bar{B}}} K_\Psi
\ket{\bar{\psi}_A}\ket{\varphi_{\bar{A}}}\,. \]
And 
\begin{align*}
  & \sum_{k,\ell} R_{\Phi}(\psi_B \tensor k_C \tensor
    \psi_A  \tensor \bar{k}_C ; \varphi_B \tensor \ell_C \tensor
    \varphi_A \tensor \bar{\ell}_C  ) \\
  & \qquad \qquad \qquad \qquad \qquad = \sum_{k,\ell} \bra{\psi_B}\bra{\bar{\varphi}_{\bar{B}}} \bra{k_C}
    \bra{\bar{\ell}_{\bar{C}}} K_\Phi \ket{\bar{\psi}_A}\ket{\varphi_{\bar{A}}} \ket{k_C}
    \ket{\bar{\ell}_{\bar{C}}} \\
  & \qquad \qquad \qquad \qquad \qquad = \bra{\psi_B}\bra{\bar{\varphi}_{\bar{B}}}
    \trace[C\bar{C}]{K_\Phi}
    \ket{\bar{\psi}_A}\ket{\varphi_{\bar{A}}}\,. \qedhere
\end{align*}
\end{proof}

\section{Finite causal boxes}
\label{app:finite}

In this work causal boxes are defined by a set of maps. As explained
in \secref{sec:boxes.set}, this allows systems to be included which
produce an unbounded number of messages and are thus not
well\-/defined as a single map on the entire set $\T$, but only on all
subsets $\T^{\leq t}$ for any $t \in \T$. In this section we define
the subset of causal boxes that are defined on the entire set $\T$ and
which are closed under composition. We call these \emph{finite} causal
boxes.

It is not sufficient to define finite causal boxes as those systems
that are captured by a map
$\Phi : \tcop{\cF^\T_X} \to \tcop{\cF^\T_Y}$ since this set is not
closed under composition. For example, consider a system with
$\T = \N_0$, one input wire and two output wires, which, for every
qubit received in position $t \in \N_0$, outputs a qubit on each wire
in position $t+1$. Furthermore, this system outputs a qubit on each
wire at position $t=0$. The system can be described by a map
$\Phi : \tcop{\cF^\T_X} \to \tcop{\cF^\T_Y}$ which outputs two
sequences of $n+1$ qubits for every sequence of $n$ qubits it receives
at positions shifted by $1$. If we now put a loop from one of the
output wires to the input wire, we get a system with one output wire
that produces a qubit at every position $t \in \N_0$. This new system
outputs an infinite sequence of qubits, which is well\-/defined for
every $t \in \T$, but not on the entire set $\T$.

The causality function (\defref{def:causality.function}) associated
with each causal box guarantees that every point $t' \in \T$ may be
reached from any other point $t \leq t'$ in a finite number of causal
steps. To obtain closure under composition for finite causal boxes, it
is necessary to limit causal boxes to a finite behavior, i.e., that
all of $\T$ may be reached in a finite number of causal steps from any
$t \in \T$. The set of points $\T \setminus \chi(\T)$ are not needed
to generate outputs, so one can think of causal boxes as terminating
\--- at least, not reading further inputs \--- when a point in
$\T \setminus \chi(\T)$ has been reached.

\begin{deff}[Finite causality function]
\label{def:finite.causality.function}
A causality function $\chi : \cut \to \cut$ is a \emph{finite
  causality function} if for every $t \in \T$ there exists an $n \in \N$ such that
\begin{equation} \label{eq:finite.causality.finite} t \notin
  \chi^n\left(\T\right) \,.\end{equation}
\end{deff}

We can now define a finite causal box.

\begin{deff}[Finite causal box]
  \label{def:finitequantumbox} 
  A \emph{$(d_X,d_Y)$\-/finite causal box} $\Phi$ is a system with input
  wire $X$ and output wire $Y$ of dimension $d_X$ and $d_Y$, defined
  by a completely positive,
  trace\-/preserving (CPTP) map
\[\Phi : \tcop{\cF^\T_X} \to \tcop{\cF^\T_Y}\] satisfying the following causality constraint: there must
exist a function $\chi : \cut \to \cut$ satisfying
\defref{def:finite.causality.function} such that for all
$\cC \in \cut$,
  \begin{equation}
    \label{eq:finitecausality} \Phi^\cC = \Phi^\cC \circ \tr_{\T \setminus \chi(\cC)}\,,
  \end{equation}
  where $\Phi^\cC \coloneqq {\tr_{\T \setminus \cC}} \circ \Phi$.
\end{deff}

\eqnref{eq:finitecausality} may be rewritten as 
\begin{equation}
  \label{eq:finitecombined} {\tr_{T \setminus \cC}} \circ \Phi = \Phi^\cC \circ \tr_{\T \setminus
    \chi(\cC)}\,,\end{equation} which is an exact replica of \eqnref{eq:combined} with
$\cD = \T$ and $\cC \in \cut$ instead of $\cC \in \bcut$.
The Stinespring and \cjreps developed in \secref{sec:representation}
for \eqnref{eq:combined} also apply to
\eqnref{eq:finitecombined}. In particular, the map $\Phi$ of a finite
causal box has a sequence representation
(\defref{def:sequencerepresentation}). It then follows that the proofs of
closure for causal boxes from \secref{sec:network} are also proofs of closure for
finite causal boxes.

The proof of composition order indepence from \secref{sec:network} and
the pseudo\-/metric definition from \secref{sec:metric} are valid for
all causal boxes and thus in particular for finite causal boxes. Note
that the causal boxes used to construct the distinguishers from
\defref{def:distinguisher} are actually finite causal boxes.

\section{Causality for total orders}
\label{app:total}

In \defref{def:causality.function} a causality function is defined to
guarantee that every point $t' \in \T$ can be reached from any point
$t \leq t'$ in a finite number of causal steps. We show in this
section that in the special case where $\T$ is totally ordered, this
condition can be reduced to requiring that the system is not
arbitrarily fast around any point.

More precisely, for a system $\Phi$, let
$\chi : \T \to \T \cup \{\bot\}$ be a monotone function such that the
output up to position $t$ can be computed from the input up to
position $\chi(t) < t$, where $\chi(t) = \bot$ means that the output
up to position $t$ can be computed without any inputs.\footnote{By
  defining $\hat{\chi} : \cut \to \cut$ as
  $\hat{\chi}(\cC) \coloneqq \bigcup_{t \in \cC} \T^{\leq \chi(t)}$
  one gets a function that trivially satisfies all requirements of
  \defref{def:causality.function} except
  \eqnref{eq:causality.finite}.} As illustrated in
\secref{sec:boxes.causality}, the condition $\chi(t) < t$ is not
sufficient to guarantee that systems are closed under composition: an
example was given in which an infinite number of causal steps were
performed before reaching the point $t = 1$. We exclude such systems
by additionally requiring that $\chi(t) < t$ must also hold in the
limit as $t \to t_0$, i.e.,
$\inf_{t > t_{0}} \chi(t) < \inf_{t > t_{0}} t$ and
$\sup_{t < t_{0}} \chi(t) < \sup_{t < t_{0}} t $.\footnote{In general
  there might not be any distance measure on $\T$, so these criteria
  are defined using infimum and supremum instead of a limit.}

In the example discussed in \secref{sec:boxes.causality}, the sequence
of output positions $\{0,1/2,3/4,7/8,15/16,\dotsc\}$ converges to $1$,
which is a point in $\T = \Q^+$. One could however give an example in
which the sequence converges to a point which is not in $\T$, e.g.,
$\sqrt{2}$. Hence we require that the systems are not arbitrarily fast
around all points $t_0$ in the completion of $\T$ (see
\defref{def:total.causality.function} here below and \cite{Sch03} for
the definition of a completion of an ordered set). The condition with
the supremum is actually redundant, so we omit it from the following
definition.


\begin{deff}[Causality function for totally ordered sets\footnote{In
    the case where $\T \subseteq \R$,
    \defref{def:total.causality.function} is equivalent to requiring
    that for every $u \in \T$ there exists a $\delta_u > 0$ such that
    for all $t \leq u$, $t-\chi(t) > \delta_u$.}]
\label{def:total.causality.function}
Let $\T$ be a totally ordered countable set. Let
$\fZ(\T) \supseteq \T$ be its smallest completion such that for every
subset $\cP \subseteq \T$, $\inf \cP \in \fZ(\T)$ and
$\sup \cP \in \fZ(\T)$.\footnote{This is called the Dedekind-MacNeille
  completion~\cite{Sch03}. For example,
  $\fZ(\N) = \N \cup \{+\infty\}$ and
  $\fZ(\Q) = \R \cup \{-\infty,+\infty\}$.} If $\T$ does not have a
minimum (maximum), we remove the infimum (supremum) from $\fZ(\T)$,
i.e., for $p \in \{\inf \T, \sup \T\}$ we define
\[\completion \coloneqq \begin{cases} \fZ(\T) \setminus
  \left\{p\right\} & \text{if $p \notin \T$}, \\ \fZ(\T) & \text{if
    $p \in \T$}.\end{cases} \]
Finally, let $\bot$ be a point defined such that $\bot < t_0$ for all
$t_0 \in \completion$. A montone function
$\chi : \T \to \T \cup \{\bot\}$ is a \emph{causality function for
  totally ordered sets} if
\begin{align*} 
  \forall t \in \T, & \quad \chi(t) < t\,, \\
  \forall t_0 \in \completion, & \quad \inf_{t > t_{0}} \chi(t) < \inf_{t > t_{0}} t\,.
\end{align*} 
\end{deff}

With the causality function defined as
$\chi : \T \to \T \cup \{\bot\}$ instead of $\chi : \cut \to \cut$,
\eqnref{eq:causality.finite} from \defref{def:causality.function} can
be rewritten as \begin{equation} \label{eq:total.causality.finite}
  \forall t,t' \in \T, \exists n \in \N, \quad \chi^n(t') <
  t\,.\end{equation}
We prove in \lemref{lem:decomposing} that
\defref{def:total.causality.function} implies that
\eqnref{eq:total.causality.finite} is satisfied.

\begin{lem}
\label{lem:decomposing}
Let $\chi : \T \to \T \cup \{\bot\}$ be a function satisfying
\defref{def:total.causality.function}. Then for all $t,t' \in \T$ there exists
an $n \in \N$ such that $\chi^n(t') \leq t$. \end{lem}

\begin{proof}
  By contradiction, assume that this is not the case for some
  $t,t' \in \T$ and consider the set $\cP = \{\chi^n(t')\}_{n \in \N}$.
  By the construction of $\completion$, the completion of $\T$,
  $ \inf \cP \in \completion$. Let $t_0 \coloneqq \inf \cP$. Because
  $\chi$ is monotone and $\chi(p) < p$, we must have that for all
  $p > t_0$, $\chi(p) > t_0$. On the other hand, the condition
  $\inf_{p > t_0} \chi(p) < \inf_{p > t_0} t$ implies that there
  exists $p > t_0$ such that $\chi(p) \leq t_0$.
\end{proof}

We now strengthen the causality definition from
\defref{def:total.causality.function} so that it captures sufficient
conditions to define finite causal boxes on totally ordered sets. As
in \appendixref{app:finite} we have to exclude unbounded behavior as
$t$ gets larger. We achieve this by requiring that the systems may not
be arbitrarily fast around the supremum of $\T$.

\begin{deff}[Finite causality function for totally ordered sets]
\label{def:finite.total.causality.function}
Let $\chi : \T \to \T \cup \{\bot\}$ be a function satisfying
\defref{def:total.causality.function}. We say that it is a
\emph{finite causality function for totally ordered sets} if
additionally,
\begin{equation*}
\sup_{t < t_{0}} \chi(t) < \sup_{t < t_{0}} t
\end{equation*} 
for $t_0 \coloneqq \sup \T$.
\end{deff}

With causality functions defined as $\chi : \T \to \T \cup \{\bot\}$,
the condition for finite causal boxes from
\eqnref{eq:finite.causality.finite} in
\defref{def:finite.causality.function} becomes
\begin{equation} \label{eq:total.finite.causality.finite} \forall t
  \in \T, \exists n
  \in \N, \forall t' \in \T, \quad \chi^n(t') < t\,.
\end{equation}
We now prove that this is satisfied by
\defref{def:finite.total.causality.function}.

\begin{lem}
\label{lem:finitedecomposing}
Let $\chi : \T \to \T \cup \{\bot\}$ be a function satisfying
\defref{def:finite.total.causality.function}. Then for every $t \in \T$
there exists an $n \in \N$ such that for all $t' \in \T$,
$\chi^n(t') < t$. \end{lem}

\begin{proof}
  If $\T$ has a maximum $t_{\max}$, the lemma follows trivially,
  because by \lemref{lem:decomposing} for all $t$ there exists an $n$
  such that $\chi^n(t_{\max}) < t$, and for all $t'$,
  $\chi^n(t') \leq \chi^n(t_{\max})$. In the case where $\T$ does not
  have a maximum, let $t_{\sup} = \sup \T$ and let
  $p = \sup_{t < t_{\sup}} \chi(t)$. From
  \defref{def:finite.total.causality.function} we have $p < t_{\sup}$,
  so there must exist $t_0 \in \T$ such that $t_0 \geq p$.
  Furthermore, for all $t \geq t_0$, $\chi(t) \leq p \leq t_0$.
  Finally, by \lemref{lem:decomposing} for all $t$ there exists an $n$
  such that $\chi^n(t_0) < t$, hence for all $t' \in \T$,
  $\chi^{n+1}(t') < t$.
\end{proof}

\section{Loops}
\label{app:loops}

In \secref{sec:network.loops} we proposed an alternative, constructive
defintion for a loop (\eqnref{eq:loop.limit}) as the limit of a
sequence of maps. Each element in this sequence loops a bit more of
the output back to the input, and the limit defines the complete
loop. For a sequence of operators, the limit is defined as follows.

\begin{deff}[Operator convergence]
  Let $\{U_i : \hilbert_X \to \hilbert_Y\}_{i=1}^\infty$ be a sequence
  of linear operators. We say that $U_i$ converges to $U$ (in the
  strong operator topology) 
  if and only if
  \[ \forall \psi \in \hilbert_X, \forall \eps > 0, \exists i_0,
  \forall i \geq i_0, \quad \norm{U_i \psi - U \psi} < \eps\,,\]
  where $\norm{\cdot}$ is the $2$\=/norm.
\end{deff}

Let $\Phi^{\cC_1} :
\tcop{\cF^{\cC_2}_{AB}} \to \tcop{\cF^{\cC_1}_{CD}}$ be a map with sequence
representation
\begin{equation} \label{eq:lemmas.loops.before}
 U^{\cC_1}_\Phi = \left( \prod_{i=1}^{n-1}I^{\cC_{i+1}}_{CD} \tensor V_i
  \tensor I^{\cC_2 \setminus \cC_{i+1}}_{AB}\right) \left(
  U^{\cC_{n}}_{\Phi} \tensor I^{\cC_2 \setminus \cC_{n+1}}_{AB} \right)
\,, \end{equation} which is illustrated in \figref{fig:loop1}.
We define the linear operator
\begin{multline} \label{eq:lemmas.loops.after} U^{\cC_1}_{\Psi_n}
  \coloneqq \left( \bra{\Omega}^{\cC_{n+1}}_C \tensor I^{\cC_1}_D
    \tensor I_{Q_1} \tensor I^{\T_1}_C \right) \left(
    \prod_{i=1}^{n-1}I^{\cC_{i+1}}_D \tensor V_i \tensor I^{\cC_2
      \setminus \cC_{i+1}}_A\right) \\ \left( U^{\cC_{n}}_{\Phi}
    \tensor I^{\cC_2 \setminus \cC_{n+1}}_A \right) \left(
    \vacuum^{\cC_{n+1}}_B \tensor I^{\cC_2}_A
  \right)\,, \end{multline} which is obtained from
\eqnref{eq:lemmas.loops.before} by looping the output on
$\cF^{\T_i}_C$ to the input on $\cF^{\T_i}_B$ for $i \leq n$,
inputting the vacuum state $\vacuum$ on the sub-wire
$\cF^{\cC_{n+1}}_B$ and projecting the sub-wire $\cF^{\cC_{n+1}}_C$ on
the vacuum state \--- the input $\vacuum^{\cC_{n+1}}_B$ and projection
on $\vacuum^{\cC_{n+1}}_C$ ensure that all operators
$\{U^{\cC_1}_{\Psi_n}\}_{n=1}^\infty$ have the same input and output
Hilbert spaces. This is depicted in \figref{fig:loop.app}, which is a
reproduction of \figref{fig:loop2} along with the extra projections on
$\vacuum^{\cC_{n+1}}_C$ and vacuum inputs $\vacuum^{\cC_{n+1}}_B$,
that had been omitted for simplicity. Let $U^{\cC_1}_{\Psi}$ be the
limit operator as more of the wire $C$ is looped back to the input
$B$, i.e.,
\[U^{\cC_1}_{\Psi} \coloneqq \lim_{n \to \infty} U^{\cC_1}_{\Psi_n} \,.\]
And define the map $\Psi^{\cC_1} : \tcop{\cF^{\cC_2}_{A}} \to \tcop{\cF^{\cC_1}_{D}}$ by
\begin{equation}\label{eq:loop.limit.app}
\Psi^{\cC_1}(\rho) \coloneqq \trace[Q_1C_1]{U^{\cC_1}_{\Psi} \rho
\hconj{\left(U^{\cC_1}_{\Psi}\right)}}\,,
\end{equation} where $\tr_{Q_1C_1}$ traces out the ancilla register
$Q_1$ as well as the sub-wire with Hilbert space $\cF^{\T_1}_C$. The main proposition that we prove in this section is that the map
$\Psi^{\cC_1}$ defined above is the map one obtains by applying the
definition of a loop (\defref{def:loop}) to $\Phi^{\cC_1}$.

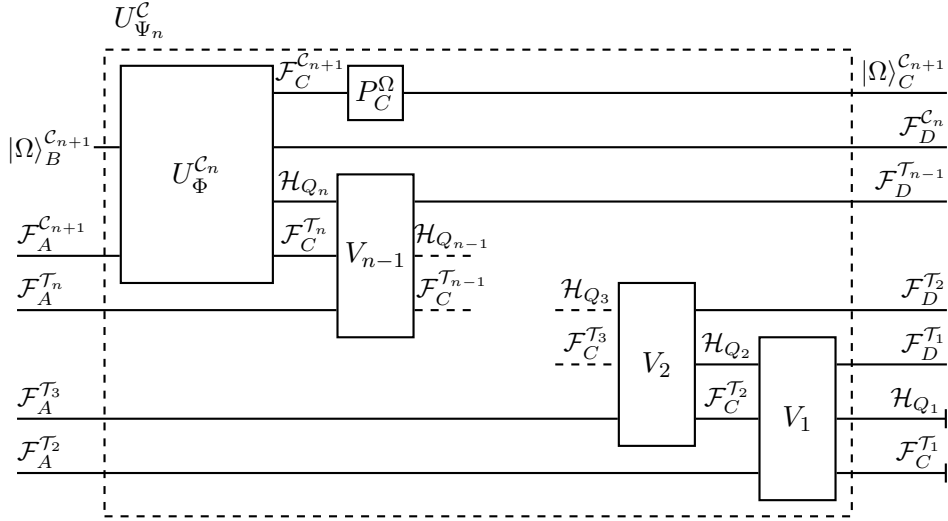
\begin{figure}[tb]
\begin{centering}
\begin{tikzpicture}[
wire/.style={-,thick},
unitary2/.style={draw,thick,inner sep=2pt,minimum width=1cm,minimum height=2.16cm},
invisible2/.style={thick,minimum width=1cm,minimum height=2.16cm},
unitary2a/.style={draw,thick,minimum width=2cm,minimum height=2.88cm},
unitary1/.style={draw,thick,inner sep=1pt,minimum width=.72cm,minimum height=.72cm}]

\def\t{.72}

\node (l1) at (0,0) {};
\node (l2) at (0,-\t) {};
\node[right,inner sep=0] (l2label) at (l2) {\small$\vacuum^{\cC_{n+1}}_B$};
\node (l3) at (0,-2*\t) {};
\node[above right,yshift=-2] at (l3) {};
\node (l4) at (0,-3*\t) {};
\node[above right,yshift=-2] at (l4) {\small$\cF^{\cC_{n+1}}_A$};
\node (l5) at (0,-4*\t) {};
\node[above right,yshift=-2] at (l5) {\small$\cF^{\T_n}_A$};
\node (l6) at (0,-5*\t) {};
\node (l7) at (0,-6*\t) {};
\node[above right,yshift=-2] at (l7) {\small$\cF^{\T_{3}}_A$};
\node (l8) at (0,-7*\t) {};
\node[above right,yshift=-2] at (l8) {\small$\cF^{\T_{2}}_A$};

\node (r1) at (12.5,0) {};
\node[above left,yshift=-2] at (r1) {\small$\vacuum^{\cC_{n+1}}_C$};
\node (r2) at (12.5,-\t) {};
\node[above left,yshift=-2] at (r2) {\small$\cF^{\cC_n}_D$};
\node (r3) at (12.5,-2*\t) {};
\node[above left,yshift=-2] at (r3) {\small$\cF^{\T_{n-1}}_D$};
\node (r4) at (12.5,-3*\t) {};
\node (r5) at (12.5,-4*\t) {};
\node[yshift=-2,above left] at (r5) {\small$\cF^{\T_{2}}_D$};
\node (r6) at (12.5,-5*\t) {};
\node[above left,yshift=-2] at (r6) {\small$\cF^{\T_{1}}_D$};
\node (r7) at (12.5,-6*\t) {};
\node[yshift=-2,above left,xshift=-2] at (r7) {\small$\hilbert_{Q_1}$};
\node (r8) at (12.5,-7*\t) {};
\node[yshift=-2,above left,xshift=-2] at (r8) {\small$\cF^{\T_1}_C$};

\node[unitary2a] (u1) at (2.5,-1.5*\t) {$U^{\cC_n}_\Phi$};
\node[unitary1] (p1) at (4.85,0) {$P^{\Omega}_C$};
\node[unitary2] (u2) at (4.85,-3*\t) {$V_{n-1}$};
\node[invisible2] (u3) at (6.7,-4*\t) {};
\node[unitary2] (u4) at (8.55,-5*\t) {$V_2$};
\node[unitary2] (u5) at (10.4,-6*\t) {$V_1$};
\node[fit=(u1)(u5),thick,dashed,draw,inner sep=.2cm] (u) {};
\node[above right] at (u.north west) {$U_{\Psi_n}^{\cC}$};

\draw[wire] (l2label) to (u1.west |- l2);
\draw[wire] (l4) to (u1.west |- l4);
\draw[wire] (l5) to (u2.west |- l5);
\draw[wire] (l7) to (u4.west |- l7);
\draw[wire] (l8) to (u5.west |- l8);

\draw[wire] (u1.east |- r3) to node[auto,yshift=-2] {\small$\hilbert_{Q_n}$} (u2.west |- r3);
\draw[wire] (u1.east |- r4) to node[auto,yshift=-2] {\small$\cF^{\T_n}_C$} (u2.west |- r4);
\draw[wire,dashed] (u2.east |- r4) to node[auto,yshift=-2,pos=.6] {\small$\hilbert_{Q_{n-1}}$} (u3.west |- r4);
\draw[wire,dashed] (u2.east |- r5) to node[auto,yshift=-2,pos=.6] {\small$\cF^{\T_{n-1}}_C$} (u3.west |- r5);
\draw[wire,dashed] (u3.east |- r5) to node[auto,yshift=-2] {\small$\hilbert_{Q_{3}}$} (u4.west |- r5);
\draw[wire,dashed] (u3.east |- r6) to node[auto,yshift=-2] {\small$\cF^{\T_{3}}_C$} (u4.west |- r6);
\draw[wire] (u4.east |- r6) to node[auto,yshift=-2] {\small$\hilbert_{Q_{2}}$} (u5.west |- r6);
\draw[wire] (u4.east |- r7) to node[auto,yshift=-2] {\small$\cF^{\T_{2}}_C$} (u5.west |- r7);

\draw[wire] (u1.east |- r1) to node[auto,yshift=-2] {\small$\cF^{\cC_{n+1}}_C$} (p1);
\draw[wire] (p1) to (r1);
\draw[wire] (u1.east |- r2) to (r2);
\draw[wire] (u2.east |- r3) to (r3);
\draw[wire] (u4.east |- r5) to (r5);
\draw[wire] (u5.east |- r6) to (r6);
\draw[wire,-|] (u5.east |- r7) to (r7);
\draw[wire,-|] (u5.east |- r8) to (r8);
\end{tikzpicture}

\end{centering}
\caption[A partial loop]{\label{fig:loop.app}The outputs on $C$ in
  positions $\T_i$ for $i \leq n$ are looped back to the inputs on
  $B$. The outputs on $C$ for $i > n$ are projected on the vacuum
  state \--- $P^\Omega_C$ denotes this projector \--- and a vacuum
  state is input on $B$ in the corresponding positions.}
\end{figure}

\begin{prop}
\label{prop:loopfromsequence}
The \cjrep of the map $\Psi^{\cC_1}$ from \eqnref{eq:loop.limit.app}
is given by \eqnref{eq:loop}. Furthermore, the terms in
$\eqnref{eq:loop}$ are absolutely convergent, and if $\Phi^{\cC_1}$ is
CPTP, then so is $\Psi^{\cC_1}$.
\end{prop}

We prove this proposition with the help of two lemmas. The first
simply states that the operators $U^{\cC_1}_{\Psi_n}$ converge.

\begin{lem}
  \label{lem:backwardconvergence}
  Let the operators $U^{\cC_1}_{\Psi_n}$ be defined as above. Then the
  limit operator
  $U^{\cC_1}_{\Psi} = \lim_{n \to \infty} U^{\cC_1}_{\Psi_n}$
  exists. Furthermore, if $U^{\cC_1}_\Phi$ is an isometry, then so is
  $U^{\cC_1}_\Psi$.
\end{lem}

\begin{proof}
  We prove the case in which $U^{\cC_1}_\Phi$ is an
  isometry.  The case of a linear operator $U^{\cC_1}_\Phi$ that is not an isometry
  follows by writing it as an isometry followed by a projection.

  Let $\{b^j_i\}_i$ be a basis of $\cF^{\T_j}_C$ and let $\Omega^n$ be
  shorthand for the vacuum state $\vacuum^{\cC_n}_C$. We build a basis
  for $\cF^{\cC_1}_{C}$ consisting of states that have a vacuum as
  prefix, i.e., all states of the form
  $c_i = \Omega^{n} \tensor \bigotimes_{j = 1}^{n-1} b^j_{i_j}$, where
  $b^j_{i_j}$ is a basis state of $\cF^{\T_j}_C$ and $n \in \N$. Any
  state $\varphi \in \cF^{\cC_1}_C$ can thus be written as
  $\varphi = \sum_i \alpha_i c_i$ for some coefficients
  $\alpha_i \in \C$. For a normalized state
  $\varphi \in \cF^{\cC_1}_C$ and any $\eps > 0$ one can always find a
  finite set of basis states indexed by $i \in \cS$ such that
  $\sum_{i \in \cS} |\alpha_i|^2 \geq 1 - \eps$. Let $n_0$ be such
  that for all $i \in \cS$, $c_i = \Omega^{n_i} \tensor c'_i$ with
  $n_i \leq n_0$.  This means that
  $\varphi = \Omega^{n_0} \tensor \varphi_{n_0-1} + \varphi'$, where
  $\varphi_{n_0-1} \in \bigotimes_{j = 1}^{n_0-1}\cF^{\T_j}_C$ and
  $\norm{\varphi'}^2 \leq \eps$.

  We fix $\eps > 0$ and $\psi \in \cF^{\cC_2}_A$. To prove that the
  operators $U^{\cC_1}_{\Psi_n}$ converge, we will show that there
  exists an $n_0$ such that for any $n,m \geq n_0$,
  \[\norm{U^{\cC_1}_{\Psi_n} \psi - U^{\cC_1}_{\Psi_m} \psi} \leq \sqrt{\eps}\,.\]
  Let
  $\varphi \coloneqq U^{\cC_1}_\Phi (\psi_A \tensor \vacuum^{\cC_2}_B)
  \in \cF^{\cC_1}_{CD}$
  and let $n_0 \in \N$ be such that
  $\varphi = \Omega^{n_0} \tensor \varphi_{n_0 -1} + \varphi'$, where
  $\varphi_{n_0 -1} \in \bigotimes_{j = 1}^{n_0-1}\cF^{\T_j}_C \tensor
  \cF^{\cC_1}_D$
  and $\norm{\varphi'}^2 \leq \eps$. Let $n \geq m \geq n_0$. The
  decompositions of $U^{\cC_1}_{\Psi_n}$ and $U^{\cC_1}_{\Psi_m}$ into
  sequences of operators as in \figref{fig:loop.app} both start with
  $U^{\cC_n}_{\Phi}$ and finish with the sequence $V_{m-1}$ to
  $V_1$. They differ only in their behavior on the positions
  $\cC_{m} \setminus \cC_n$. Here, $U^{\cC_1}_{\Psi_n}$ consists in
  the sequence $V_{n-1}$ to $V_m$ that are applied to the values on
  the $A$ and $C$ wires as well as the internal memory with Hilbert
  spaces $\hilbert_{Q_i}$, whereas the sequence decomposition of the
  operator $U^{\cC_1}_{\Psi_m}$ additionally contains projectors on
  the vacuum state and the operators $V_i$, $i \in \{m,\dotsc,n-1\}$,
  are only applied to input states that have a vacuum on the $B$
  wire. Because these are linear operators we have
  $U^{\cC_1}_{\Psi_m} \psi = \Omega^{m+1} \tensor \varphi_{m}$ and
  $U^{\cC_1}_{\Psi_n} \psi = \Omega^{m+1} \tensor \varphi_{m} +
  \varphi'$.
  And because $m \geq n_0$, $\varphi_{m} $ must have weight
  $\norm{\varphi_{m}}^2 \geq 1 -\eps$, hence
  $\norm{\varphi'}^2 \leq \eps$. Putting this together we get
  \begin{align*}
    \norm{U^{\cC_1}_{\Psi_n} \psi - U^{\cC_1}_{\Psi_m} \psi}
    & = \norm{\Omega^{m+1} \tensor \varphi_{m} + \varphi'
       - \Omega^{m+1} \tensor \varphi_{m}} \\
    & = \norm{\varphi'} \leq \sqrt{\eps}\,.
  \end{align*}

  It is immediate that the limit operator $U^{\cC_1}_\Psi$ is an
  isometry, since for any $\psi$ with $\norm{\psi} = 1$ and any $\eps$
  there exists $n_0$ such that for any $n \geq n_0$,
  $\norm{U^{\cC_1}_{\Psi_{n}} \psi} \geq 1 - \eps$.
\end{proof}

We now need to relate the maps $\Psi^{\cC_1}_n$ which have the \streps
$U^{\cC_1}_{\Psi_n}$ to their \cjreps. In the following lemma we do
this first for a map
$\Phi : \tcop{\hilbert_{AB}} \to \tcop{\hilbert_{CD}}$ for which the
output on $C$ does not depend on the input on $B$. This special case
was illustrated in \figref{fig:connecting}, where $U_{\Phi}$ is
depicted on the left and the resulting system after connecting $C$ to
$B$ is drawn on the right.

\begin{lem}
  \label{lem:connecting}
  Let $\Phi : \tcop{\hilbert_{AB}} \to \tcop{\hilbert_{CD}}$ be a CP
  map with a \strep given by
  \[ U_\Phi = \left(I_C \tensor V\right) \left(U_\Psi \tensor
    I_B\right)\,,\]
  where $U_\Psi : \tcop{\hilbert_{A}} \to \tcop{\hilbert_{CQ}}$ and
  $V : \tcop{\hilbert_{QB}} \to \tcop{\hilbert_{DR}}$. Let
  $\{\ket{k_C}\}_k$ and $\{\ket{\ell_C}\}_\ell$ be any orthonormal
  bases of $\hilbert_C$ and let $\{\ket{k_B}\}_k$ and
  $\{\ket{\ell_B}\}_\ell$ denote the corresponding bases of
  $\hilbert_B$, i.e., for all $k$ and $\ell$,
  $\ket{k_C} \cong \ket{k_B}$ and $\ket{\ell_C} \cong \ket{\ell_B}$.
  The map $\Phi': \tcop{\hilbert_{A}} \to \tcop{\hilbert_{D}}$
  obtained by plugging $C$ into $B$ has
  \cjrep \begin{multline} \label{eq:loop.specialcase}
    R_{\Phi'}(\psi_{D} \tensor \psi_{A} ; \varphi_{D} \tensor
    \varphi_{A}) = \\ \sum_{k,\ell} R_{\Phi}(k_C \tensor \psi_{D}
    \tensor \psi_{A} \tensor \bar{k}_B ; \ell_C \tensor \varphi_{D}
    \tensor \varphi_{A} \tensor \bar{\ell}_{B})\,, \end{multline}
  where
  $\ket{\bar{k}_B} = \sum_{i = 1}^{\infty}\ket{i}
  \overline{\braket{i}{k}}$
  for the basis $\{\ket{i_B}\}_i$ of $\hilbert_B$ used in the \cjrep
  of $\Phi$. Furthermore, the terms in this sum are absolutely
  convergent.
\end{lem}

\begin{proof}
  Define $\bar{\Psi} : \tcop{\hilbert_{AB}} \to \tcop{\hilbert_{CQB}}$
  and $\bar{\Gamma} : \tcop{\hilbert_{CQB}} \to \tcop{\hilbert_{CD}}$
  as \begin{align*} \bar{\Psi}(\rho) & = \left(U_\Psi \tensor I_B
    \right) \rho \left(\hconj{U}_\Psi \tensor I_B \right)\,, \\
    \bar{\Gamma}(\rho) & = \ktrace[R]{\left( I_C \tensor V
      \right) \rho \left( I_C \tensor \hconj{V} \right)}\,,
\end{align*}
and $\bar{\Psi}' : \tcop{\hilbert_{A}} \to \tcop{\hilbert_{CQ}}$ and
$\bar{\Gamma}' : \tcop{\hilbert_{QB}} \to \tcop{\hilbert_{D}}$
as \begin{align*} \bar{\Psi}'(\rho) & = U_\Psi \rho
  \hconj{U}_\Psi, \\
  \bar{\Gamma}'(\rho) & = \trace[R]{V \rho \hconj{V}}\,.
\end{align*}
Then $\Phi = \bar{\Gamma} \circ \bar{\Psi}$ and $\Phi' = \bar{\Gamma}'
\circ \bar{\Psi}'$. Writing up the \cjreps of these two maps we get
for any basis $\{\ket{i_Q}\}_i$ of $Q$,
\begin{multline*}
  R_\Phi(k_C \tensor \psi_D \tensor \psi_A \tensor \bar{k}_B ; \ell_C
  \tensor \varphi_D \tensor \varphi_A \tensor \bar{\ell}_B) \\
  \begin{split}& = \bra{k_C,\psi_D} \bar{\Gamma}
    \left(\bar{\Psi}\left(
        \ketbra{\bar{\psi}_A,k_B}{\bar{\varphi}_A,\ell_B} \right)
    \right) \ket{\ell_C,\varphi_D} \\
    & = \sum_{e,h,i,j,m,n} \bra{k_C,\psi_D} \bar{\Gamma} \left(
      \ketbra{e_C,i_Q,m_B}{h_C,j_Q,n_B} \right)
    \ket{\ell_C,\varphi_D} \\ & \qquad \qquad \qquad
    \qquad \quad \bra{e_C,i_Q,m_B} \bar{\Psi}\left(
      \ketbra{\bar{\psi}_A,k_B}{\bar{\varphi}_A,\ell_B} \right)
    \ket{h_C,j_Q,n_B} \\
    & = \sum_{i,j} \bra{k_C,\psi_D} \bar{\Gamma} \left(
      \ketbra{k_C,i_Q,k_B}{\ell_C,j_Q,\ell_B} \right)
    \ket{\ell_C,\varphi_D} \\ & \qquad \qquad \qquad
    \qquad \quad \bra{k_C,i_Q,k_B} \bar{\Psi}\left(
      \ketbra{\bar{\psi}_A,k_B}{\bar{\varphi}_A,\ell_B} \right)
    \ket{\ell_C,j_Q,k_B} \\
    & = \sum_{i,j} \bra{\psi_D} \bar{\Gamma}' \left(
      \ketbra{i_Q,k_B}{j_Q,\ell_B} \right) \ket{\varphi_D}
    \bra{k_C,i_Q} \bar{\Psi}'\left(
      \ketbra{\bar{\psi}_A}{\bar{\varphi}_A} \right)
    \ket{\ell_C,j_Q}\,, \end{split}
\end{multline*} 
where we have used that $\bar{\Psi}$ and $\bar{\Gamma}$ perform
identity on $B$ and $C$, respectively. And for $\Phi'$ we obtain.
\begin{multline*}
  R_{\Phi'}(\psi_D \tensor \psi_A ; \varphi_D \tensor \varphi_A) \\
  \begin{split} & = \bra{\psi_D} \bar{\Gamma}' \left(\bar{\Psi}'\left(
      \ketbra{\bar{\psi}_A}{\bar{\varphi}_A} \right) \right)
  \ket{\varphi_D} \\
  & = \sum_{i,j,k,\ell} \bra{\psi_D} \bar{\Gamma}' \left(
    \ketbra{i_Q,k_B}{j_Q,\ell_B} \right) \ket{\varphi_D} \bra{k_C,i_Q}
  \bar{\Psi}'\left( \ketbra{\bar{\psi}_A}{\bar{\varphi}_A} \right)
  \ket{\ell_C,j_Q} \\
  & = \sum_{k,\ell} R_{\Phi}(k_C \tensor \psi_{D} \tensor \psi_{A}
  \tensor \bar{k}_B ; \ell_C \tensor \varphi_{D} \tensor \varphi_{A}
  \tensor \bar{\ell}_{B})\,.\end{split}
\end{multline*}

Note that this holds for any orthonormal bases $\{\ket{k_C}\}_k$ and
$\{\ket{\ell_C}\}_\ell$ of $\hilbert_C$. Hence the order in the sum is
irrelevant and the terms are absolutely convergent.
\end{proof}

With these two lemmas, we can now prove
\propref{prop:loopfromsequence}. We first use \lemref{lem:connecting}
to show that the \cjrep of the system after the loop, namely
$\Psi^{\cC}$, is given by \eqnref{eq:loop}. Since in
\lemref{lem:backwardconvergence} we prove that the limit operator
$U^\cC_\Psi$ is an isometry, then $\Psi^{\cC}$ must be CPTP.

\begin{proof}[Proof of \propref{prop:loopfromsequence}]
  Let $R^{\cC_1}_{\Phi}(\cdot;\cdot)$ be the \cjrep of $\Phi^{\cC_1}$. We
  derive the \cjrep of $\Psi^{\cC_1}_n$ by repeating $n$ times the
  proof from  \lemref{lem:connecting}, i.e., connecting $\cF^{\T_1}_C$ to
  $\cF^{\T_1}_B$, $\cF^{\T_2}_C$ to $\cF^{\T_2}_B$, etc. This yields
  \begin{multline*} R^{\cC_1}_{\Psi_n}(\psi_D \tensor \psi_A ; \varphi_D
    \tensor \varphi_A) = \\
    = \sum_{k,\ell}
    R^{\cC_1}_{\Phi}\Bigl(\Omega^{\cC_{n+1}}_C \tensor k_C \tensor \psi_D
    \tensor \psi_A \tensor \bar{\Omega}^{\cC_{n+1}}_B \tensor
      \bar{k}_B ; \\
    \Omega^{\cC_{n+1}}_C \tensor \ell_C \tensor \varphi_D \tensor
    \varphi_A \tensor \bar{\Omega}^{\cC_{n+1}}_B \tensor \bar{\ell}_B
    \Bigr)\,,
  \end{multline*} 
  where $\{k_C\}_k$ and $\{\ell_C\}_\ell$ are bases of
  $\bigotimes_{i = 1}^n \cF^{\T_i}_C = \cF^{\cC_1 \setminus
    \cC_{n+1}}_C$.
  As in \lemref{lem:connecting}, the order of the summation is not
  relevant and so the terms are absolutely convergent. Using the fact
  that $\bigcup_{i=1}^\infty \T_i = \cC_1$, we immediately get the
  limiting case,
\begin{multline*} R^{\cC_1}_{\Psi}(\psi_D \tensor \psi_A ; \varphi_D
    \tensor \varphi_A) \\ \begin{split} & = \lim_{n \to \infty} R^{\cC_1}_{\Psi_n}(\psi_D \tensor \psi_A ; \varphi_D
    \tensor \varphi_A) \\
    & = \sum_{k,\ell} R^{\cC_1}_{\Phi}\left(k_C \tensor \psi_D
    \tensor \psi_A \tensor \bar{k}_B ; \ell_C \tensor \varphi_D \tensor
    \varphi_A \tensor \bar{\ell}_B \right) \,, \end{split}
\end{multline*}
where $\{k_C\}_k$ and $\{\ell_C\}_k$ are bases of $\cF^{\cC_1}_C$, and
as above, the terms are absolutely convergent. Finally, from
\lemref{lem:backwardconvergence} we know that $U^{\cC_1}_{\Psi}$ is an
isometry if $U^{\cC_1}_\Phi$ is an isometry, hence $\Psi^{\cC_1}$ is
CPTP if $\Phi^{\cC_1}$ is CPTP.
\end{proof}

\section{Alternative distinguishers}
\label{app:alternative}

In this section we consider two changes that could be made to the
notion of distinguishers and the corresponding pseudo\-/metric
introduced in
\secref{sec:metric}. In \appendixref{app:alternative.subnormalized} we
define distinguishers as subnormalized causal boxes: instead of a
normalized system which outputs either $0$ or $1$, we consider the
subnormalized system resulting from conditioning the output on
$0$. This results in a simplification of the pseudo\-/metric
definition. The second change discussed
in \appendixref{app:alternative.limit} consists in removing the
constraint that a distinguisher $\aD$ terminates at a fixed point
$t_{\aD}$. Instead we consider a possibly infinite sequence of bounded
cuts $\cC_1 \subseteq \cC_2 \subseteq \dotsb$, and define the
distinguisher's output to be the limit value of the output on $\cC_i$
as $i \to \infty$.  We then prove that the resulting pseudo\-/metric
is equivalent to the one from \secref{sec:metric}.

\subsection{Subnormalized distinguishers}
\label{app:alternative.subnormalized}

Let $\hat{\fD} = \{\hat{\aD}\}$ be a set of distinguishers. Since by
construction we define the output value of a distinguisher $\hat{\aD}$
to be computed at some point $t_{\hat{\aD}}$, the behavior of
$\hat{\aD}$ on all points $p \nleq t_{\hat{\aD}}$ is irrelevant, and
one can assume that $\hat{\aD}$ is entirely described by the map
$\hat{\aD}^{\leq t_{\hat{\aD}}}$. For any $\cC \in \bcut$ one can take
\begin{equation} \label{eq:distinguisher.completion} \hat{\aD}^\cC(\rho)
  \coloneqq \trace[\T\setminus\cC]{\hat{\aD}^{\leq t_{\hat{\aD}}}(\rho) \tensor
    \proj{\Omega}^{\T \setminus \T^{\leq t_{\hat{\aD}}}}}\,.\end{equation}

Given such a set $\hat{\fD} = \{\hat{\aD_i}\}_i$, we define a set of
subnormalized distinguishers $\fD = \{\aD_i\}_i$ by projecting the
output of the distinguisher $\hat{\aD}_i$ on the vacuum state
$\vacuum$.  Like for $\hat{\aD}_i$, each $\aD_i$ is entirely described
by a subnormalized map $\aD^{\leq t_{\aD}}_i$, where
$t_{\aD} = t_{\hat{\aD}}$. Using the construction from
\eqnref{eq:distinguisher.completion}, any subnormalized map
$\aD^{\leq t_\aD} = \id^{\leq t_\aD} \tensor (\aD')^{\leq t_{\aD}}$
such that $(\aD')^{\leq t_{\aD}}$ satisfies causality is a valid
subnormalized distinguisher.

Connecting a subnormalized $(m,n)$\-/distinguisher $\aD$ to an
$(m,n)$\-/causal boxes $\Phi$ results in a system
$\aD\Phi = \{\aD^\cC\Phi^\cC\}_{\cC \in \bcut}$ that has no output or
input wires. It is thus a set of numbers, where
$\aD^{\leq t_{\aD}}\Phi^{\leq t_{\aD}} \in [0,1]$ corresponds to the
probability of $\hat{\aD}$ outputting a vacuum state, i.e.,
\[\aD^{\leq t_{\aD}}\Phi^{\leq t_{\aD}} = \Pr\left[\hat{\aD}[\Phi] = 0\right]\,.\]
Using the \natrep introduced in \remref{rem:nat.rep}, this probability
may be written as a trace, namely
\[\aD^{\leq t_{\aD}}\Phi^{\leq t_{\aD}} = \trace{K^{\leq t_{\aD}}_{\aD\Phi}}\,,\]
where $K^{\leq t_{\aD}}_{\aD\Phi}$ is the \natrep of
$\mathrm{SWAP}_{\aD\Phi} \circ \left(\aD^{\leq t_{\aD}} \tensor
  \Phi^{\leq t_{\aD}}\right)$,
where $\mathrm{SWAP}_{\aD\Phi}$ permutes the output wires so that they
are aligned with the corresponding input wires and get connected by
the trace operator $\tr$.

The statistical distance between binary random variables may be
written in terms of one of the outcomes:
\begin{multline*} \frac{1}{2} \sum_{x \in \{0,1\}} \left| \Pr
    \left[\hat{\aD}[\Phi] = x \right] - \Pr \left[\hat{\aD}[\Psi] = x \right]
  \right| \\ = \left| \Pr\left[\hat{\aD}[\Phi] = 0 \right] - \Pr
    \left[\hat{\aD}[\Psi] = 0 \right] \right|\,. \end{multline*} Hence the
distinguisher pseudo\-/metric from \defref{def:metric} may equivalently be written as
\[d^{\fD}(\Phi,\Psi) = \sup_{\aD \in \fD}\left| \trace{K^{\leq
      t_{\aD}}_{\aD\Phi}} -\trace{K^{\leq
      t_{\aD}}_{\aD\Psi}}\right|\,,\]
where $t_\aD$ is the position at which the output value of $\aD$ is
computed.

Putting this together we get an alternative definition of the distance
between causal boxes which is equivalent to \defref{def:metric}.

\begin{deff}[Distance] \label{def:alternative.metric} Let
  $\fD = \{\aD\}$ be a set of subnormalized distinguishers as
  described above. Let $\Phi$ and $\Psi$ be two $(m,n)$\-/causal
  boxes. The distance between $\Phi$ and $\Psi$ with respect to $\fD$
  is given by
  \[d^{\fD}(\Phi,\Psi) \coloneqq \sup_{\aD \in \fD}\left|
    \trace{K^{\leq t_{\aD}}_{\aD\Phi}}
    -\trace{K^{\leq t_{\aD}}_{\aD\Psi}}
  \right|\,.\]
\end{deff}

\subsection{The distinguishing limit}
\label{app:alternative.limit}

In \secref{sec:metric} we only consider the output of a distinguisher
$\aD$ up to some point $t_{\aD} \in \T$. In this section we
consider an alternative definition of a distinguisher that is not
constrained in such a way. Instead, we consider the output on a totally ordered
(possibly infinite) sequence of bounded cuts $\{\cC_i\}_{i \in \cI}$. The output of the
distinguisher is then taken to be the limit value as the sequence
progresses. For example, if the set $\T$ is totally ordered, then one
can choose $\cI = \T$ and $\cC_i = \T^{\leq i}$.

\begin{deff}[Limit distinguisher] \label{def:distinguisher.limit} Let
  $\cI$ be a totally ordered set with $\infty \coloneqq \sup \cI$ and
  let $\{\cC_i\}_{i \in \cI}$ be a sequence of bounded cuts such that
  $i \leq j \implies \cC_i \subseteq \cC_j$. A \emph{$(m,n)$\-/limit
    distinguisher}
  $\aD = \{{\id^\cC} \tensor \hat{\aD}^\cC\}_{\cC \in \bcut}$ consists
  of a $(\hat{n},\hat{m})$\-/causal box
  $\hat{\aD} = \{\hat{\aD}^\cC\}_{\cC \in \bcut}$ with
  $m+1-\hat{m} = n - \hat{n}$, a sequence of bounded cuts
  $\{\cC_i\}_{i \in \cI}$ as described above, and a specification of
  how the distinguisher is connected to an $(m,n)$\-/dimensional
  system \--- i.e., which input and output sub-wires are connected to
  $\hat{\aD}$ and which are directly connected by a loop.  For an
  $(m,n)$\-/distinguisher $\aD$ and an $(m,n)$\-/causal box $\Phi$,
  let $\aD\Phi$ denote the causal box with no input wire and a
  $1$\=/dimensional output wire resulting from connecting the systems
  as specified. We define $\aD^\cC[\Phi^\cC]$ to be the binary random
  variable on $\{0,1\}$ obtained by projecting the output of $\aD\Phi$
  within $\cC$ on $P^{\cC}_0 = \proj{\Omega}^{\cC}$ and
  $P^{\cC}_1 = I^{\cC} - \proj{\Omega}^{\cC}$. $\aD[\Phi]$ is then
  defined as the limit over $\{\cC_i\}$ as $i \to \infty$, namely
  \[\Pr\left[\aD[\Phi]=x\right] \coloneqq \lim_{i \to \infty}
  \Pr\left[\aD^{\cC_i}[\Phi^{\cC_i}] = x\right],\]
  which we also write
  $\aD[\Phi] = \lim_{i \to \infty} \aD^{\cC_i}[\Phi^{\cC_i}]$.
\end{deff}

Note that this limit is always well\-/defined, because for any $i <
j$,
\begin{align*} \Pr\left[\aD^{\cC_i}[\Phi^{\cC_i}] = 0\right] & \geq
  \Pr\left[\aD^{\cC_j}[\Phi^{\cC_j}] = 0\right] \\ \intertext{and}
  \Pr\left[\aD^{\cC_i}[\Phi^{\cC_i}] = 1\right] & \leq
  \Pr\left[\aD^{\cC_j}[\Phi^{\cC_j}] = 1\right]\,, \end{align*} but
both are bounded, by $0$ and $1$, respectively.

The sequence of bounded cuts $\{\cC_i\}_{i \in \cI}$ is essential in
\defref{def:distinguisher.limit}. An output at two points $t$ and $t'$
that have no common future \--- i.e., $\nexists t_0$ such that
$t_0 \geq t$ and $t_0 \geq t'$ \--- is not necessarily well\-/defined
on both points simultaneously. This sequence then tells us which is
the relevant point. Since different distinguishers might use different
cuts $\{\cC_i\}_{i \in \cI}$, the entire set $\T$ is still covered by
considering the supremum over sets of distinguishers.

\begin{deff}[Limit distance] \label{def:metric.limit} Given a set of
  $(m,n)$\-/limit distinguishers $\fD$, the limit distance between two
  $(m,n)$\-/causal boxes $\Phi$ and $\Psi$ is defined as
\[d^{\fD}(\Phi,\Psi) \coloneqq \sup_{\aD \in \fD}
\delta(\aD[\Phi],\aD[\Psi]),\] where $\delta(\cdot,\cdot)$ is the
statistical or total variation distance.
\end{deff}

As in \secref{sec:metric} one can show that this distance is a
pseudo\-/metric. It is also a metric if $\fD$ is the set of all
distinguishers. One can also rewrite the definition using
subnormalized boxes as
in \appendixref{app:alternative.subnormalized}. We do not write up
these proofs, since they are nearly identical to the case of
distinguishers from \defref{def:distinguisher}.

Since a distinguisher $\aD$ that takes a decision before some point
$t_\aD$ is a special case of a limit distinguisher with a single cut
$\cC_1 = \T^{\leq t_\aD}$, \defref{def:metric.limit} is more general
than \defref{def:metric}. We now show that the converse also holds: we
prove that for any set of limit distinguishers $\fD$ we can find a set
of distinguishers $\fE$ that is at least as good at distinguishing any
pair of boxes. For $i \in \cI$ and a limit distinguisher $\aD$
(satisfying \defref{def:distinguisher.limit}), we define a
distinguisher $\aD_i$ (satisfying \defref{def:distinguisher}) as
follows. $\aD_i$ behaves identically to $\aD$ within $\cC_i$, but
ignores all inputs and produces no outputs out of $\cC_i$, i.e.,
$\aD_i$ is defined by the map $\aD^{\cC_i}$. The output of $\aD_i$ is
then computed on $\T^{\leq t_i}$ for any $t_i \in \T$ such that
$\cC_i \subseteq \T^{\leq t_i}$.

Let $\fD$ be any set of distinguishers, and let
\[\fE = \{\aD_i : \aD \in \fD, i \in \cI\}\]
be the set obtained as described above. We now prove that $\fE$ is at
least as good as $\fD$ at distinguishing causal boxes.

\begin{lem}
\label{lem:metric.limit}
Let $\fD$ and $\fE$ be defined as above. Then for any $\Phi$ and
$\Psi$, \[ d^{\fD}(\Phi,\Psi) \leq d^{\fE}(\Phi,\Psi). \]
\end{lem}

\begin{proof}
  For any $\aD \in \fD$ one has
\begin{align*}
\delta(\aD[\Phi],\aD[\Psi]) & = \lim_{i \to \infty}
\delta(\aD^{\cC_i}[\Phi^{\cC_i}],\aD^{\cC_i}[\Psi^{\cC_i}]) \\
& \leq \sup_i
  \delta(\aD^{\cC_i}[\Phi^{\cC_i}],\aD^{\cC_i}[\Psi^{\cC_i}]) \\ & =
  \sup_{i}
  \delta(\aD_i[\Phi],
  \aD_i[\Psi])\,, \end{align*} and
therefore \begin{align*} d^{\fD}(\Phi,\Psi) & = \sup_{\aD}
  \delta(\aD[\Phi],\aD[\Psi]) \\ & \leq
  \sup_{\aD,i}
  \delta(\aD_i[\Phi],
  \aD_i[\Psi]) \\ & =
  d^{\fE}(\Phi,\Psi)\,. \qedhere \end{align*}
\end{proof}

\begin{rem}
  We get equality in \lemref{lem:metric.limit} if $\fE \subseteq \fD$,
  i.e., if for every $\aD \in \fD$, the same distinguisher truncated
  at $\cC_i$ is also in $\fD$.
\end{rem}

\section{Technical lemmas}
\label{app:lemmas}

\paragraph{Parallel composition of causality functions.}
In \propref{prop:parallel} we prove that the parallel composition of
two causal boxes is a new valid box. To do this, we define the
causality function for the new box as the union of the causality
functions of both components. We prove in
\lemref{lem:causality.parallel} that this results in a valid causality
function.

\begin{lem}
  \label{lem:causality.parallel} Let $\chi_1$ and $\chi_2$ be two causality functions for
  some partially ordered set $\T$. Then the function defined as
  $\chi_{\text{new}}(\cC) \coloneqq \chi_1(\cC) \cup \chi_2(\cC)$ is also a
  causality function.
\end{lem}

\begin{proof}
  We need to prove that the four conditions from
  \defref{def:causality.function} hold. \eqnref{eq:causality.homomorphism} follows because
  \begin{align*}
    \chi_{\text{new}}(\cC \cup \cD) & = \chi_1(\cC \cup \cD) \cup
                                      \chi_2(\cC \cup \cD) \\
    & = \chi_1(\cC) \cup \chi_1(\cD) \cup \chi_2(\cC) \cup \chi_2(\cD) \\
    & = \chi_{\text{new}}(\cC) \cup \chi_{\text{new}}(\cD)\,.
  \end{align*}

 \eqnref{eq:causality.monotone} is satisfied because if $\cC
 \subseteq \cD$, then
  \begin{align*}
    \chi_{\text{new}}(\cC) & = \chi_1(\cC) \cup \chi_2(\cC) \\
    & \subseteq \chi_1(\cD) \cup \chi_2(\cD) \\
    & = \chi_{\text{new}}(\cD)\,.
  \end{align*}

  Since $\chi_1(\cC) \subsetneq \cC$ and $\chi_2(\cC) \subsetneq \cC$,
  we have that $\chi_{\text{new}}(\cC) \subseteq \cC$. To prove that
  \eqnref{eq:causality.decreasing} is satisfied, it remains to show
  that $\chi_{\text{new}}(\cC) \neq \cC$. This follows immediately if
  \eqnref{eq:causality.finite} holds for $\chi_{\text{new}}$, since
  otherwise for all $n$, $\chi^n_{\text{new}}(\cC) = \cC$ and we
  therefore have $t \in \chi^n_{\text{new}}(\cC)$ for all $n$.

  Hence it remains to prove that \eqnref{eq:causality.finite} is
  satisfied, i.e., we need to show that for every $\cC \in \bcut$ and
  $t \in \cC$, there exists an $n \in \N$ such that
  $t \notin \chi^n_{\text{new}}(\cC)$.  Since $\chi_1$ and $\chi_2$
  are valid causality functions, we know that there exists an $n_1$
  such that $t \notin \chi^{n_1}_1(\cC)$ and an $n_2$ such that
  $t \notin \chi^{n_2}_2(\cC)$. We now prove that
  $t \notin \chi^{n_1+n_2}_{\text{new}}(\cC)$. For $s \in \{1,2\}^n$
  we define
  $\chi_s(\cC) \coloneqq \chi_{s_1} \circ \dotsb \circ \chi_{s_n}
  (\cC)$.
  From \eqnref{eq:causality.homomorphism} it follows that
  $\chi^n_{\text{new}}(\cC) = \bigcup_{s \in \{1,2\}^n} \chi_s(\cC)$.
  Furthermore, combining \eqnsref{eq:causality.monotone} and
  \eqref{eq:causality.decreasing} we find that
  $\chi' \circ \chi \circ \chi' (\cC) \subseteq \chi \circ
  \chi' (\cC) \subseteq \chi(\cC)$.
  Thus, if at least $k$ bits of $s$ take the value $1$ ($2$), then
  $\chi_s(\cC) \subseteq \chi^k_1(\cC)$
  ($\chi_s(\cC) \subseteq \chi^k_1(\cC)$). So
  $t \notin \chi^{n_1+n_2}_{\text{new}}(\cC).$
\end{proof}

\paragraph{Commuting limits.} Swapping the order in which two loops
are applied to a quantum box corresponds to swapping the order of the
limits implicit in \eqnref{eq:loop}. We give here a simple condition
that allows the order of the limits to be swapped, which is sufficient
to prove in \lemref{lem:loop.commute} that loops commute. We define
the double limit as
\[ \lim_{m,n \to \infty} a_{mn} = L \coloniff \forall \eps > 0,
\exists N \in \N,
\forall m,n \geq N, \quad |a_{mn} - L | < \eps\,.\]

\begin{lem} \label{lem:limit.swap}
Let $a_{mn}$ be a sequence of complex numbers indexed by $m,n \in
\N$. If \begin{align*} \lim_{m \to \infty} \lim_{n \to \infty} a_{mn}
          & & \text{and} & & \lim_{m,n \to \infty} a_{mn} \end{align*}
both exist, then they converge to the same value.
\end{lem}

\begin{proof} Let $\lim_{m,n \to \infty} a_{mn} = L$ and
  $\lim_{n \to \infty} a_{mn} = L_m$. We want to prove that
  $\lim_{m \to \infty} L_m = L$. Fix $\eps > 0$ and let $M$ be such
  that $\forall m,n \geq M$, $|a_{mn} - L | < \eps$. For any
  $m \geq M$, pick $N_m \geq M$ such that $|a_{mN_m} - L_m| <
  \eps$. Hence, there exists $M$ such that for any $m \geq M$
\[ |L_m - L | \leq |L_m - a_{mN_m}| + |a_{mN_m} - L| < 2 \eps\,. \qedhere\]
\end{proof}

\providecommand{\bibhead}[1]{}
\expandafter\ifx\csname pdfbookmark\endcsname\relax%
  \providecommand{\tocrefpdfbookmark}{}
\else\providecommand{\tocrefpdfbookmark}{%
   \phantomsection%
   \addcontentsline{toc}{section}{\refname}}%
\fi

\tocrefpdfbookmark

\end{document}